\documentclass[11pt]{article}

\usepackage{preamble}
\usepackage{makecell}

\usepackage{geometry}
\newcommand\coeffPU{48}
\newcommand\coeffQU{72}

\title{
Algorithms for Discrepancy, Matchings and Approximations: Fast, Simple, and Practical
}

\author{M\'onika Csik\'os\footnote{Universit\'e Paris Cit\'e, IRIF, CNRS UMR 8243 and DI ENS, Universit\'e PSL. E-mail: csikos@irif.fr.} ~and Nabil H. Mustafa\footnote{Universit\'e Sorbonne Paris Nord, Laboratoire LIPN, CNRS 7030. E-mail: nabil.mustafa@univ-paris13.fr.}}

\date{}

\begin{document}

\maketitle

\begin{abstract}
\noindent
We study one of the key tools in data approximation and optimization: low-discrepancy colorings. Formally, given a finite set system $(X,\S)$, the \emph{discrepancy} of a two-coloring $\chi:X\to\{-1,1\}$ is defined as $\max_{S \in \S}\abs{\chi(S)}$, where $\chi(S)=\sum\limits_{x \in S}\chi(x)$.  \\ \\
We propose a randomized algorithm which, for any $d>0$ and $(X,\S)$ with dual shatter function $\pi^*(k)=O(k^d)$, returns a coloring with  expected discrepancy $O\round{\sqrt{|X|^{1-1/d}\log|\S|}}$ (this bound is tight) in time $\aO\round{|\S|\cdot|X|^{1/d}+|X|^{2+1/d}}$, improving upon the previous-best time of $O(|\S|\cdot|X|^3)$ by at least a factor of $|X|^{2-1/d}$ when $|\S|\geq|X|$.
This setup includes
many geometric classes, families of bounded dual VC-dimension, and others. As an immediate consequence, we obtain an improved algorithm to construct $\eps$-approximations of sub-quadratic size.\\ \\
Our method uses primal-dual reweighing with an improved analysis of randomly updated weights and exploits the structural properties of the set system via matchings with low crossing number---a fundamental structure in computational geometry. In particular, we get the same $|X|^{2-1/d}$ factor speed-up on the construction time of matchings with crossing number $O\round{|X|^{1-1/d}}$, which is the first improvement since the 1980s. \\ \\
The proposed algorithms are very simple, which makes it possible, for the first time, to compute colorings with near-optimal discrepancies and near-optimal sized approximations for abstract and geometric set systems in dimensions higher than $2$. 
 \\
\end{abstract}

\textbf{Keywords:} discrepancy, approximations,  low-crossing matchings, VC-dimension, MWU

\newpage

\tableofcontents

\newcommand\capx{C_{\mathrm{apx}}} 
\newcommand\dvc{\mathrm{d_{VC}}} 

\section{Introduction and Main Result}\label{sec:intro}

A \emph{set system} is a pair $(X,\S)$, where $X$ is a set and $\S$ is a collection of subsets of $X$. The elements of $\S$ are called \emph{ranges}. We consider finite set systems, where both $|X|$ and $|\S|$ are finite, and use the notation $n = |X|$ and $m =|\S|$ throughout this paper. 
We study the \emph{discrepancy problem}, which asks for a two-coloring $\chi \colon X \to \left\{-1, 1\right\}$  that minimizes the \emph{discrepancy}
			\[
			 \disc_{ \S } \round{\chi} 
			 =\max_{S \in \S} \left| \sum_{x \in S} \chi \left( x \right) \right|.
			 \]


Starting from 1950s, the study of low-discrepancy colorings has been an active area of research with applications in various branches of mathematics and computer science. 
As it is often termed, the `discrepancy method' inspired many approximation algorithms for discrete optimisation problems. For instance, it is an important tool in rounding fractional solutions of a linear system of equations to integral ones \citep{LVS86, rothvoss2012, BN17,rothvoss2013} and in recent combinatorial algorithms proposed for problems such as  bin-packing \citep{rothvoss2013,EPR13,HR17}, or scheduling problems \citep{BCKL14,BDJR22}.
In data approximation, a coloring with discrepancy $o(\sqrt{n})$ can 
be used to construct  $o(1/\eps^2)$-sized $\eps$-approximations (or $\eps$-samples), outperforming the $\Theta(1/\eps^2)$ guarantee of a single random sample~\citep{MWW93}.
Discrepancy is also closely connected to the sample complexity of learning. For instance, in the paper of~\cite{BBL02}, discrepancy of a random balanced coloring is used to construct penalized empirical risk minimization algorithms, leading to improved statistical guarantees. Furthermore, the study of Rademacher complexity
can be seen as a study of discrepancy of a random coloring. In more recent works, the discrepancy method had a key role in core-set constructions for kernel density estimation~\citep{phil13,PT20,tai20} and quantizing neural networks \citep{LS21}.
For further details and other examples of applications, we refer the interested reader to dedicated books on discrepancy~\citep{Chazelle:2000:DMR:507108, CST14,MatDiscBook}.

 \subsection*{Previous algorithms} 
 \cite{S85} showed that for any set system $(X,\S)$ there exists a two-coloring of $X$
with discrepancy $O \left( \sqrt{n \ln(m/n)}  \right)$, which is tight for $m = O(n)$.
However, his original proof only demonstrated the existence of such a coloring, without any efficient algorithm to construct it.
Finding a polynomial-time algorithm to construct colorings with optimal discrepancy had remained a major open problem for $25$ years, until a breakthrough result of \cite{DBLP:conf/focs/Bansal10}, who gave a randomized polynomial-time algorithm with near-optimal discrepancy guarantees. 
Since then, several researchers have proposed new polynomial-time algorithms with optimal discrepancy guarantees~\citep{hss14,DBLP:journals/siamcomp/LovettM15,DBLP:conf/ipco/LevyRR17,DBLP:conf/stoc/BansalDGL18}.
At the core of these methods is a random walk algorithm  which starts with the uniformly $0$ coloring, and at each step, updates the color of an element by adding a small increment to its coloring. If a variable reaches `$-1$' or `$1$', its value is fixed. The increment is determined by either solving an appropriate SDP \citep{DBLP:conf/focs/Bansal10,DBLP:journals/siamcomp/LovettM15, hss14}, or Gram-Schmidt orthogonalization \citep{DBLP:conf/stoc/BansalDGL18}, or by a deterministic algorithm using the multiplicative weights update (MWU) method \citep{DBLP:conf/ipco/LevyRR17}. 
The next table contains a summary:

\begin{table}[ht!]{}
\centering
\resizebox{0.95\textwidth}{!}{
\begin{tabular}{c|c|c|c|c}
\toprule
\multicolumn{1}{c|}{\textsc{Set system}}                       
& \multicolumn{1}{c|}{\textsc{Discrepancy}}                                               
& \multicolumn{1}{c|}{\textsc{Time}} 
& \multicolumn{1}{c|}{\textsc{Method}} 
& \multicolumn{1}{c}{\textsc{Citation}} 

\\ \midrule
\rule{0pt}{3ex}
\multirow{4}{*}{arbitrary} 
& \multirow{4}{*}{$O\left(\sqrt{n \ln \round{\frac{m}{n}}}\right)$} 

                            

                            
 & $\aO\round{n^3 + m^3}$    
 & random walk
 & \cite{DBLP:journals/siamcomp/LovettM15}        

 \\ 
   \rule{0pt}{3.5ex}                         
 &  
 & $O\round{n^4m}$    
 & random walk via MWU    
 & \cite{DBLP:conf/ipco/LevyRR17}        

 \\ 
    \rule{0pt}{3.5ex}   
 & 
 & $O\round{n^{3.38}+ nm^{2.38}}$    
 & Gram-Schmidt walk    
 & \cite{DBLP:conf/stoc/BansalDGL18}

 \\ \midrule

 \rule{0pt}{3ex}
{$\pi_{\S}(k) = O(k^d)$} 
& {$O\left(\sqrt{n^{1-1/d} }\right)$} 
& $\aO\round{n^3 + m^3}$
& partial coloring      
& \makecell*{\cite{DBLP:journals/dcg/Matousek95a} \\ \cite{DBLP:journals/siamcomp/LovettM15}}        

\\[1ex] 

  \midrule

 \rule{0pt}{3ex}
\multirow{1}{*}{$\pi_{\S}^*(k) = O(k^d)$} 
& \multirow{1}{*}{$O\left(\sqrt{n^{1-1/d} \ln m}\right)$} 
& $O(mn^3)$    
& MWU      
& \cite{MWW93}

\\[1ex] 

  \midrule \midrule

\rule{0pt}{3ex}
\multirow{1}{*}{$\pi_{\S}^*(k) = O(k^d)$} 
& \multirow{1}{*}{$O\left(\sqrt{n^{1-1/d} \ln m}\right)$} 
& $O(mn^{1/d} + n^{2+1/d})$    
& \makecell*{Sampling + \\ Primal-Dual + MWU }
& \textbf{This Paper}

\\[1ex] 

  \midrule  

\rule{0pt}{3ex}
\multirow{1}{*}{$\pi_{\S}^*(k) = O(k^d)$} 
& \multirow{1}{*}{$O\round{\sqrt{  n^{1-\alpha/d}\ln m   }}$} 
& $O\round{mn^{\alpha/d} + n^{1+\alpha(1+1/d)} }$    
& \makecell*{Sampling + Pruning + \\ Primal-Dual + MWU  }
& \textbf{This Paper}

\\[1ex] 



 \bottomrule

\end{tabular}
}
		\label{table:prev-results}
\end{table}

\noindent \textbf{Despite heavy interest for the past decades, still no efficient implementations with these guarantees exist. Indeed, that remains one of the open questions; see \href{https://homepages.cwi.nl/~dadush/workshop/discrepancy-ip/open-problems.html}{here}.}   \\

\noindent
In this work, we consider set systems with polynomially bounded dual shatter function. 

\begin{definition}[Dual-shatter function]
Let $(X,\S)$ be a set system. For any $\R \subseteq \S$, we say that the elements $x,y \in X$ are equivalent with respect to $\R$ if $x$ belongs to the same sets of $\R$ as $y$. Then $\pi_{\S}^* (k)$, where  $\pi_\S^*$  is called the dual-shatter function of $\S$, is defined to be the maximum number of equivalence classes on $X$ defined by any $k$-element subfamily $\R \subseteq \S$.
\end{definition}
\noindent
The class of set systems with polynomially bounded $\pi_\S^*(k)$ contains several fundamental cases:
\begin{itemize}
\item  set systems with dual VC-dimension at most $d$ (it implies  $\pi_{\S}^*(k) \leq \round{\frac{ek}{d}}^d$~\citep{Sa72,Sh72});
\item geometric set systems  induced by (unions or intersections of) half-spaces, balls, etc;
\item geometric set systems where $X$ is a set of $n$ points in $\RR^d$ and each range in $\S$ can be obtained as an intersection of $X$ with a semialgebraic set of bounded complexity;
 \item  set systems $(X, \S)$ with the property that the common intersection of any $d$ ranges from $\S$ has size at most $c$, for given constants $c$ and $d$ \citep{mat97}.
\end{itemize}

\bigskip

\newpage
\noindent We now present our five main algorithms. 
\textit{A highlight of our algorithms, besides near-quadratic improvement over previous-best running times, is that they avoid any input-specific tools, such as spatial partitioning. Thus we get improved pratical
constructions for many fundamental geometric set systems, narrowing the gap between theory and practice.}

\bigskip	

\noindent
\textbf{{\Large 1.} \textsc{Discrepancy}.} Our main result on low-discrepancy colorings is the following.
\begin{tcolorbox}
\begin{restatable}[Main Theorem]{thm}{maintheorem}
\label{thm:main}
	Let $(X,\S)$ be a finite set system, $n = |X|, m = |\S|$, and $c,d$ be constants such that $\pi^*_\S(k) \leq c \cdot k^d$.  Then there is a randomized algorithm that constructs a coloring $\chi$ of $X$  with expected discrepancy at most
	   \[
	   3\sqrt{ \frac{ 9c^{1/d}}{2} \cdot n^{1-1/d}\ln m  + 19 \ln^2 m \ln n}
	   \]
	   with at most
	   \[
		 \frac{34n^{2+ 1/d} \ln n}{c^{1/d}}  + \frac{25mn^{1/d}\ln (mn)\cdot\log n}{c^{1/d}} 
		\] 
		expected calls to the membership Oracle of $(X,\S)$.
\end{restatable}
\end{tcolorbox} 
 \medskip

\noindent 
Our algorithm is very simple and does not use any advanced subroutines or data structures: 

\medskip

{\centering
\begin{minipage}{.928\linewidth}
\begin{algorithm}[H] 
\algotitle{ \textsc{LowDiscColor-DualShatter}}{discalgoDS}
\caption{  \textsc{LowDiscColor-DualShatter}$\big((X,\S),  d\big)$}
	\label{algo:discrepancy-outline-dualshat}
	\While{ $|X| \geq 4$ }{
	$n \leftarrow |X|$, $E \leftarrow \binom{X}{2}$ \tcp*{\text{$E$ is the set of all edges on $X$}}

	set the weight of each element in $E, \S$ to $1$ 

 	\For{$i = 1, \dots, n/4$ }{
		sample $e_i = \{x_i,y_i\}$ from $E$ and \ $S_i $ from $\S$ (according to their weights)

		set $\chi(x_i) =  \pm 1$ with prob. $1/2$; set $\chi(y_i) = - \chi (x_i)$

        set $X \leftarrow X \setminus \left\{x_i, y_i \right\}$, and the weight of $e_i$ and its  adjacent edges to zero 

		$E_i \subseteq E:$ uniform sample of size $\aO\round{|E|/n^{1-1/d}}$

		halve weight of each $e \in E_i$ satisfying $|e \cap S_i| = 1$
		
	    $\S_i \subseteq \S:$ uniform sample of size $\aO\round{|\S|/n^{1-1/d}}$
	    
	    double weight of each $S \in \S_i$ satisfying $|S \cap e_i| = 1$ 
	}	}
	\textbf{return} $\chi$ (color the remaining at most $4$ elements of $X$ arbitrarily)
\end{algorithm}
\end{minipage}
\par 
}

\medskip
 
\noindent Importantly, the improved running time and the simplicity of this new method make it possible  to perform an empirical study of low-discrepancy colorings of abstract and high-dimensional geometric set systems. 
As an illustration,  the figures below show the average discrepancies in set systems induced by half-spaces in dimensions $2,3,$ and $4$, observed over $10$ repetitions of our method, compared with a purely random coloring (the shaded areas denote $\pm 1$ standard deviation from the mean). 


\begin{figure*}[ht]
			\centering
			\includegraphics[width=0.32\textwidth]{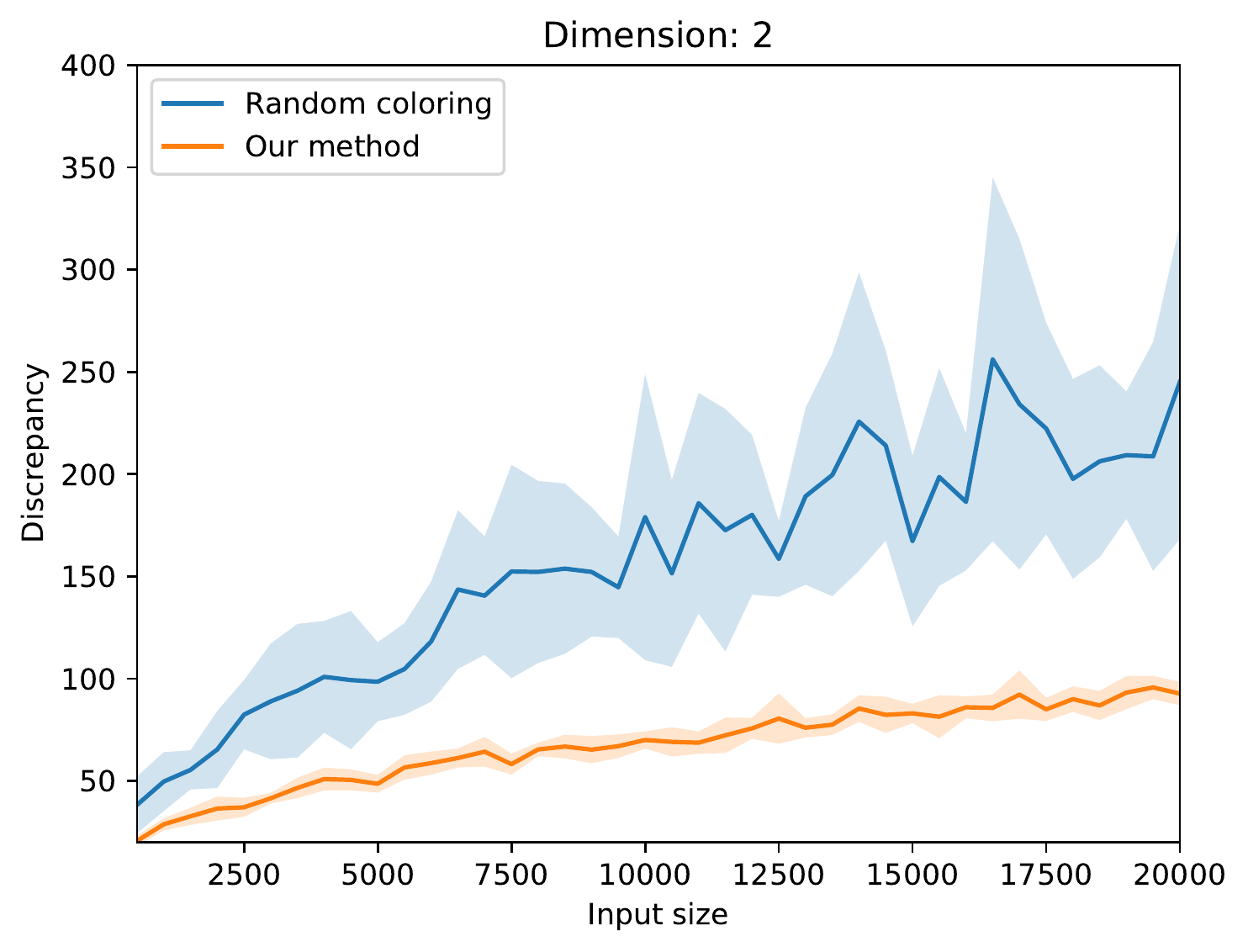}
			\hfill
			\includegraphics[width=0.32\textwidth]{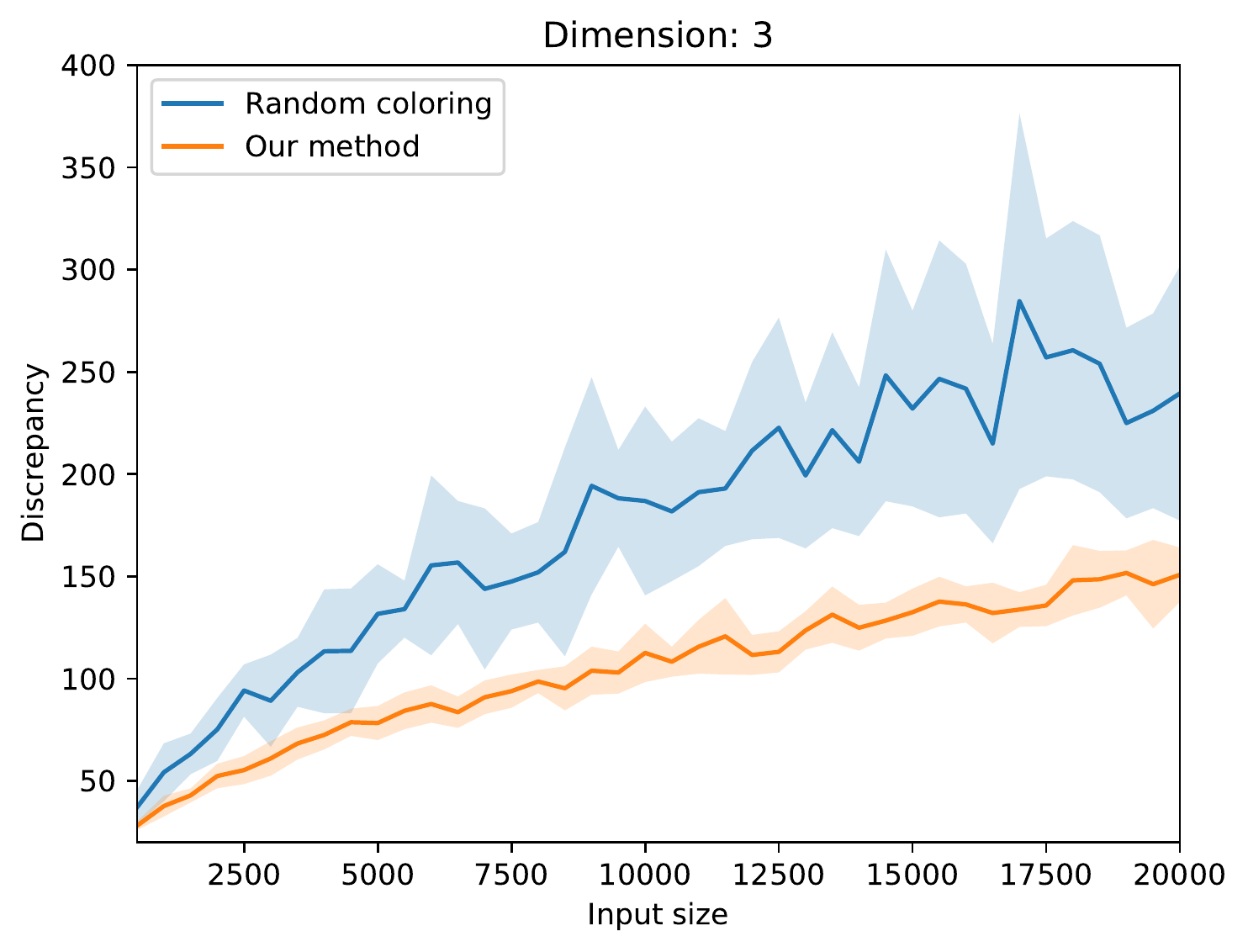}
			\hfill
			\includegraphics[width=0.32\textwidth]{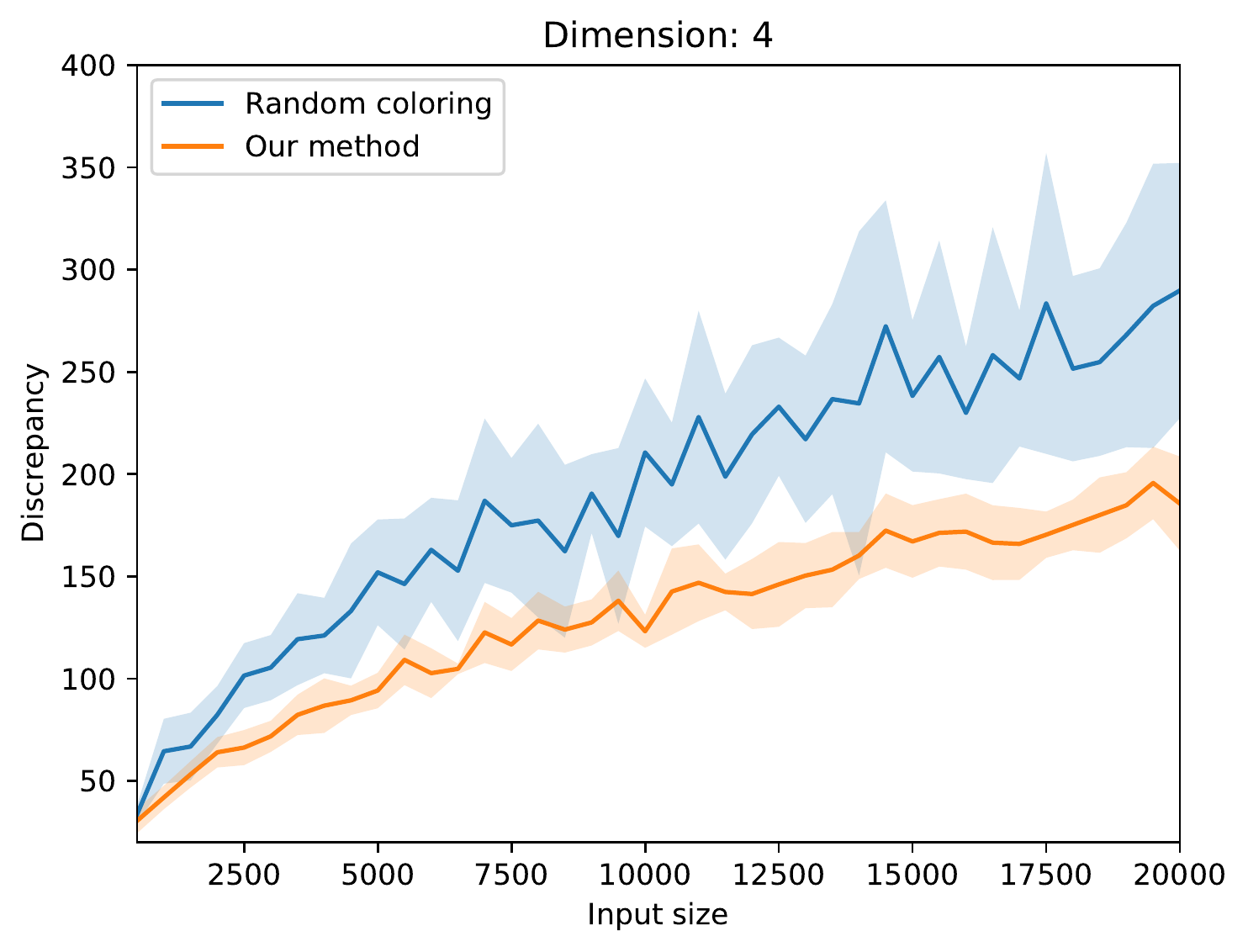}
			\label{fig:intro-disc}
	\end{figure*}

\bigskip
\bigskip
\newpage
\noindent \textbf{{\Large 2.}  \textsc{Matchings with Low Crossing Number}.} 
The key property which guarantees that the output of \nameref{discalgoDS} has low discrepancy is that any range in $\S$ crosses\footnote{We say that $S \in \S$ \emph{crosses} an edge  $\{x,y\} $ if $|S \cap \{x,y\}| =1$.} at most $O\round{n^{1-1/d}}$ of the selected edges $\curly{\curly{x_i, y_i}}_{i=1}^{n/2}$. 

In general, given a set system $(X, \S)$ and a perfect matching\footnote{Partition of $X$ into $n/2$ disjoint pairs (edges).} $M$ of $X$,
we define the \emph{crossing number} of $M$ with respect to $\S$ as the maximum number of edges of $M$ crossed by a single range $S\in \S$. 
The study of perfect matchings (along with spanning paths and spanning trees) with low crossing number was originally introduced for geometric range searching \citep{Wel88,CW89}. Since then, they have found applications in various fields, for instance, discrepancy theory \citep{MWW93}, learning theory \citep{AMY16}, or algorithmic graph theory \citep{DHV20}. 

The core of \nameref{discalgoDS} can be generalized to construct low-crossing matchings
in set systems satisfying the following assumption:
\begin{tcolorbox}
\begin{assumption*}[\textsc{MainAssumption}$(a,b,\gamma)$]\label{assumption}
	$(X,\S)$ is a finite set system with $m \geq n$, $m \geq 34$, such that any $Y \subseteq X$ has a perfect matching with crossing number at most $a|Y|^\gamma + b$ with respect to $\S$.\\

	It is known that if $\pi^*_\S(k) \leq c k^d$, then $(X,\S)$ satisfies the \textsc{MainAssumption} with parameters $a = \frac{(2c)^{1/d} }{2\ln 2(1-1/d)}$, $b = \frac{\ln m}{\ln 2}$, and $\gamma = 1-1/d$.~\citep{MatDiscBook}
	\end{assumption*}
\end{tcolorbox}


\noindent The main technical ingredient of this work---of independent interest and improving the previous-best construction time of $O(mn^3)$ known for several decades---is the following. 

\begin{theorem}\label{thm:main-matching-result}
	Let $(X,\S)$ be a set system satisfying \nameref{assumption}.
	Then there is a randomized algorithm that returns a perfect matching of $X$ with expected crossing number at most
		\[
		\frac{3a}{\gamma}  n^\gamma + \frac{3b\log n}{2} + 18 \ln \round{mn}\log n
		\]
		with  at most
		\[
		\min\curly{  \frac{24n^{3-\gamma} \ln n}{a}  + \frac{18mn^{1-\gamma}\ln mn}{a}\cdot\min\curly{\frac{2}{1-\gamma},~ \log n} ,~ \frac{n^3}{7} + \frac{mn}{2}}.
		\]
		expected calls to the membership Oracle of $(X, \S)$.
\end{theorem}

\noindent \Cref{fig:matchings} illustrates the matchings constructed by our algorithm and random sampling on different input point arrangements  and range types. It is clear that our method explicitly takes into account the information about ranges, which leads to different outcomes for set systems induced by half-spaces and balls.  On the other hand,  random sampling fails to preserve the intrinsic structure of the point set. 
We find it surprising that our method, that is based only on (non-uniform) sampling, gives a matching that adapts so well to each specific instance. 


\begin{figure}[ht!]
			\centering
			\includegraphics[width=0.28\textwidth]{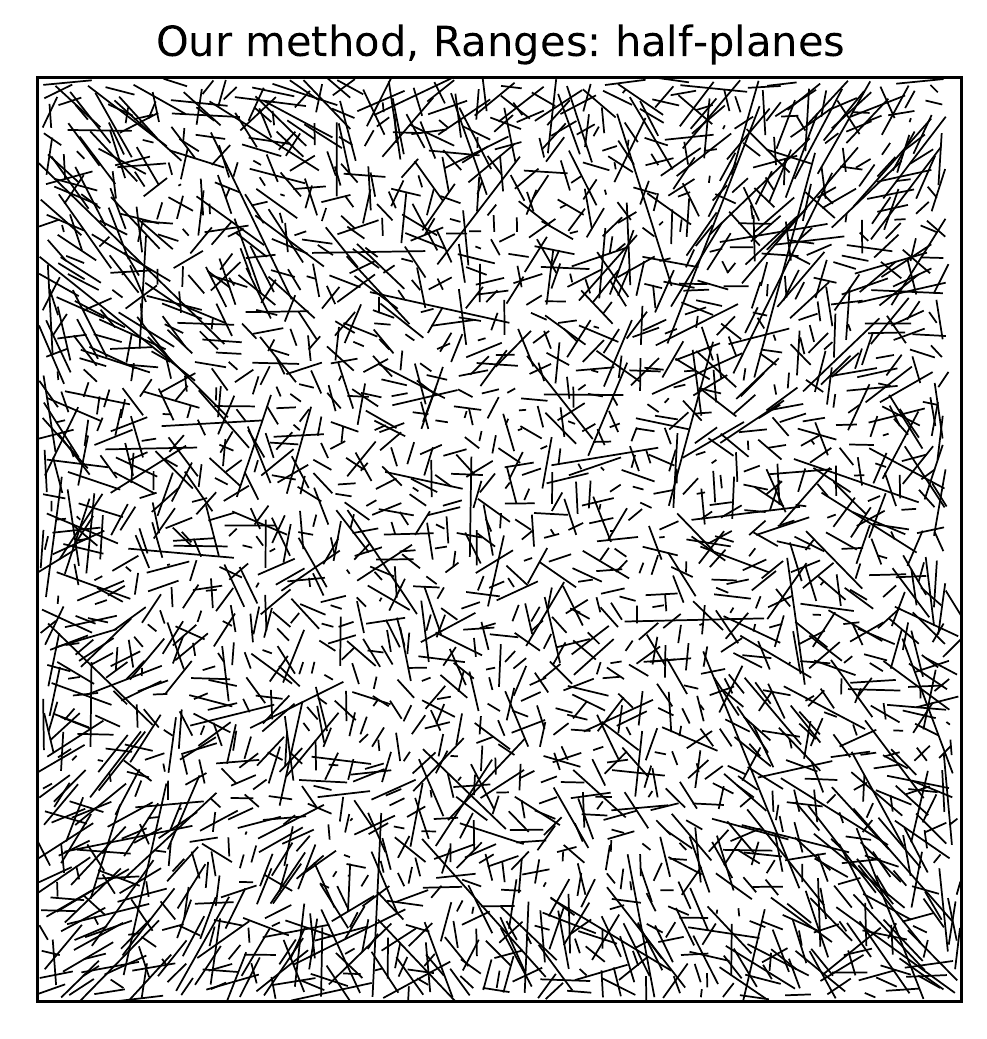}
			\includegraphics[width=0.28\textwidth]{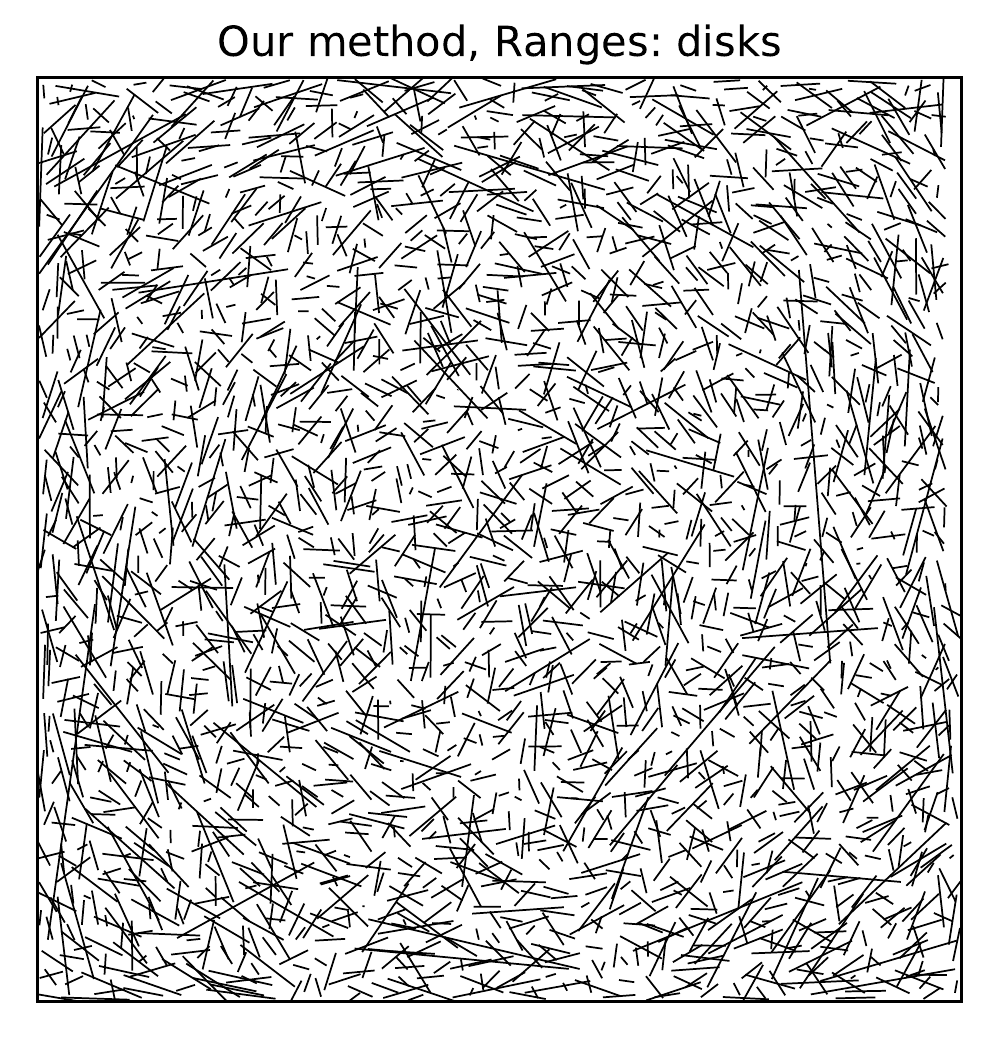}
			\includegraphics[width=0.28\textwidth]{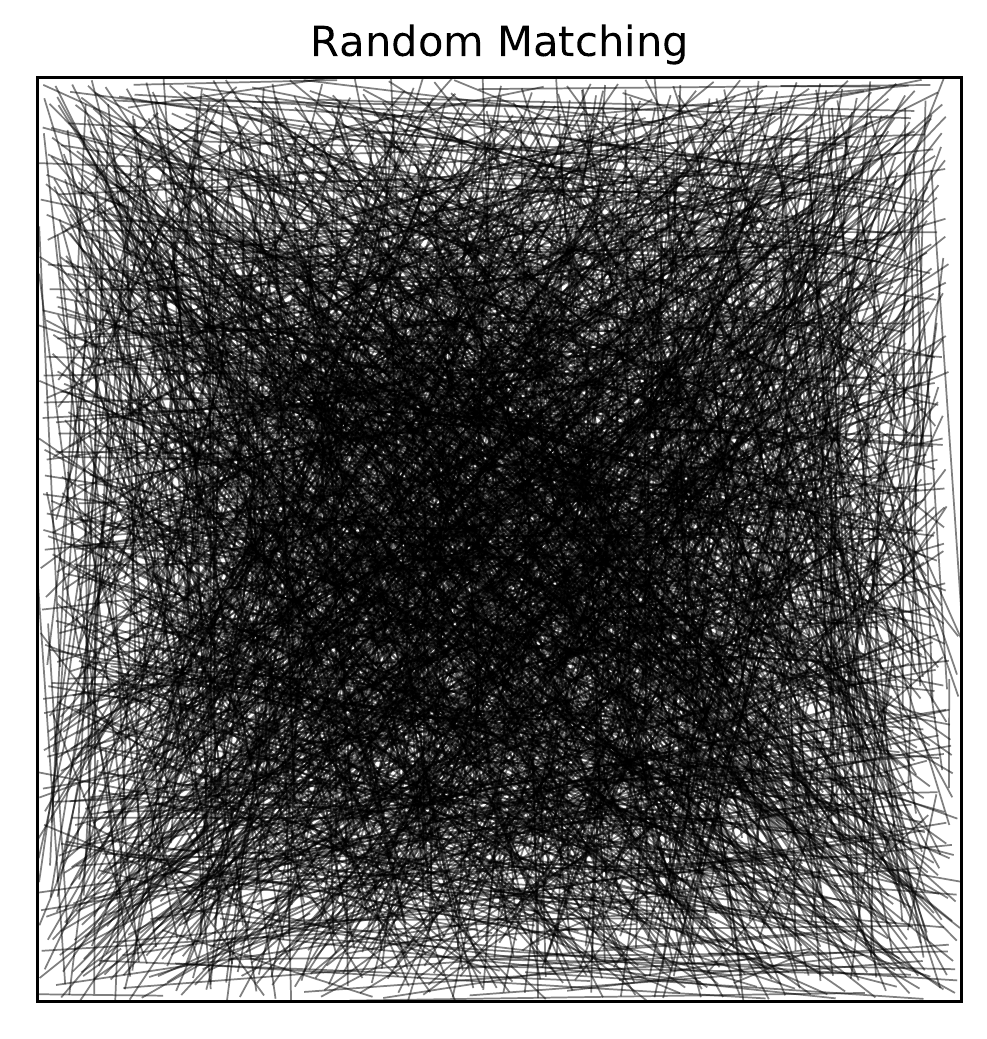}\\
			\includegraphics[width=0.28\textwidth]{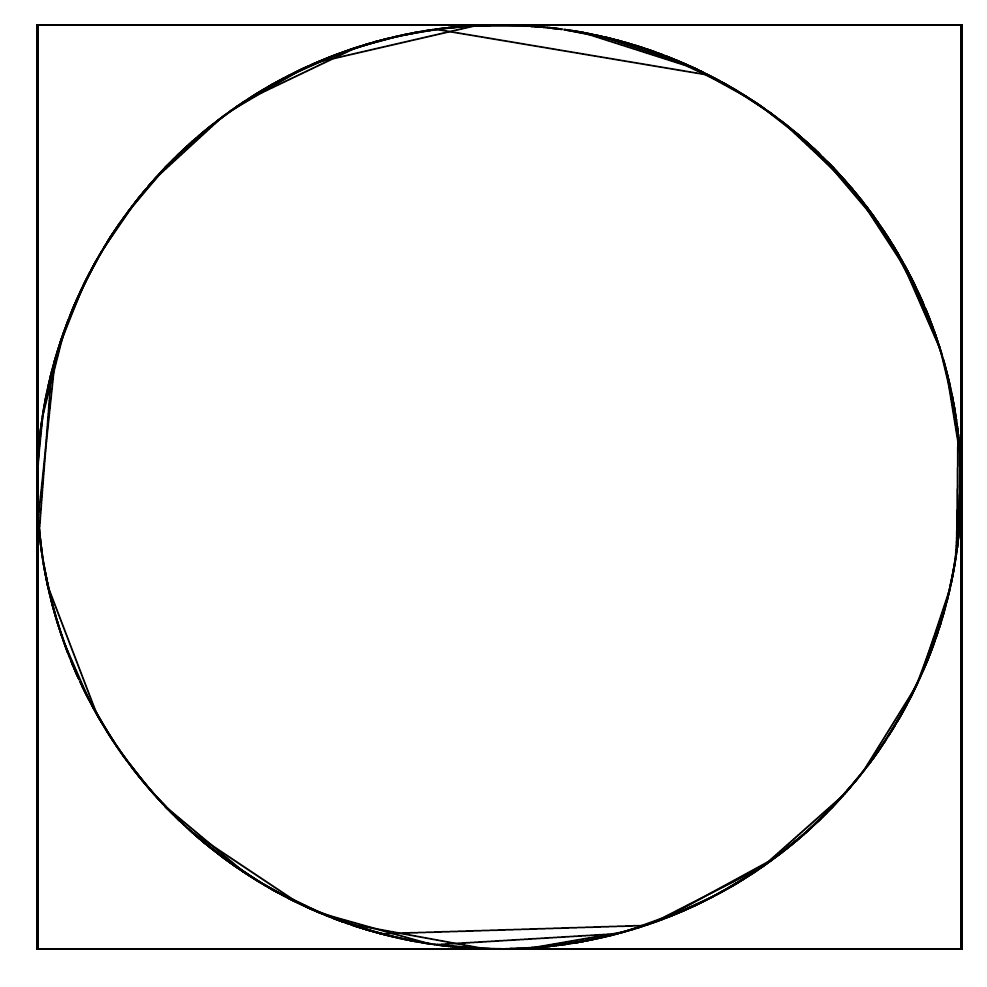}
			\includegraphics[width=0.28\textwidth]{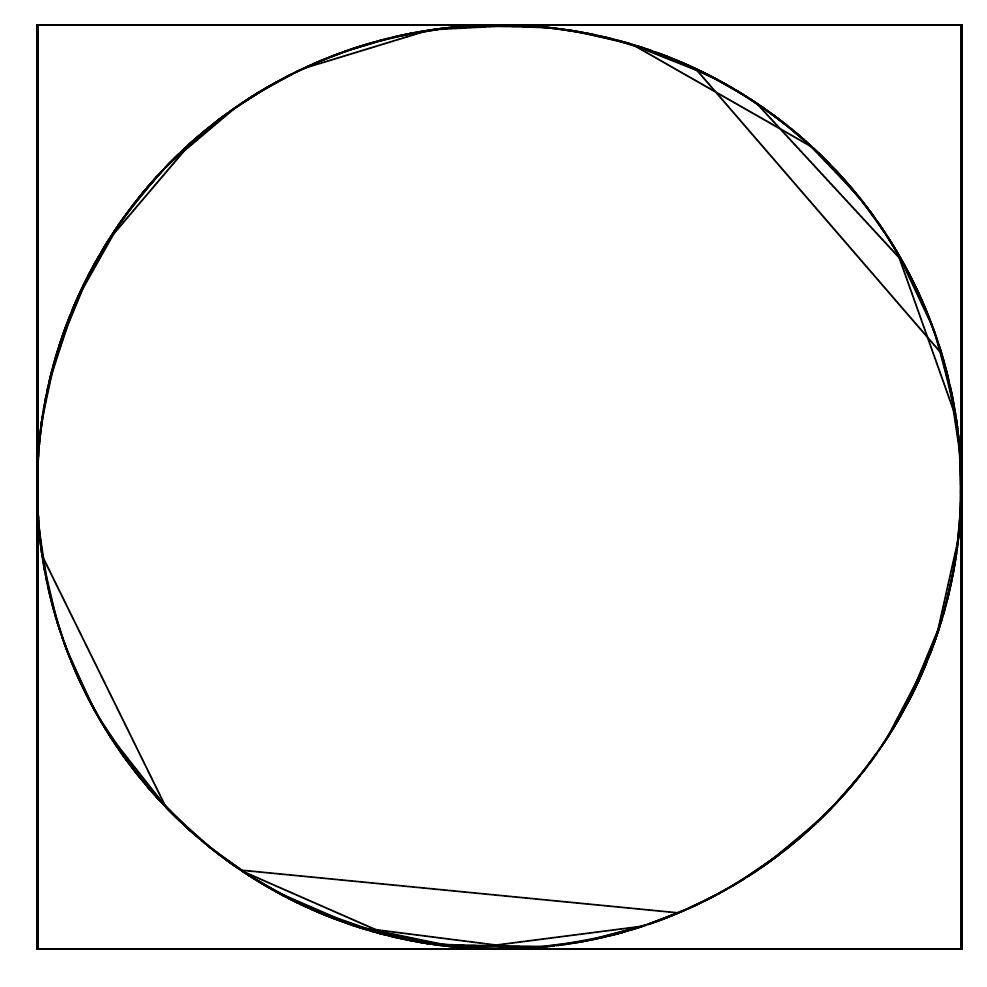}
			\includegraphics[width=0.28\textwidth]{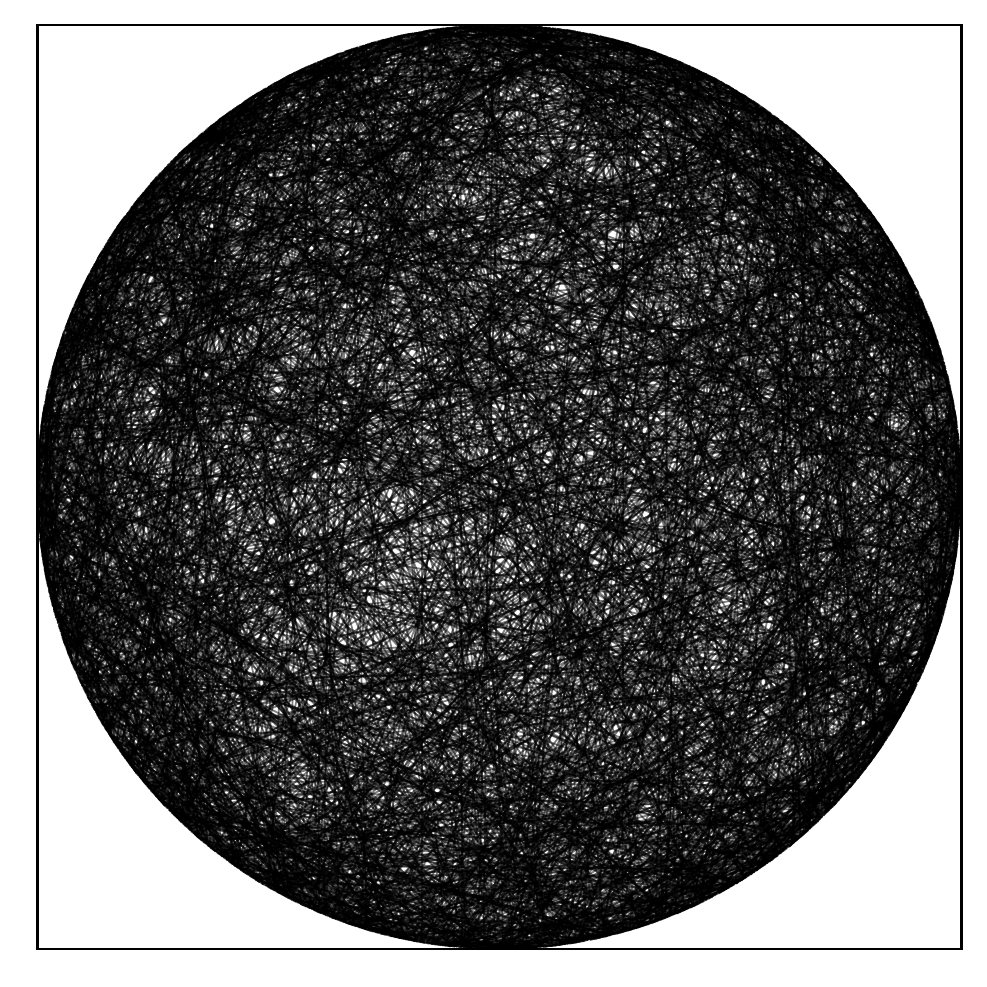}\\
			\includegraphics[width=0.28\textwidth]{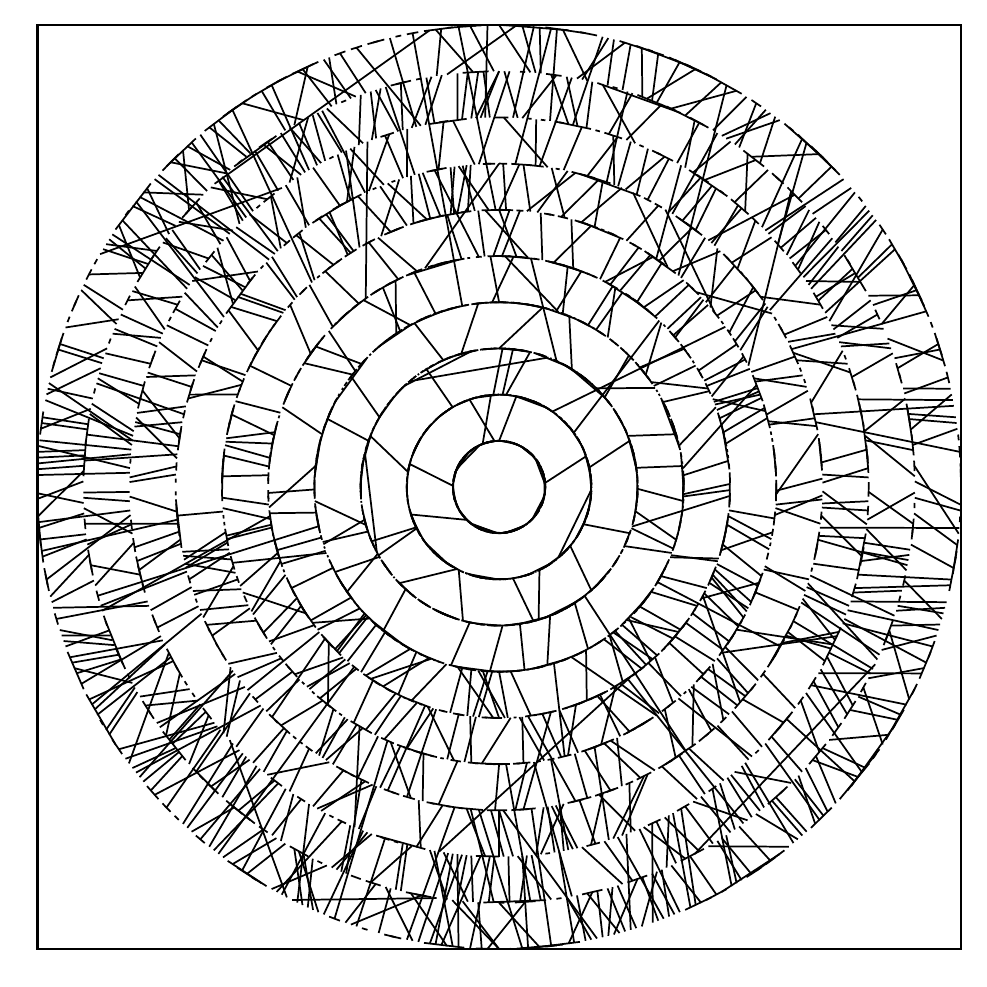}
			\includegraphics[width=0.28\textwidth]{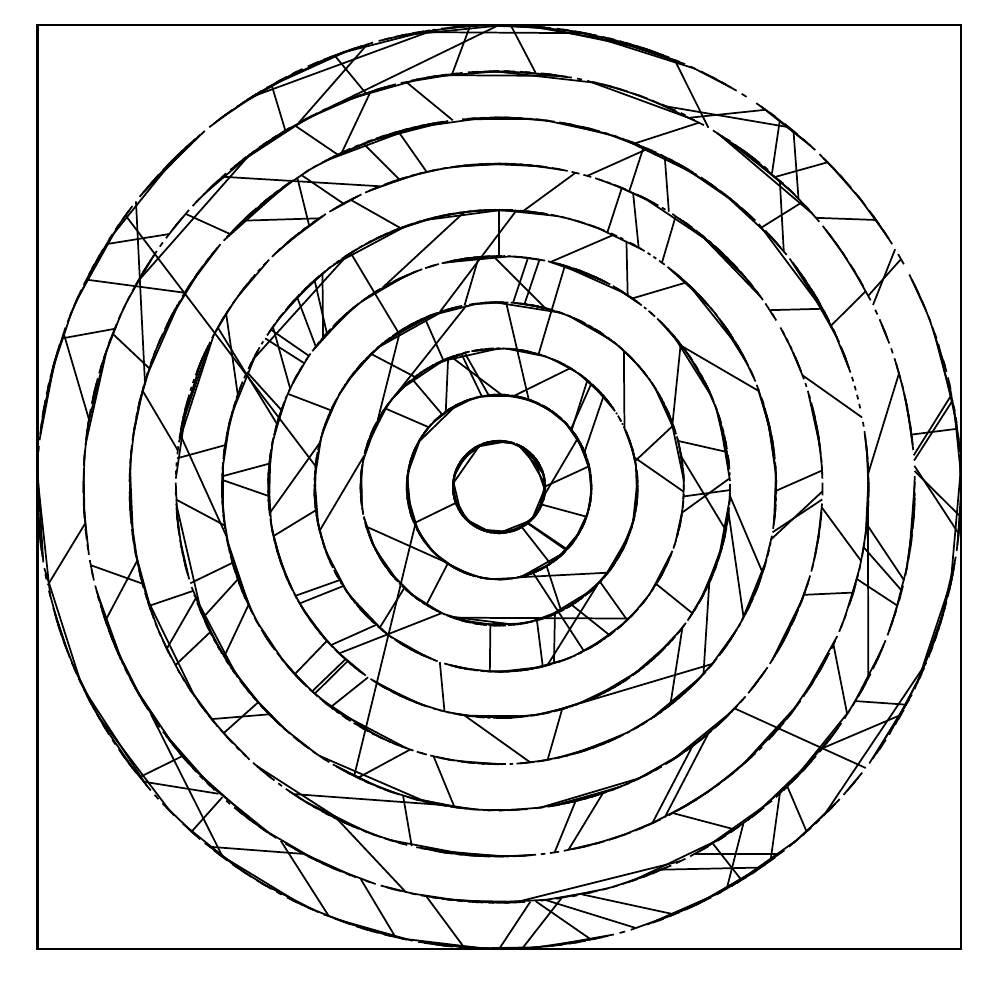}
			\includegraphics[width=0.28\textwidth]{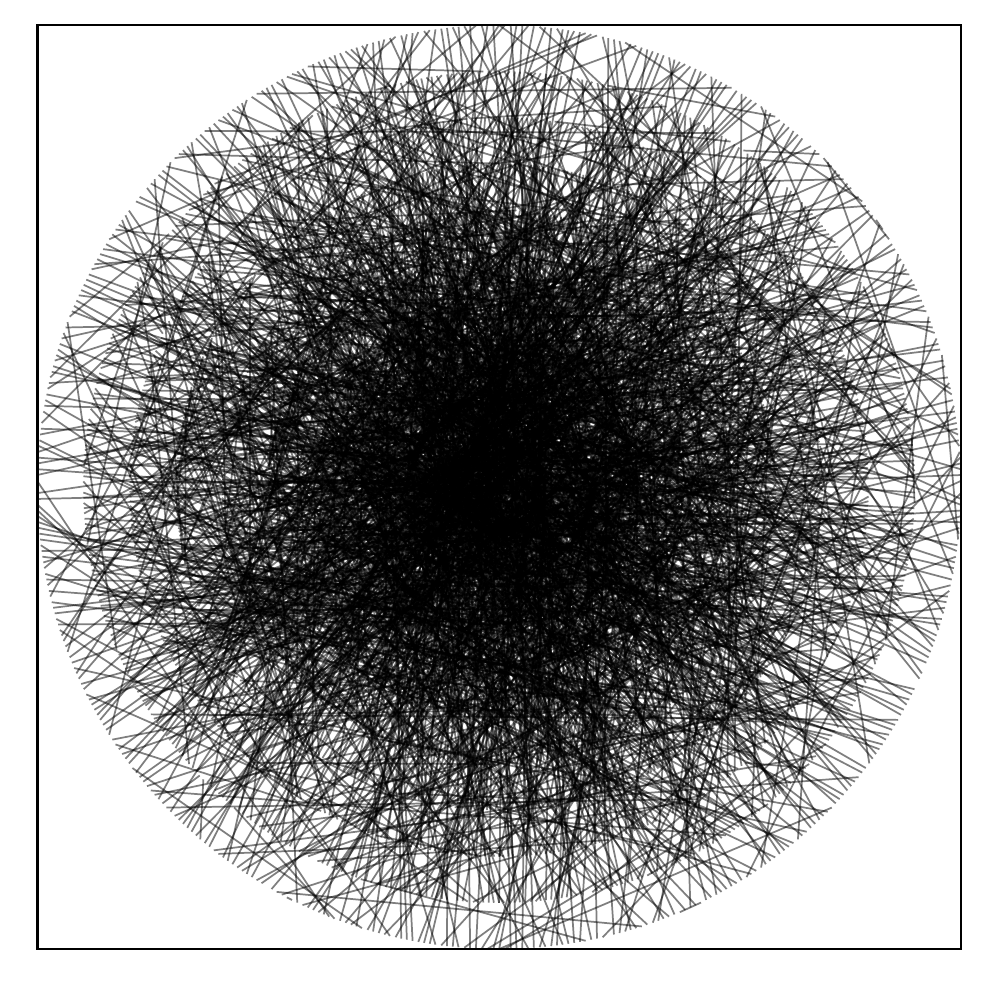}
			\caption{Matchings of $5000$ points. 
			\texttt{Left column:} our method with half-plane ranges. \texttt{Middle column:} our method with disk ranges. \texttt{Right column:} random sampling. We emphasize that each of these figures contain exactly $2500$ edges, which correspond to the matching of $5000$ points. 
			}
			\label{fig:matchings}
	\end{figure}

\bigskip

\noindent \textbf{{\Large 3.}  \textsc{Approximations}.} An immediate consequence of our algorithms is an efficient construction algorithm of $\eps$-approximations of sub-quandratic size. Given a finite set system $(X,\S)$ and a parameter $\eps \in (0,1) $, a set $A \subseteq X$ is an \emph{$\eps$-approximation} of $(X,\S)$ if the following holds for all $S \in \S$:
		\[
		\abs{
			\frac{|S|}{|X|} - \frac{|A \cap S|}{|A|}
		} \leq \eps .
		\]		
		Furthermore,   let $\eps(A,X,\S)$ denote the smallest $\eps$ for which $A$ is an $\eps$-approximation of $(X,\S)$:
	\[
		\eps(A,X,\S) = \max\limits_{S\in\S} \abs{
			\frac{|S|}{|X|} - \frac{|A \cap S|}{|A|}
		}.
	\]

	\newpage
	\noindent
 The iterative application of \nameref{thm:main} implies the following   on $\eps$-approximations. 

\begin{corollary} \label{cor:epsapproximations} 
Let $\eps \in (0,1)$, $(X,\S)$ be a set system and $c,d$ be constants such that $\pi^*_\S(k) \leq c k^d$. Then there is a randomized algorithm which returns a set $A \subset X$ of size 
     \begin{align*}
     O\round{ \max\curly{
     	\round{\frac{c^{1/d}\ln m}{\eps^2} }^{\frac{d}{d+1}},
     	\frac{\sqrt{\ln n} \ln m}{\eps}
     	}
     }
     \end{align*}
     with     	$\EE[\eps(A,X,\S)] \leq \eps$,
    and with an expected
    \[
    	O\round{ mn^{1/d} \ln(mn)  \log^2 n + n^{2 +1/d}\ln n}
    \]
    calls to the membership Oracle of $(X,\S)$.
\end{corollary}

\begin{remark*} Previous-best algorithms for constructing $o(1/\eps^2)$-sized $\eps$-approximations of  set systems with polynomially bounded dual shatter functions were based on the low-discrepancy coloring approach of \cite{MWW93}, with time complexity $O(mn^3)$.
\end{remark*}

\noindent
For set systems where uniform sampling yields small-sized $\eps$-approximations, for instance set systems with bounded VC-dimension, the guarantees can be improved.

\newpage

\begin{corollary}\label{cor:vcdim-apx-result-inverted-dualvc}
Let $\eps \in (0,1)$, $(X,\S)$ be a set system with VC-dimension $\dvc\geq 2$ and $c,d$ be constants such that $\pi^*_\S(k) \leq c k^d$. 
    Then
     there is a randomized algorithm that returns a set $A \subset X$ of size at most
     $
     O\round{
     \max\curly{
     	\round{\frac{c^{1/d} \cdot \dvc \ln \round{\frac{\dvc}{\eps}}}{\eps^2} }^{\frac{d}{d+1}},
     	\frac{\dvc \ln^{3/2} \round{\frac{\dvc}{\eps}}}{\eps}
     }}
     $
     with expected approximation guarantee
    $
    	\EE[\eps(A_j,X,\S)] \leq \eps,
    $
    and with an expected
    $
    	O\round{ \dvc \round{ \frac{\dvc}{\eps^2} }^{\dvc+1/d} \ln^3\round{\frac{\dvc}{\eps}}}
    $
    calls to the membership Oracle of $(X,\S)$.
\end{corollary}


The figures above present a visual comparison of the approximations created by our method (\texttt{Top row}) and random sampling (\texttt{Bottom row}). Both methods are applied to the same  set system on $10,000$ points with ranges induced by disks.

	\begin{figure}[ht!]
			\centering
			 \includegraphics[width=0.24\textwidth]{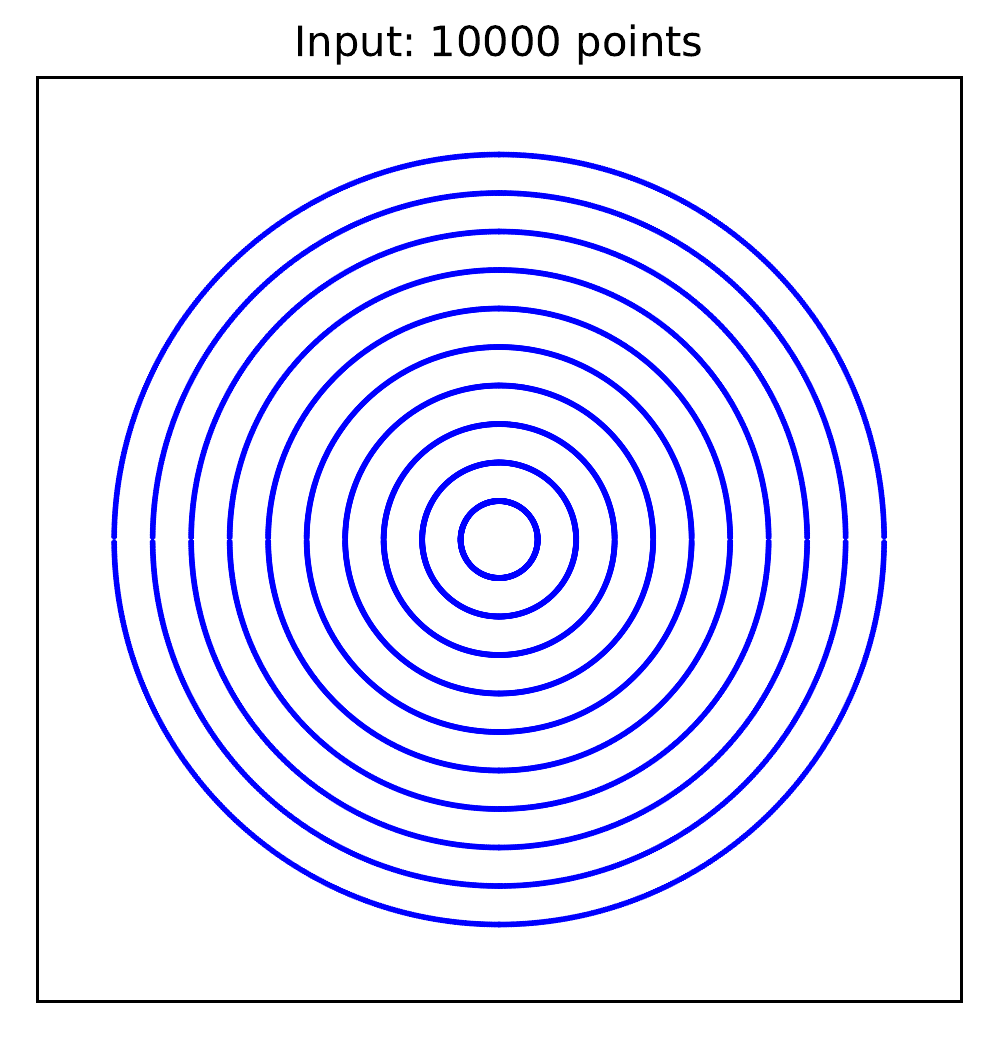}
			 \includegraphics[width=0.24\textwidth]{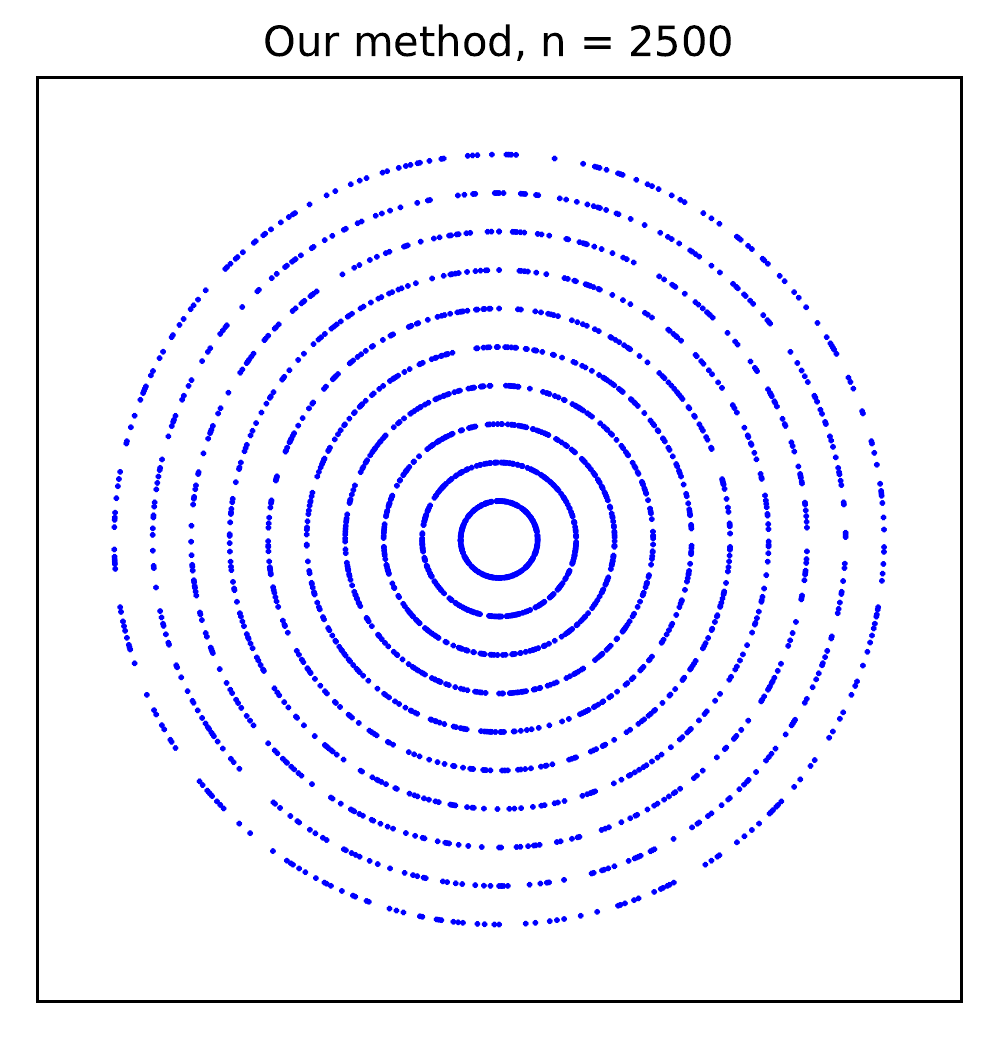}
			 \includegraphics[width=0.24\textwidth]{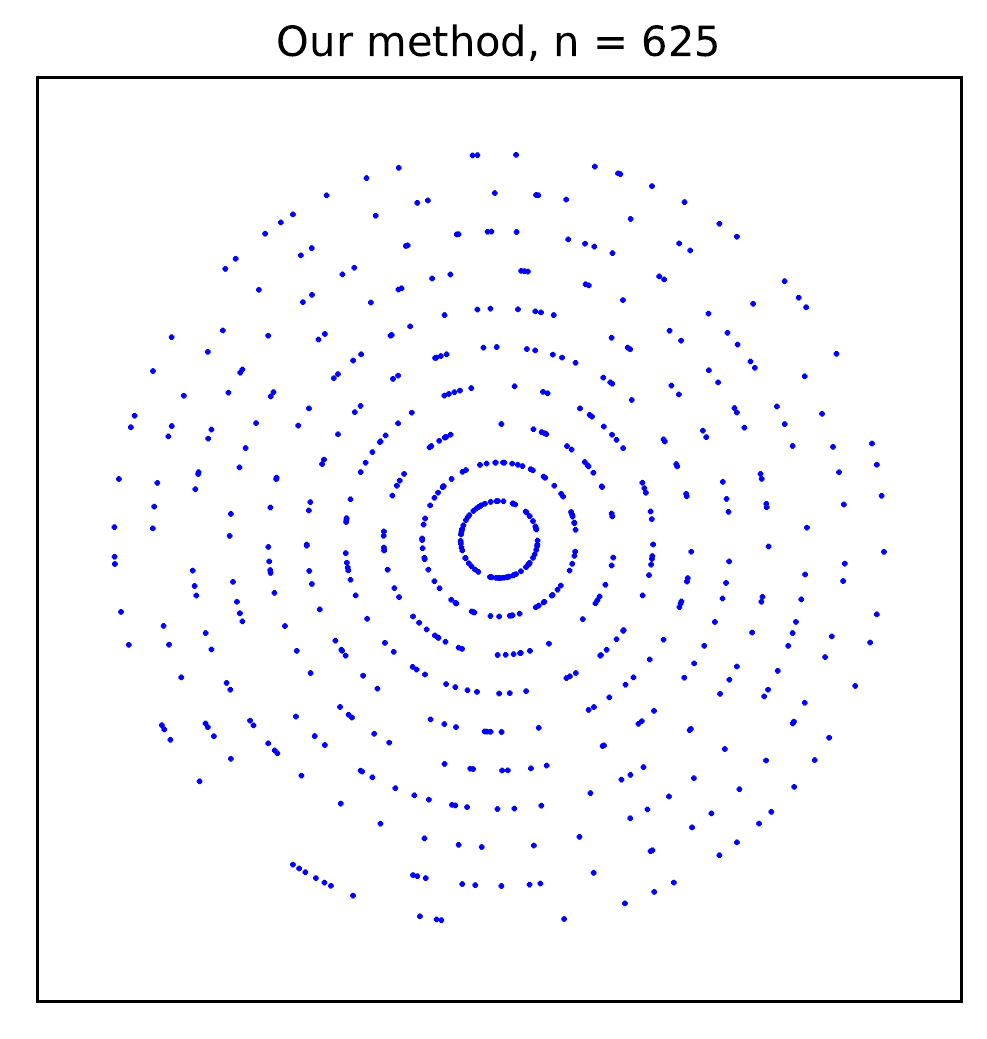}
			 \includegraphics[width=0.24\textwidth]{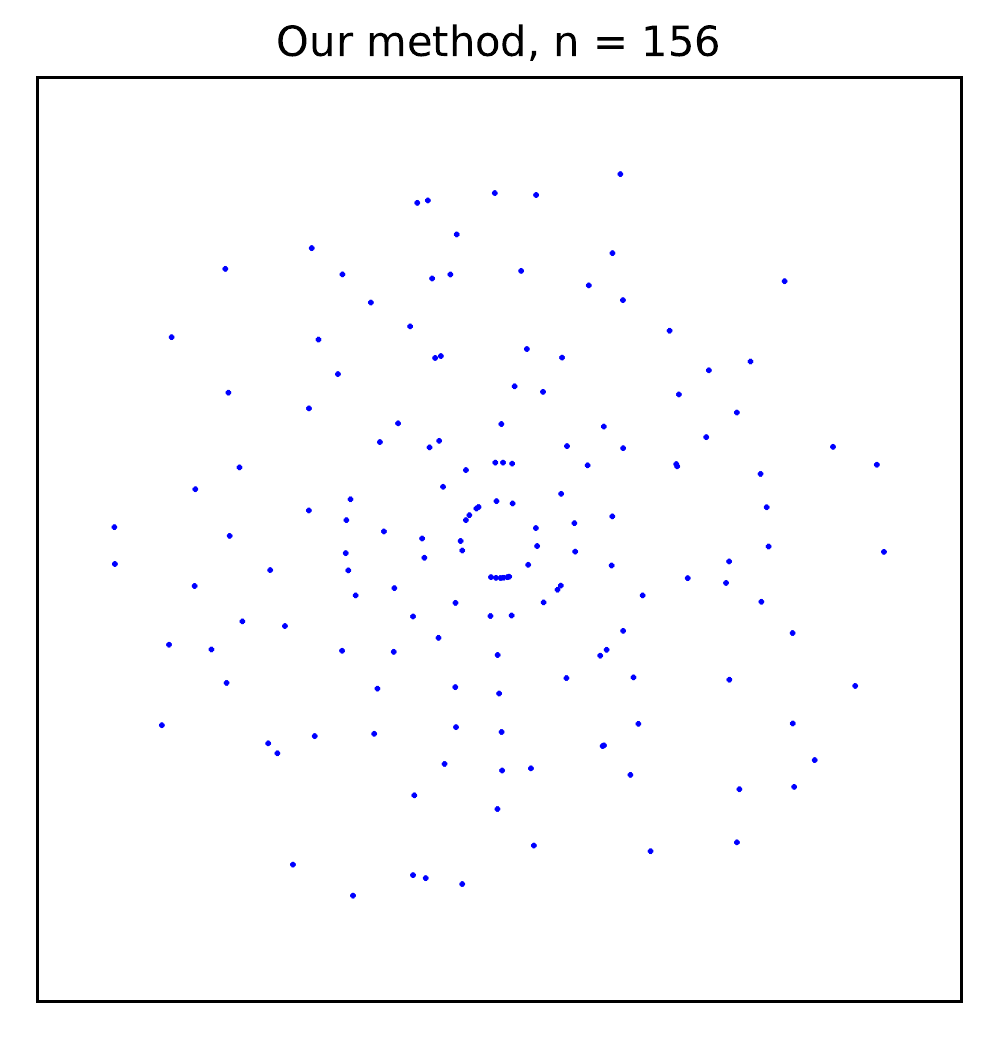}\\
			\includegraphics[width=0.24\textwidth]{CirclesOurs10000to10000.pdf}
			 \includegraphics[width=0.24\textwidth]{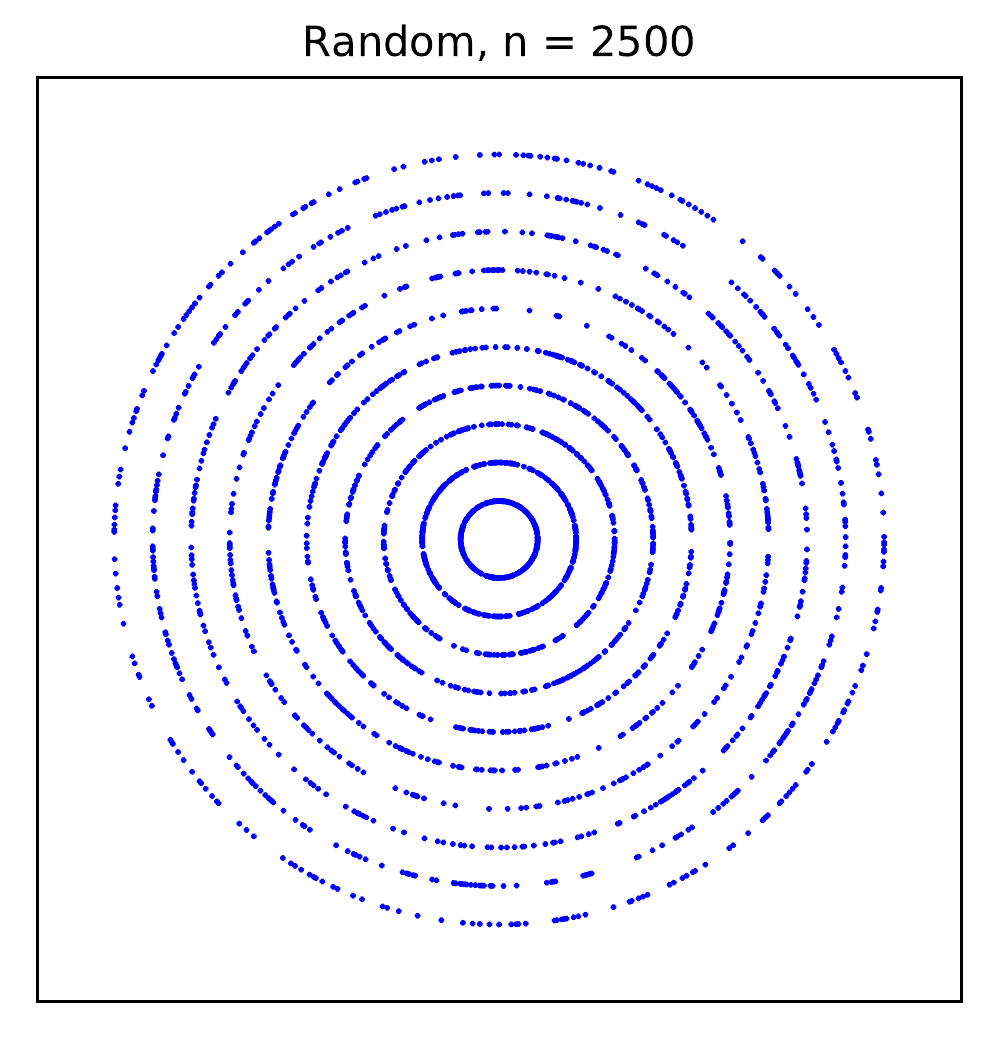}
			 \includegraphics[width=0.24\textwidth]{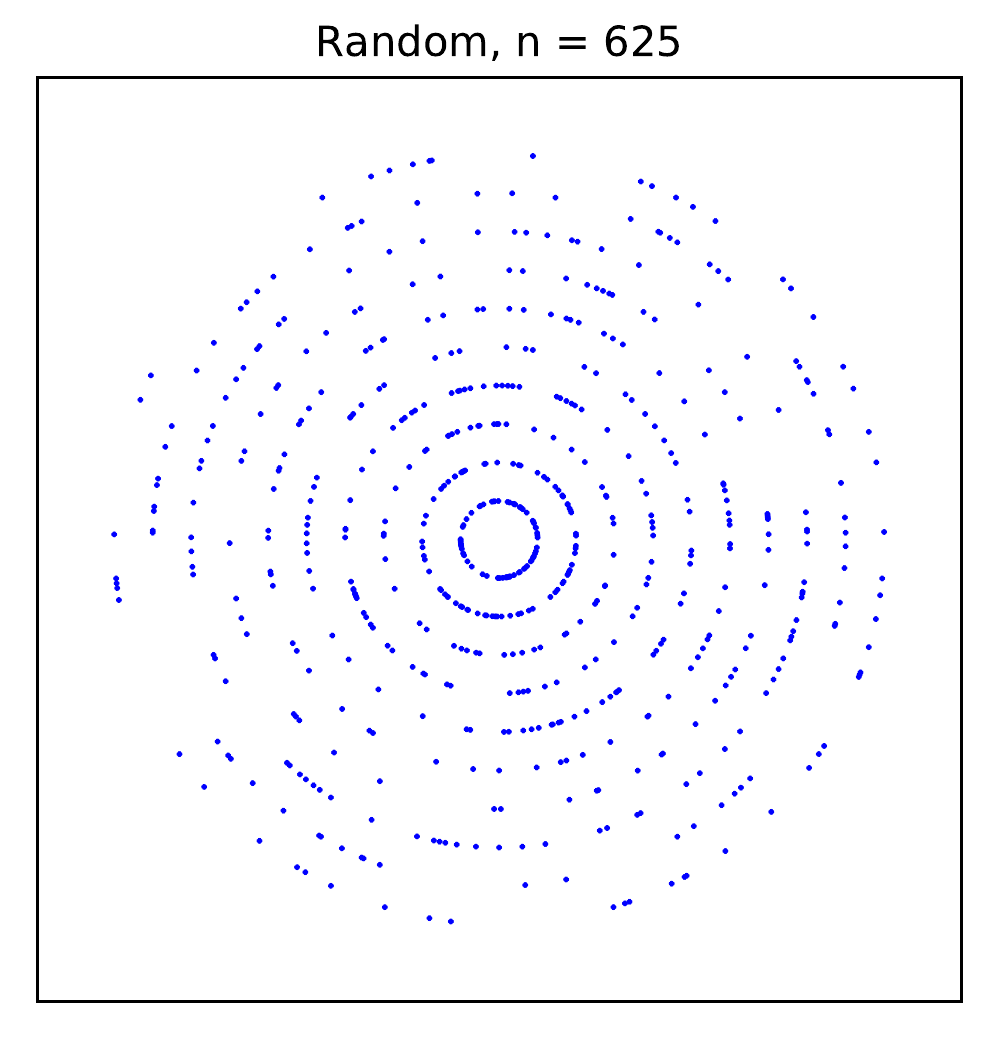}
			 \includegraphics[width=0.24\textwidth]{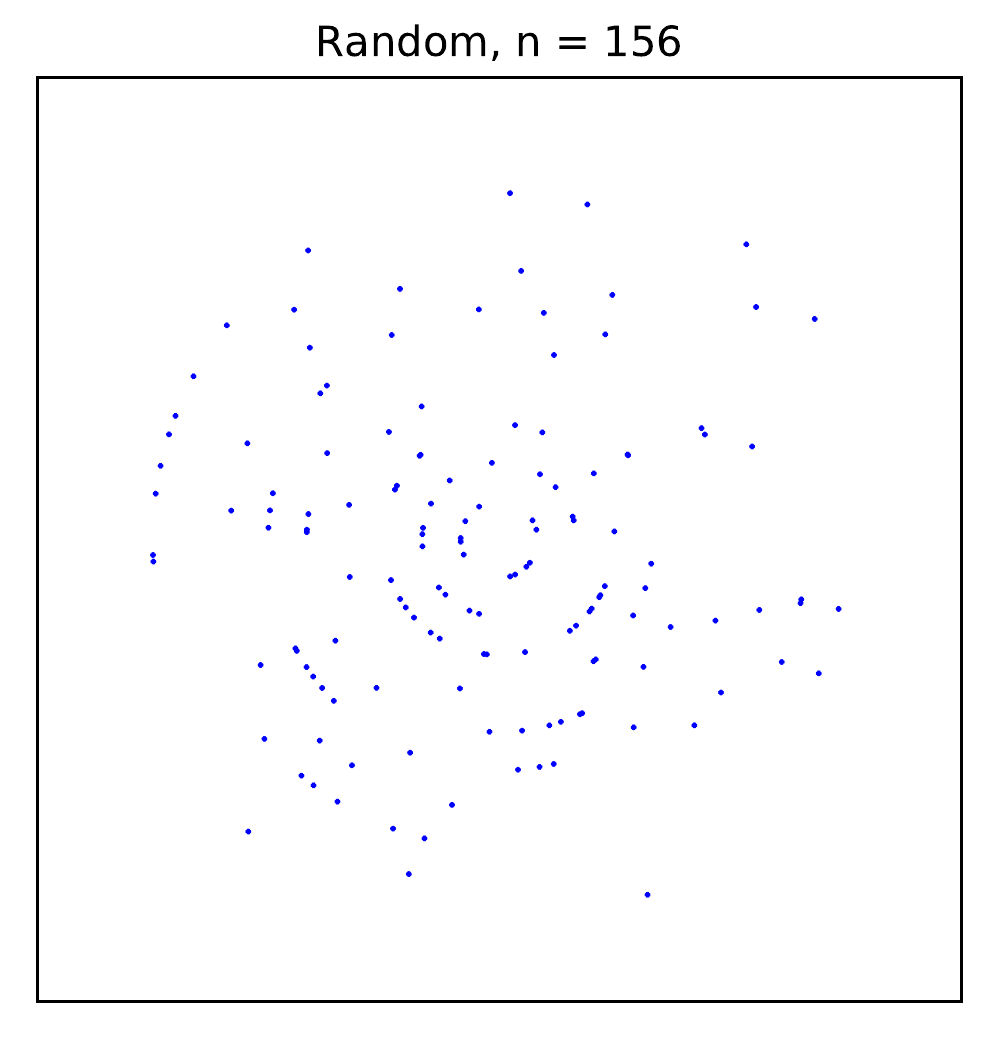}
			 \label{fig:approximations}
 	\end{figure}

\bigskip

\noindent \textbf{{\Large 4.}   \textsc{An even faster algorithm for discrepancy}.} We also propose and analyse an accelerated version of our method, where instead of maintaining the weights on all the $O(n^2)$ edges induced by $X$, we use further random sampling.
The following theorem describes the trade-off between the expected running time and discrepancy guarantee as function of a pre-sampling parameter $\alpha$.

\begin{theorem}\label{thm:presampled-disc-guarantee}
	Let $(X,\S)$ be a set system and  $d$ be a constant such that $\pi^*_{\S}(k) =O\round{ k^d}$. For any  $0 < \alpha \leq 1 $, there is a randomized algorithm which constructs a coloring $\chi$ of $X$  with expected discrepancy 
   \[
   O\round{\sqrt{  n^{1-\alpha/d}\ln m  + \ln^2 m \log n}},
   \]
   with at most
		\[
		O\round{
		  n^{1+\alpha(1+1/d)} \ln^2 n + mn^{\alpha/d}\ln (mn) \cdot \log n 
		}
		\]
   expected calls to the membership Oracle of $(X, \S)$.
\end{theorem}

\noindent
The randomized algorithm that achieves the guarantees of \Cref{thm:presampled-disc-guarantee} is presented in \Cref{sec:Presampling-proof} (\nameref{sampleddiscalgo}). It is essentially the algorithm  \nameref{discalgoDS} run on an initial random sample of edges with a small modification: to incorporate the pre-sampling step in the analysis, we need to recurse slightly more often (after $n/16$ steps instead of $n/4$).
The proof of \Cref{thm:presampled-disc-guarantee} relies on the following theorem on matchings in random edge-sets, which might be of independent interest.

 \begin{restatable}[]{theorem}{presamplingtradeoff}\label{lemma:presampling-guarantee}
	Let $(X,\S)$ be a set system with dual shatter function $\pi^*(k) = O(k^d)$, $ \alpha \in(0,1]$, and $\delta \in (0,1)$. Let $E$ be a uniform sample from $\binom{X}{2}$, where each edge is picked i.i.d. with probability
    \[
        p =  \min\curly{\frac{2\ln n}{n^{1-\alpha}} + \frac{4\ln(2/\delta)}{n^{2-\alpha}},~1}.
    \]
	Then with probability at least $1-\delta$, $E$ contains a matching of size $n/4$ with crossing number
    \[
        O\round{ n^{1- \alpha/d} + \ln |\S|}.
    \]

    Moreover, for any $d \geq 2$, and $n_0 \in \NN$ there is a set system $(X,\S)$ with $|X| = n \geq n_0$ and dual shatter function $\pi_\S^*(k) = O(k^d)$ such that for any $ \alpha \in(0,1]$ and $p(n) = o\round{n^{\alpha-1}}$ if $E$ is a random edge-set obtained by selecting each edge in $\binom{X}{2}$ i.i.d. with probability $p(n)$, then with probability at least $1/2$, every matching in $E$ of size $n/4$ has crossing number $\omega\round{n^{1-\alpha/d}}$ with respect to $\S$.
\end{restatable}
 
\bigskip


\noindent  \textbf{{\Large 5.} \textsc{Geometric systems}.} 
Set systems of bounded dual shatter function and bounded dual VC dimension arise naturally in many geometric
scenarios. Previous works on the above three problems---discrepancy, matchings approximations---heavily relied on spatial partitioning
techniques, which   essentially blocked any further progress and limited their practical applicability for the past   decades. We refer
the reader to \Cref{subsec:geometricsetsystems} for a detailed explanation.
The precise guarantees for several geometric set systems and their proofs are presented in \Cref{sec:corollaries}.

%
%
%

\section{Previous Results}\label{sec:prevres}

\subsection{Discrepancy}\label{sec:prevres-dualshat} 

A first bound on the combinatorial discrepancy of $(X,\S)$ follows immediately from Chernoff's bound,
which implies that a random two-coloring $\chi$
of $X$ satisfies $\disc_{ \S } \round{\chi} = O \big( \sqrt{n \ln m } \big)$ with probability at least $\tfrac{1}{2}$.
This also gives a randomized algorithm to obtain such a coloring, and
it is possible to derandomize the method yielding a deterministic algorithm with running time $O \left( n m \right)$ \citep{Chazelle:2000:DMR:507108}.

\cite{S85} showed that for any set system $(X,S)$, there exists a coloring of $X$
with discrepancy $O \left( \sqrt{n \ln(m/n)}  \right)$, which is tight and improves the general bound for $m = O \left( n\right)$. 
A series of algorithms for its construction started with the work of \cite{DBLP:conf/focs/Bansal10}, who gave
 the first polynomial-time randomized algorithm (using SDP rounding) to compute a coloring with discrepancy $O\round{\sqrt{n} \ln(m/n)}$, which matches the bound of Spencer for $m = O(n)$. Later \cite{DBLP:journals/siamcomp/LovettM15}
gave a combinatorial randomized algorithm for constructing colorings with discrepancy $O\round{\sqrt{n \ln(m/n)}}$ and improved the expected running time to  $\aO \round{ n^3 + m^3}$; see also~\cite{DBLP:journals/siamcomp/Rothvoss17} for a different proof. The algorithm of Bansal was de-randomized by \cite{BS13} (but still used a non-constructive method to prove the feasability of an underlying SDP), and later, \cite{DBLP:conf/ipco/LevyRR17} used the multiplicative weights update technique
to give a deterministic $O \left( n^4 m \right)$-time algorithm  
to compute a two-coloring with discrepancy $O \left( \sqrt{ n \ln (m/n) } \right)$ for
an arbitrary set system.
See also~\cite{DBLP:conf/stoc/BansalDGL18} for a random-walk algorithm
for Banaszczyk's discrepancy bound, with running time $O \left( n^{3.3728..} + nm^{2.3728..} \right)$ (the exponent depends
on the running time for matrix multiplication).

For general set systems, one cannot hope to have polynomial-time algorithms with better guarantees: it was shown by \cite{CNN11} that there exists a set system with $m=O(n)$ for which it is NP-hard to decide whether discrepancy zero or $\Omega(\sqrt{n})$.

Lastly, we mention that there is an active line of research considering sparse set systems\footnote{Where every element is contained in at most $t$ ranges for some constant $t$.} \citep{BF81,DBLP:conf/ipco/LevyRR17,BDS19} and the (stochastic) online setting \citep{spencer77,barany79,swan00,BS20,BJSS20, ALS21,BJMSS21, GGKKS22}.




\bigskip

Improved discrepancy bounds can be obtained if the set system satisfies additional constraints. In particular, we have the following result for set systems with polynomially bounded shatter function:


\begin{theorem}[\cite{ALON1999280,MWW93, mat97}] \label{thm:dualvcdim}
	Let $(X, \S)$ be a finite set system and $d$ be a constant such that $\pi^*_\S(k) = O\round{ k^d}$. Then
	there exists a polynomial-time algorithm to compute
	a two-coloring of $X$ with discrepancy $O \left( \sqrt{n^{1-1/d}\ln m} \right)$. 
	Furthermore, for any $d$, there exists a set system with dual shatter function $\pi^*_\S(k) = O\round{ k^d}$ such that any two-coloring of $X$ has discrepancy $\Omega \left( \sqrt{n^{1-1/d}\ln n} \right)$.
\end{theorem}
\noindent

  If $d$ is considered as a constant, the upper and lower bounds of \Cref{thm:dualvcdim} match:
	if the dual-shatter function of $(X,\S)$ is $\pi^*_\S(k) = O\round{ k^d}$, then $\vcdim(X,\S) \leq 2^d$ and thus by the Sauer-Shelah lemma, $ m = |\S| \leq\round{ \frac{en}{2^d}}^{2^d}$; see~\cite[Chapter 5]{MatDiscBook} for further details.


\paragraph{Algorithms.}



The classical proof of the upper-bound in \Cref{thm:dualvcdim} uses the multiplicative weights update (MWU) technique\footnote{For an excellent survey on the MWU technique see~\cite{AHK12}.} as follows.
The algorithm maintains a weight function $\pi$ on $\S$, with initial weights set to $1$. 
For any pair $\{x,y\} \in X$, let $\Delta_\S(x,y)$ denote the set of those sets $S \in \S$ which satisfy $\abs{S \cap \{x,y\}} =1$ and let   $\tilde\pi(x, y) = \sum_{S \in \Delta_\S(x, y) } {\pi}(S)$.
The algorithm colors two elements of $X$ at a time (for simplicity, we assume that $X$ is even) and proceeds as follows
\begin{enumerate}
	\item[for] $i = 1, \ldots, \frac{n}{2}$
	\item Find a pair  $\{x_i,y_i\} \in X$  that minimizes $\tilde\pi(x, y)$.
     \item Set $\chi(x_i) = \begin{cases} 1 &\text{ with probability } 1/2 \\ -1 &\text{ with probability } 1/2\end{cases}$
,  \quad $\chi \left( y_i \right) = -\chi \left( x_i \right)$,\\    
     	remove $x_i,y_i$ from $X$.
	\item Update $\pi$ by doubling the weight of each set in $\Delta_\S(x_i, y_i)$.
\end{enumerate}
The reweighing scheme ensures the key property that
for each $S\in \S$,
\begin{equation}\label{eq:CW-crnr-bound-first}
	\left| \curly{ i ~:~ S \in\Delta_\S(x_i, y_i) } \right| = O \round{n^{1-1/d}}.
\end{equation}
This implies, using Chernoff's bound and the union bound,
that the discrepancy of the resulting coloring is $ O\round{ \sqrt{ n^{1-1/d} \ln m}} $ with probability at least $1/2$~\cite{MWW93}.
The algorithmic bottleneck is finding the pair $\{x_i, y_i\}$ that minimizes $\tilde \pi$ at each iteration.
Using the incidence matrix for $\S$, this can be done in $O \left( n^2 m \right)$
steps, and thus the algorithm has overall running time $O \left( n^3 m \right)$. 

\subsection{Matchings with low crossing number.}

Given a set $X$, a \emph{matching} in $X$ is a set of disjoint edges (pairs) from $X$.  A \emph{(perfect) matching of $X$} is a matching of size $\floor{n/2}$ plus a loop (an edge $\{x,x\}$) if $n$ is odd. The \emph{size} of a matching is the number of its edges. We say that a range $S \in \S$ \emph{crosses} a pair $\{x,y\} $  if and only if $|S \cap \{x,y\}| =1$ and define the \emph{crossing number} of a matching $M$ with respect to $\S$ as the number of edges of $M$ crossed by a single range $S\in \S$.

Notice that the pairs $\curly{x_i, y_i}$ selected by the MWU algorithm form a perfect matching of $X$. 
Furthermore, the key property stated in \Cref{eq:CW-crnr-bound-first} can be simply formulated as: the matching $M = \curly{\curly{x_i, y_i}}_{i=1}^{n/2}$ has crossing number $O\round{n^{1-1/d}}$ with respect to $\S$. 
The study of perfect matchings (along with spanning paths and spanning trees) with low crossing number was originally introduced for geometric range searching \citep{Wel88,CW89}. Since then, they have found applications in various fields, for instance, discrepancy theory \citep{MWW93}, learning theory \citep{AMY16}, or algorithmic graph theory \citep{DHV20}. 

  The original method of Chazelle and Welzl builds a perfect matching using the multiplicative weight update (MWU) method. Briefly, the algorithm maintains a weight function $\pi$ on $\S$, with initial weights set to $1$. It selects edges iteratively, always choosing an edge that is guaranteed to be crossed by sets of low total weight in $\pi$; it then updates $\pi$ based on the chosen edge. The algorithmic bottleneck is in finding such an edge: for an abstract set system without additional structure, this  takes $O(n^2m)$ time for each of the $n/2$ iterations.

	Another approach    was proposed by \cite{harp09} (see also~\cite{FLM04}). His result implies that
	if $\kappa = \Theta (n^\gamma)$ for some $\gamma \in [1/\log n,1]$, then
	a spanning tree of crossing number $O(\kappa/\gamma)$ can be found by solving an LP on $\binom{n}{2}$ variables and $m+n$ constraints. 
	There also exists an algorithm using a general framework  of rounding fractional solutions of minimax integer programs with matroid constraints.
	This method gives a randomized algorithm that constructs a spanning tree  with expected crossing number
	at most $\kappa + O(\sqrt{\kappa\log m})$  in time $\aO ( mn^4 +n^8)$~\citep{CVZ09}.

\subsection{Geometric set systems.} \label{subsec:geometricsetsystems}

	Now we  turn to the case where $X$ is   a set of $n$ points in $\RR^d$  and $\S$ consists of   subsets of $X$ that are induced by  geometric objects.
	In this setting, improved bounds are made possible using spatial partitioning.
	The current-best algorithms for geometric set systems induced by half-spaces recursively construct
	simplicial partitions,
	stored in a hierarchical structure called the partition tree, which then at its base level gives a matching with low crossing number.
	This approach is used in the breakthrough result of \cite{Chan12}  who gave an  $O ( n \log n)$ time algorithm to build
	partition trees with respect to half-spaces in $\RR^d$, which then implies the same for computing matchings with crossing number $O(n^{1-1/d})$.



For set systems where $X$ is a set of $n$ points in $\RR^d$ and  $\S$ consists of subsets of $X$ that are induced by certain geometric objects, improved bounds are made possible using spatial partitioning.
For instance, if $\S$ consists of the subsets of $X$ that are induced by half-spaces, one can apply the algorithm of \cite{Chan12} to construct a perfect matching with crossing number $O\round{n^{1-1/d}}$ in time $O ( n \ln n)$, which then implies the same running-time for computing coloring with discrepancy $O\round{ \sqrt{ n^{1-1/d} \ln m}}$.
While the use of spatial partitioning gives $o(mn^3)$ running times,  progress remains blocked in several ways:
\begin{enumerate}
	\item[a)]  Spatial partitioning only exist in certain geometric settings; it is not possible when dealing with abstract set systems such as those arising in learning theory or graph theory.
	Indeed, as shown by \cite{AHW87}, they do not always exist in settings satisfying
	the requirements of \Cref{thm:dualvcdim} (e.g., the projective plane).

	\item[b)]  Optimal bounds for constructing simplicial partitions are only known for the case of half-spaces; this is one of the main problems left open by \cite{Chan12}. Despite a series  of research for semi-algebraic set systems (using linearization, cuttings, and more recently, polynomial partitioning~\citep{AMS13}), the   bounds are still sub-optimal for polynomials of degree larger than $2$, with  exponential dependence on the dimension.

	\item[c)] There are large constants in the asymptotic notation depending on the dimension $d$ both
	in the running time as well as the crossing number bounds, due to the use of cuttings (see~\cite{EHKS20}). 

	\item[d)] Practical implementation of spatial partitioning in $\RR^d$, $d > 2$,  remains an open problem in geometric computing, even for half-spaces. In particular,
	for $d > 2$, we know of no implementations for low-crossing matchings; nor for constructing $o\left(\frac{d}{\eps^2}\right)$-sized  $\eps$-approximations even for   half-spaces in $\RR^3$. 
	\begin{remark*}In $\mathbb R^2$, an algorithm to create optimal cuttings (the main tool in the construction of spacial partitioning for half-space ranges) was implemented~\cite{HarPe00Cutting} and was recently used for computing near-optimal $\eps$-approximations with respect to half-spaces in $\RR^2$~\citep{MP18}. 
	\end{remark*}
\end{enumerate}

\noindent
Thus, one of our main objectives was to find an efficient algorithm that \emph{does not} use spatial partitioning;  continuing the recent theme of such algorithms proposed for $\eps$-nets and $\eps$-approximations~\citep{V10, CGKS12, MDG17, M19}.

\section{Proofs}\label{sec:main-proof}

\subsection{Outline and ideas}

Our methods rest on the following three key ideas:

\begin{enumerate}
	\item We replace the bottleneck algorithmic step  of finding a minimum weight pair (with respect to $\tilde \pi$) in the multiplicative weights update technique (see \Cref{sec:prevres-dualshat})  by simply sampling a pair according to a carefully maintained distribution.   In particular, we maintain weights not only on the ranges in $\S$, but also on $\binom{X}{2}$. At each iteration we sample a range $S$ and an edge (pair) $e$ according to the current weights. Then we color the endpoints of $e$ and update the weights by \emph{doubling} the weight of each range that crosses $e$ and \emph{halving} the weight of each edge that is crossed by $S$.
	 In comparison to previous MWU-based solutions for constructing low-discrepancy colorings \citep{DBLP:conf/ipco/LevyRR17}, our method is much simpler and faster. 

	\item The idea of maintaining `primal-dual' weights has been used earlier to approximately solve matrix  games \citep{GK95} and in geometric optimization~\citep{AP14}. In our case, the process is more elaborate as
	we are constructing a proper coloring at the same time as reweighing. Therefore, at the end of each iteration, as we color the endpoints $e$, we are forced to set the weights of $e$ and all edges adjacent to $e$ to $0$. This  breaks down the reweighing scheme, as the removal of the edges  amplifies the error introduced in later iterations and thus our maintained weights degrade over time. However, we prove that restarting the algorithm by `resetting' all the weights a logarithmic number of times suffices to ensure the required low discrepancy.

	\item Finally, updating the weights of \emph{all} edges and sets crossing the randomly picked set and edge  would be   too expensive. Instead, we show that  updating the weights of a \emph{uniform} sample of  
	edges and 
	ranges at each iteration is sufficient for our purposes. The key observation here is that the standard multiplicative weights proof has an additive smaller-order term; we take advantage of this gap to improve the running time at the cost of amplifying this term, just enough so that it is still within a constant factor of the optimal solution.
\end{enumerate}

\noindent
We start by proving the main technical ingredient of this work: \Cref{thm:main-matching-result}.

\subsection{Proof of Theorem \ref{thm:main-matching-result}}\label{subsec:matching-proof}


The algorithm that achieves the guarantees of \Cref{thm:main-matching-result} is   presented in \Cref{algo:main}.

\begin{algorithm}[ht!]
\algotitle{\textsc{BuildMatching}}{mainalgo}
\caption{  \textsc{BuildMatching}$\big((X,\S),  a, b, \gamma\big)$}
\label{algo:main}

	$M \leftarrow \emptyset$

\While{ $|X| \geq 4$ }{

	$E \leftarrow \binom{X}{2}$

	$M \leftarrow M \cup \nameref{subalgo} \big((X,\S), E,  a, b, \gamma, |X|/4 \big)$ \tcp*{see \Cref{algo:match-half}}

	$ X \leftarrow X \setminus \mathrm{endpoints}\round{M}$\tcp*{remove the elements covered by $M$}

	}

	match the remaining elements of $X$ randomly and add the edges to $M$

	\tcp*{if $|X|$ is odd, we allow one edge to be a loop}

	\textbf{return} $M$

\end{algorithm}
\begin{algorithm}[ht!]
\algotitle{\textsc{PartialMatching}}{subalgo}
\caption{  \textsc{PartialMatching}$\big((X,\S), E,  a, b, \gamma, t \big)$}
	\label{algo:match-half}
	$\omega_1(e) \leftarrow 1, \quad \pi_1(S) \leftarrow 1 \quad\quad\forall e \in E,~ S \in \S$\
\BlankLine

		$\proba \leftarrow \min\curly{\frac{\coeffPU\ln(|E| \cdot t)}{a|X|^\gamma +b}  ,\, 1}$

		$\probaS \leftarrow \min\curly{\frac{\coeffQU\ln(|\S|\cdot t)}{a|X|^\gamma +b} ,\, 1}$

	\BlankLine
	\For{$i = 1, \dots, t$ }{

\BlankLine

		$\omega_{i}(E) \leftarrow \sum_{e \in E}\omega_i (e)$

		$\pi_{i}(\S) \leftarrow \sum_{S \in \S}\pi_i (S)$

		choose $e_i \sim \omega_i$
		\tcp*{$\PP[ e_i =e ] = \frac{\omega_i(e)}{\omega_i(E)} ~~ \forall e \in E$}

		choose $S_i \sim \pi_i$
		\tcp*{$\PP[ S_i =S ] = \frac{\pi_i(S)}{\pi_i(\S)} ~~ \forall S \in \S$}

		\BlankLine

		$E_i \leftarrow $ sample from $E$ with probability $\proba$
		\tcp*{$\PP [ e \in E_i ] = \proba ~~ \forall e \in E$}

		$\S_i \leftarrow $ sample from $\S$ with probability $\probaS$
		\tcp*{$\PP [ S \in \S_i ] = \probaS ~~ \forall S \in \S$}
		\tcp*{$\Inc(e, S) = 1$ if $e$ crosses $S$, $\Inc(e, S) = 0$ otherwise}
		\For{ $e \in E_i$ } {
			$\omega_{i+1}(e) \leftarrow \omega_i (e) \big(1 - \frac 1 2  \Inc(e,S_i)\big)$
			\tcp*{halve weight if $S_i$ crosses $e$}
		}
		\For{ $S \in \S_i$}{
			$\pi_{i+1}(S)  \leftarrow \pi_i (S) \big(1 +   \Inc(e_i,S)\big)$
			\tcp*{double weight if $S$ crosses $e_i$}
		}

		\BlankLine

		set the weight in $\omega_{i+1}$ of $e_i$ and of each edge adjacent to $e_i$ to zero

	}

	\textbf{return} $\{ e_1, \dots, e_{t} \}$
\end{algorithm}

\newpage
\noindent
The main result of this section is the following theorem:
\begin{theorem}\label{thm:general-case-half-pts}
	Let $(X,\S)$ be a set system, which satisfies \nameref{assumption} and let $E$ denote the set of all pairs (edges) from $X$.
	Then
	\nameref{subalgo} $\round{ (X, \S), E, a, b, \gamma, n/4 }$ returns a matching of size\footnote{The size of a matching is the number of its edges.} $n/4$ with expected crossing number at most
		\[
		a\round{\frac{n}{2}}^\gamma + b
		+
		\max\curly{\frac{a\round{\frac{n}{2}}^\gamma +b}{2} ,~18\ln(mn/4)},
		\]
		with an expected number of Oracle calls at most
		\[
			\min\curly{ \frac 6 a \round{ n^{3-\gamma} \ln\frac{n^3}{4} + 3 mn^{1-\gamma} \ln\frac{mn}{4}},~ \frac{n^3 + 2mn}{8} }.
		\]
		
\end{theorem}

\noindent
Before we present its proof, we first show how \Cref{thm:general-case-half-pts} implies \Cref{thm:main-matching-result}.
The algorithm \nameref{mainalgo} makes $\log n$ calls to \nameref{subalgo} with exponentially decreasing input sizes. In particular, the overall expected number of membership Oracle calls of \nameref{mainalgo} can be bounded as
\begin{align*}
&\sum_{i = 0}^{\log n} \min\curly{ \frac 6 a \round{ \round{\frac{n}{2^i}}^{3-\gamma} \ln\frac{n^3}{2^{3i+2}} + 3 m\round{\frac{n}{2^i}}^{1-\gamma} \ln\frac{mn}{2^{i+2}}},~ \frac{\round{\frac{n}{2^i}}^3 + {\frac{mn}{2^{i-1}}}}{8} }\\
&\leq
 \min\curly{ \sum_{i = 0}^{\log n}\frac 6 a \round{ \round{\frac{n^{3-\gamma}}{4^i}} \ln\frac{n^3}{2^{3i+2}} + 3 m\round{\frac{n}{2^i}}^{1-\gamma} \ln\frac{mn}{2^{i+2}}},~ \sum_{i = 0}^{\log n}\frac{\frac{n^3}{8^i} + {\frac{2mn}{2^{i}}}}{8}}\\
 &\leq
 \min\curly{ \frac 6 a \round{ \frac{4}{3} n^{3-\gamma} \ln n^3 + 3m\cdot\min\curly{\frac{2n^{1-\gamma}}{1-\gamma},~ n^{1-\gamma}\log n}  \ln mn},~ \frac{n^3}{7} + \frac{mn}{2}}.
\end{align*}
As for the crossing number, \Cref{thm:general-case-half-pts} implies that \nameref{mainalgo}$\big((X,\S),  a, b, \gamma\big)$ returns a matching with expected crossing number at most
\begin{align*}
    \sum\limits_{i=1}^{\log n  } \bracket{\frac{3a}{2}\round{\frac{n}{2^i}}^\gamma +\frac{3b}{2} + 18 \ln \frac{mn}{2^{i+1}}}
    &<
    \frac{3an^\gamma}{2}  \bracket{\sum\limits_{i=1}^{\infty} \round{\frac{1}{2^\gamma}}^i} +
    \round{\frac{3b}{2} + 18 \ln \round{mn}}\log n \\
    &<
    \frac{3a}{\gamma}  n^\gamma + \round{\frac{3b}{2} + 18 \ln \round{mn}}\log n \,.
\end{align*}
Hence, we have shown that \Cref{thm:main-matching-result} is a consequence of \Cref{thm:general-case-half-pts}. \qed

\subsubsection*{Proof of \Cref{thm:general-case-half-pts}.} For an edge $e$ and a set $S$, we define 
\[
\Inc(e,S) = \begin{cases} 1 &\text{if $S$ crosses $e$}, \\ 0 &\text{otherwise.} \end{cases}
\]
We will deduce \Cref{thm:general-case-half-pts} from the next lemma, which is proved later in this section.

\begin{restatable}{lemma}{mainlemma}\label{lemma:MW-bound}
	Let $t \in [1, |X|/4]$ be an integer and let $\{e_1, \dots, e_t\}$, $\{S_1, \dots, S_t\}$, $\proba$, and $\probaS$ as in \Cref{algo:match-half}. Furthermore let $\tilde E_t \subseteq E$ denote the set of edges that have non-zero weight when
	 \nameref{subalgo}$\big((X,\S), E,  a,b,\gamma, t \big)$ terminates. Then
	\begin{equation}\label{eq:MW-bound}
		\EE\bracket{
			\max\limits_{S \in \S} \sum\limits_{i=1}^{t} \Inc(e_i,S)}
		\leq
\frac 1 2\EE\bracket{\min\limits_{e\in \tilde E_t} \sum\limits_{i=1}^{t} \Inc(e,S_i)}
		+
		\frac{\round{\frac{128}{13} + \frac{8}{\ln 2}}\ln(|E|t)}{3\proba}
		+
		\frac{\round{16 + \frac{4}{\ln 2}}\ln(|\S|t)}{3\probaS }.
	\end{equation}
\end{restatable}


\noindent
Setting $t=n/4$, the left-hand side of \Cref{eq:MW-bound} is precisely the expected crossing number of the  edges $\{e_1, \dots, e_{n/4}\}$ returned by \nameref{subalgo}$\round{ (X, \S), E, a, b, \gamma, n/4 }$.
To bound the expectation in the right-hand side of \Cref{eq:MW-bound}, we use the following lemma.

\begin{lemma}
	\label{lemma:robust-short-edge}
	Let $(Y,\R)$ be a set system, $w: \R \to \RR_{\geq 0}$, and $\kappa$ be such that $Y$ has a perfect matching with crossing number at most $\kappa$ with respect to $\R$. Then there is an edge $\{x,y\}$ in $\binom{Y}{2}$ such that
	\[
		\sum\limits_{R \text{\emph{ crosses }} \{x,y\}} w(R) \leq \frac{2w(\R)\cdot \kappa}{|Y|}.
	\]
\end{lemma}
\begin{proof}
	Let $M$ be a perfect matching of $Y$ such that any set of $\R$ crosses  at most $\kappa$ edges of $M$. Then if we consider the weighted sum
	there are at most $w(\R) \cdot \kappa$ crossings  between the edges of $M$ and sets in $\R$ counted with weights. By the pigeonhole principle, there is an edge in $M$ that is crossed by sets of total weight at most
	\[
    \frac{w(\R) \cdot \kappa }{ |M|} =
	\frac{w(\R) \cdot \kappa }{|Y|/2} = \frac{ 2 w(\R) \kappa }{|Y|}
	\]
	sets of $\R$.
\end{proof}
Let $\tilde X_{n/4} \subset X$ denote the set of points that are not covered by the edges $\{e_1, \dots, e_{n/4}\}$. 
Note that $|\tilde{X}_{n/4}| = n/2$ and that since \nameref{mainalgo} calls \nameref{subalgo} with $E = \binom{X}{2}$, we have  $\tilde E_{n/4}  = \binom{\tilde{X}_{n/4}}{2}$. Moreover, since $(X,\S)$ satisfies \nameref{assumption}, there exists a perfect matching of $\tilde X_{n/4}$ with crossing number at most $a|{\tilde X_{n/4}}|^\gamma + b$. 
Applying Lemma \ref{lemma:robust-short-edge} to $ Y = \tilde X_{n/4}$ and  $\R = \{S_1, \dots, S_{n/4}\}$ with weights $w(S_i) = 1$, we get that there is an edge $e \in \tilde E_{n/4}$ that satisfies
\begin{equation}\label{eq:short-edge-withprobas}
	\sum\limits_{i=1}^{n/4} \Inc(e,S_i)
	\leq
    \frac{2\cdot \frac n 4  \round{a|\tilde X_{n/4}|^\gamma + b}}{|\tilde X_{n/4}|}
    = 2\cdot \frac n 4 \cdot \frac{ a(n/2)^\gamma + b}{n/2}   
    =
    a\round{\frac n 2 }^\gamma + b
    .
\end{equation}
Since \Cref{eq:short-edge-withprobas} holds for any choice of $\{S_1, \dots, S_{n/4}\}$ and $\tilde X_{n/4}$, we can conclude that
\begin{equation}\label{eq:expected-shortedge-inrest}
\EE\bracket{\min_{e \in \tilde E_{n/4}}\sum\limits_{i=1}^{n/4} \Inc(e,S_i)}
\leq 
a\round{\frac n 2 }^\gamma + b.
\end{equation}

\noindent
Now Equations \eqref{eq:MW-bound} and \eqref{eq:expected-shortedge-inrest} imply that the expected crossing number of the edges returned by \nameref{subalgo}$\round{ (X, \S), E, a, b, \gamma, n/4 }$ can be bounded as
\begin{align*}
	&\EE\bracket{
			\max\limits_{S \in \S} \sum\limits_{i=1}^{t} \Inc(e_i,S)}
		\leq
		a\round{\frac n 2}^\gamma + b
		+
		\frac{\round{\frac{128}{13} + \frac{8}{\ln 2}}\ln(|E|\cdot \frac n 4)}{3\proba}
		+
		\frac{\round{16 + \frac{4}{\ln 2}}\ln(|\S|\cdot \frac n 4)}{3\probaS }\\
		&~\leq
		a\round{\frac n 2}^\gamma + b
		+
		\frac{\round{\frac{128}{13} + \frac{8}{\ln 2}}\ln\round{\frac{|E|n}{4}}}{3\min\curly{\frac{\coeffPU}{an^\gamma +b}  \cdot \ln\round{\frac{|E|n}{4}},\, 1}}
		+
		\frac{\round{16 + \frac{4}{\ln 2}}\ln\round{\frac{|\S|n}{4}}}{3\min\curly{\frac{\coeffQU}{an^\gamma +b}  \cdot \ln\round{\frac{|\S|n}{4}},\, 1} }\\
		&~\leq
		a\round{\frac{n}{2}}^\gamma + b
		+
		\frac{\round{128 + \frac{104}{\ln 2}} \max\curly{\frac{an^\gamma +b}{\coeffPU}  ,~\ln\round{\frac{|E|n}{4}}}}{39} 
		+\frac{\round{16 + \frac{4}{\ln 2}} \max\curly{\frac{an^\gamma +b}{\coeffQU},~\ln\round{\frac{|\S|n}{4}}}}{3}  \\
		&~\leq
		a\round{\frac{n}{2}}^\gamma + b
		+
		\frac{\round{128 + \frac{104}{\ln 2}} \max\curly{\frac{an^\gamma +b}{\coeffPU} ,~\frac{3}{2}\ln\round{\frac{|\S|n}{4}}}}{39}
		+
		\frac{\round{16 + \frac{4}{\ln 2}}\max\curly{\frac{an^\gamma +b}{\coeffQU}  ,~\ln\round{\frac{|\S|n}{4}}} }{3}\\
		&=
		a\round{\frac{n}{2}}^\gamma + b
		+
		\round{\frac{128 + \frac{104}{\ln 2}}{26} + \frac{16 + \frac{4}{\ln 2} }{3}}\max\curly{\frac{an^\gamma +b}{\coeffQU} ,~\ln\round{\frac{|\S|n}{4}}}\\
		&\leq
		a\round{\frac{n}{2}}^\gamma + b
		+
		\max\curly{\frac{an^\gamma +b}{4} ,~18\ln\round{\frac{|\S|n}{4}}}.
\end{align*}

\noindent
Finally, we bound the number of membership Oracle calls. At each iteration $i=1, \dots, n/4$, we update the weights of at most $ \frac{n^2}{2} \proba +  m\probaS $ elements in expectation, each requiring one call to the membership Oracle.
Thus in expectation, the total number of membership Oracle calls is at most
\begin{align*}
	 \frac{n}{4} &\round{ \frac{n^2}{2}  \min\curly{\frac{\coeffPU \ln\frac{n^3}{4}}{an^\gamma + b},~ 1} + m  \min\curly{\frac{\coeffQU \ln\frac{mn}{4}}{an^\gamma + b},~ 1} } \\
	 &\leq \min\curly{ \frac 6 a \round{ n^{3-\gamma} \ln\frac{n^3}{4} + 3 mn^{1-\gamma} \ln\frac{mn}{4}},~ \frac{n^3 + 2mn}{8} }.
\end{align*}
Thus, we conclude that \Cref{thm:general-case-half-pts} is implied by \Cref{lemma:MW-bound}.
\qed

\subsubsection*{Proof of \Cref{lemma:MW-bound}.}\label{sec:proof-of-MW-lemma}
The proof is subdivided into three lemmas.
The first lemma is proved by examining the total weight of the sets of $\S$ in $\pi_{t+1}$.

\begin{lemma}\label{lemma:hyperplane-weights}
	\begin{equation*}\label{eq:hyperplane-weights}
		\EE\bracket{
			\max\limits_{S \in \S} \sum\limits_{i=1}^{t} \Inc(e_i,S)}
		\leq
		\frac{4}{3\ln 2}\cdot \sum\limits_{i=1}^t \EE\bracket{ \sum\limits_{S \in \S} \frac{\pi_{i} (S)}{\pi_{i}(\S)}   \Inc(e_i,S)}
		+
		\frac{\round{16 + \frac{4}{\ln 2}} \ln(|\S|t)}{3\probaS}
	\end{equation*}
\end{lemma}
\begin{proof}
	Let $\pi_{t+1} (\S)$ denote the total weight of the sets of $\S$ in $\pi_{t+1}$. We bound $\pi_{t+1} (\S)$ in two different ways. On the one hand, $\pi_{t+1}(\S)$ is clearly lower-bounded by the weight of the set of maximum weight in $\pi_{t+1}$.  Recall that the weight of a set $S$ is doubled in iteration $i$ if and only if $S \in \S_i$ \emph{and} $S$ crosses $e_i$, therefore
	\begin{equation*}
		\pi_{t+1}(\S)
		\geq
		\max\limits_{S \in \S}  \pi_{t+1}(S)
		=
		2^{\max\limits_{S \in \S} \sum\limits_{i=1}^{t} \Inc(e_i,S)
			\cdot \bm{1}_{\{S \in \S_i\}} },
	\end{equation*}
	where $\bm{1}_{\A}$ denotes the indicator random variable of the event $\A$.
	On the other hand, we can express $\pi_{t+1}(\S)$ using the update rule of the algorithm
	\begin{align*}
		\pi_{t+1}(\S)
		&=
		\sum\limits_{S \in \S} \pi_{t+1}(S)
		=
		\sum\limits_{S \in \S} \pi_{t} (S) \round{1 +   \Inc(e_t,S)\cdot \bm{1}_{\{S \in \S_t\}}}\\
		&=
		\sum\limits_{S \in \S} \pi_{t} (S) + \sum\limits_{S \in \S}   \pi_{t} (S)   \Inc(e_t,S)\cdot \bm{1}_{\{S \in \S_t\}}\\
		&=
		\pi_t(\S) +  \pi_{t}(\S)   \sum\limits_{S \in \S} \frac{\pi_{t} (S)}{\pi_{t}(\S)}   \Inc(e_t,S)\cdot \bm{1}_{\{S \in \S_t\}}\\
		&=
		\pi_{t}(\S) \round{ 1+  \sum\limits_{S \in \S} \frac{\pi_{t} (S)}{\pi_{t}(\S)}   \Inc(e_t,S)\cdot \bm{1}_{\{S \in \S_t\}}}.
	\end{align*}
	Unfolding this recursion and using the fact that $1 + a \leq \exp(a)$, we get
	\begin{align*}
		\pi_{t+1}(\S)
		&=
		\pi_1(\S) \prod\limits_{i=1}^{t} \round{ 1+  \sum\limits_{S \in \S} \frac{\pi_{i} (S)}{\pi_{i}(\S)}   \Inc(e_i,S)\cdot \bm{1}_{\{S \in \S_i\}}}\\
		&\leq
		|\S| \cdot \exp \round{  \sum\limits_{i=1}^{t} \sum\limits_{S \in \S} \frac{\pi_{i} (S)}{\pi_{i}(\S)}   \Inc(e_i,S)\cdot \bm{1}_{\{S \in \S_i\}}}.
	\end{align*}
	Putting together the obtained  upper and  lower bounds on $\pi_{t+1}(\S)$, we get
	\[
	2^{\max\limits_{S \in \S} \sum\limits_{i=1}^{t} \Inc(e_i,S)
		\cdot \bm{1}_{\{S \in \S_i\}} }
	\leq
	|\S| \cdot \exp \round{  \sum\limits_{i=1}^{t} \sum\limits_{S \in \S} \frac{\pi_{i} (S)}{\pi_{i}(\S)}   \Inc(e_i,S)\cdot \bm{1}_{\{S \in \S_i\}}} .
	\]
	Taking the logarithm of each side yields
	\begin{equation}\label{eq:hplane-weights-after-log}
		\ln(2) \cdot \max\limits_{S \in \S} \sum\limits_{i=1}^{t} \Inc(e_i,S) \cdot \bm{1}_{\{S \in \S_i\}}
		\leq
		\sum\limits_{i=1}^{t} \sum\limits_{S \in \S} \frac{\pi_{i} (S)}{\pi_{i}(\S)}   \Inc(e_i,S)\cdot \bm{1}_{\{S \in \S_i\}}
		+ \ln |\S| \,.
	\end{equation}
	If $\probaS = 1$, then $\bm 1_{\{S \in \S_i\}} = 1$ for all $i$ and  $S\in \S$, thus taking total expectation we conclude
	\begin{equation*}
		\EE\bracket{
			\max\limits_{S \in \S} \sum\limits_{i=1}^{t} \Inc(e_i,S)}
		\leq
		\frac{1}{\ln 2}\sum\limits_{i=1}^t \EE\bracket{ \sum\limits_{S \in \S} \frac{\pi_{i} (S)}{\pi_{i}(\S)}   \Inc(e_i,S)}
		+
		\frac{\ln|\S|}{\ln 2}.
	\end{equation*}

	\newcommand\epsprobaS{\frac{3\probaS}{4}} 
	\newcommand\pepsS{4} 
	\newcommand\onemepsS{\frac 3 4} 

	\noindent
	Assume that $\probaS < 1$.
		Since $\max f(x) - \max g(x) \leq \max (f(x) - g(x))$, \Cref{eq:hplane-weights-after-log} implies
	\begin{align*}
		 \ln(2) \cdot \onemepsS \cdot \max\limits_{S \in \S} \sum\limits_{i=1}^{t} \Inc(e_i,S)\cdot \probaS
		\leq
		&\ln(2) \cdot \max\limits_{S \in \S} \sum\limits_{i=1}^{t} \Inc(e_i,S)\cdot\round{\epsprobaS - \bm 1_{\{S \in \S_i\}} }
		\\
		&+\sum\limits_{i=1}^{t} \sum\limits_{S \in \S} \frac{\pi_{i} (S)}{\pi_{i}(\S)}   \Inc(e_i,S)\cdot \bm{1}_{\{S \in \S_i\}}
		+ \ln |\S| \,.
	\end{align*}
	Taking total expectation of each side, we obtain
	\begin{equation}\label{eq:plane-wt-after-exp-and-proba}
	\begin{aligned}
		\onemepsS \ln(2) \cdot\EE\bracket{ \max\limits_{S \in \S}\sum\limits_{i=1}^{t} \Inc(e_i,S)\cdot \probaS}
		\leq
		&\ln(2) \cdot \EE\bracket{\max\limits_{S \in \S} \sum\limits_{i=1}^{t} \Inc(e_i,S)\cdot\round{\epsprobaS - \bm 1_{\{S \in \S_i\}} }}
		\\
		&+\sum\limits_{i=1}^{t} \sum\limits_{S \in \S} \EE\bracket{\frac{\pi_{i} (S)}{\pi_{i}(\S)}   \Inc(e_i,S)\cdot \bm{1}_{\{S \in \S_i\}}}
		+ \ln |\S| \,.
	\end{aligned}
	\end{equation}


	Observe that  for each fixed $i$, the random variables $\{\pi_i,e_i\}$ and $\S_i$ are independent, thus
	\begin{equation}\label{eq:probaS-linearity-of-exp}
		\sum\limits_{i=1}^{t} \sum\limits_{S \in \S} \EE\bracket{\frac{\pi_{i} (S)}{\pi_{i}(\S)}   \Inc(e_i,S)\cdot \bm{1}_{\{S \in \S_i\}}}
		=
		 \probaS \cdot\sum\limits_{i=1}^{t} \sum\limits_{S \in \S} \EE\bracket{\frac{\pi_{i} (S)}{\pi_{i}(\S)}   \Inc(e_i,S)}.
	\end{equation}

	To bound the expectation of $\max\limits_{S \in \S} \sum\limits_{i=1}^{t} \Inc(e_i,S)\cdot\round{\epsprobaS - \bm 1_{\{S \in \S_i\}} }$, we will need the following Azuma-type inequality for martingales.

	\begin{lemma}[{\cite[Lemma 10]{KY14}}]\label{lemma:Azuma}
		Let $X = \sum_{i=1}^T x_i$ and $Y = \sum_{i=1}^T y_i$  be sums of non-negative random variables, where $T$ is a random stopping time with finite expectation, and, for all $i$, $|x_i - y_i | < 1$ and
		\[
			\EE \bracket{ x_i - y_i ~\bigg|~ \sum_{s < i} x_s, \sum_{s < i} y_s} \leq 0.
		\]
		Let $\eps \in [0,1]$ and $A \in \RR$, then
		\[
			\PP\bracket{ (1-\eps) X \geq Y+A } \leq \exp (-\eps A).
		\]
	\end{lemma}

	\begin{claim}\label{claim:martingale-bound}
	\[
		 \PP\bracket{\max\limits_{S \in \S} \sum\limits_{i=1}^{t} \Inc(e_i,S)\cdot\round{\epsprobaS - \bm 1_{\{S \in \S_i\}} }
		\geq 3 \ln(|\S|t)} \leq  \frac{1}{t}.
	\]
	\end{claim}

	\begin{proof}
		For each $i\in [1,t]$ and $S \in \S$, consider the random variables $x_i(S) = \Inc(e_i,S)\cdot \probaS$ and $y_i(S) = \Inc(e_i,S)\cdot\bm 1_{\{S \in \S_i\}} $, which are measurable with respect to $e_i$ and $\S_i$. For any $i$ and $S \in \S$, we have $|x_i(S) - y_i(S) | \leq 1$.
		 Since $\S_i$ is independent of $e_i$, $\sum_{k < i} x_k(S)$, and $\sum_{k < i} y_k(S)$, we have
			$$
				\EE\bracket{x_i(S) - y_i(S) ~\bigg|~ \sum_{k < i} x_k(S), \sum_{k < i} y_k(S)} = 0
			$$
		as $\EE\bracket{\probaS - \bm 1_{\{S \in \S_i\}}} =0$ for all $i \in [1,t]$ and $S \in \S$.

		\noindent
		Therefore, \Cref{lemma:Azuma} with $\eps = 1/\pepsS$, combined with the union bound implies for any $A \in \RR$,
			\begin{align*}
				 \PP\left(\max\limits_{S \in \S} \sum\limits_{i=1}^{t} \Inc(e_i,S)\cdot\round{\epsprobaS - \bm 1_{\{S \in \S_i\}} } \geq A\right) \leq |\S|\exp\left(-\frac{A}{\pepsS}\right).
			\end{align*}
			Setting $A =  \pepsS\ln(|\S|t)$, we conclude the proof of \Cref{claim:martingale-bound}.
	\end{proof}

	\noindent
	Applying \Cref{claim:martingale-bound} and using that $\sum_{i=1}^{t} \Inc(e_i,S)\cdot\round{\epsprobaS - \bm 1_{\{S \in \S_i\}}} \leq t$ always holds, we get

	\begin{equation}\label{eq:probaS-bounding-difference}
		\EE\bracket{\max\limits_{S \in \S} \sum\limits_{i=1}^{t} \Inc(e_i,S)\cdot\round{\epsprobaS - \bm 1_{\{S \in \S_i\}} }}
		\leq
		 \pepsS\ln(|\S|t)
		 +
		 t \cdot \frac 1 t
		 \leq
		 \pepsS \ln(|\S|t) + 1.
		 \end{equation}
	Hence Equations \eqref{eq:plane-wt-after-exp-and-proba}, \eqref{eq:probaS-linearity-of-exp}, and \eqref{eq:probaS-bounding-difference} imply
	\begin{align*}
	&\frac{3 \ln 2}{4}
	 \probaS \cdot
	\EE\bracket{
		\max\limits_{S \in \S} \sum\limits_{i=1}^{t} \Inc(e_i,S)}\leq
	 \sum\limits_{i=1}^t \probaS\cdot\EE\bracket{ \sum\limits_{S \in \S} \frac{\pi_{i} (S)}{\pi_{i}(\S)}   \Inc(e_i,S)}
	+ \ln(2) \cdot\round{ \pepsS \ln(|\S|t)+ 1}
	+
	\ln |\S|.
	\end{align*}
	Dividing both sides by $\frac{\probaS\cdot 3 \ln 2}{4}$, we obtain 
	\begin{align*}
	\EE\bracket{
		\max\limits_{S \in \S} \sum\limits_{i=1}^{t} \Inc(e_i,S)}
	&\leq
	\frac{4}{3\ln 2}\cdot \sum\limits_{i=1}^t \EE\bracket{\sum\limits_{S \in \S} \frac{\pi_{i} (S)}{\pi_{i}(\S)}   \Inc(e_i,S)}
	+
	\frac{16 \ln(|\S|t)+ 4 	+4\frac{\ln |\S|}{\ln 2}}{3\probaS }\\
	 &\leq
	\frac{4}{3\ln 2}\cdot \sum\limits_{i=1}^t \EE\bracket{\sum\limits_{S \in \S} \frac{\pi_{i} (S)}{\pi_{i}(\S)}   \Inc(e_i,S)}
	+
	\frac{\round{ 16 + \frac{4}{\ln 2}} \ln(|\S|t)}{3\probaS }.
		\end{align*}
	This concludes the proof of \Cref{lemma:hyperplane-weights}.
	\end{proof}

	\newcommand\epsproba{\frac{3\proba}{16}} 
	\newcommand\peps{\frac{16}{13}} 

The next lemma is proven by applying analogous arguments for the total weight of edges in $\omega_{t+1}$ with a small adjustment as in each iteration we set some edge weights to zero. Recall that $\tilde E_t$ denotes the set of edges that have non-zero when \nameref{subalgo}$\round{(X,\S),E,a,b,\gamma,t}$ terminates, in other words, $\tilde E_t$ is the set of edges that have non-zero weight in $\omega_{t+1}$.
\begin{lemma}\label{lemma:edge-weights}
	\begin{equation*}
		\sum\limits_{i=1}^{t} \sum\limits_{e \in E} \EE\bracket{\frac{\omega_{i} (e)}{\omega_{i}(E)}   \Inc(e,S_i)}
	< \frac{3\ln 2}{8}  \cdot \EE\bracket{\min\limits_{e\in \tilde E_t} \sum\limits_{i=1}^{t} \Inc(e,S_i)}
	+   \frac{\round{\frac{32\ln 2}{13} + 2}  \ln(|E|t) }{\proba}.
	\end{equation*}
\end{lemma}

\begin{proof}
	Let $\omega_{t+1} (E)$ denote the total weight of edges in $\omega_{t+1}$. Again, we lower-bound $\omega_{t+1} (E)$ by the largest edge-weight in $\omega_{t+1}$, which is now attained at some edge of $\tilde E_t$:
	\begin{align*}
		\omega_{t+1}(E)
		\geq
		\max\limits_{e \in E}  \omega_{t+1}(e)
		=
		\max\limits_{e \in \tilde E_t}  \omega_{t+1}(e)
		=
		\round{\frac 1 2 }^{\min\limits_{e \in \tilde E_t} \sum\limits_{i=1}^{t} \Inc(e,S_i) \cdot \bm{1}_{\{e \in E_i\}}} .
	\end{align*}
	The upper bound is obtained by using the algorithm's weight update rule. Since  $e_t$ has positive weight in $\omega_t$, but its weight in $\omega_{t+1}$ is set to $0$, we have a strict inequality
	\begin{align*}
		\omega_{t+1}(E)
		&=
		\sum\limits_{e \in E} \omega_{t+1}(e)
		<
		\sum\limits_{e \in E} \omega_{t} (e) \round{1 - \frac 1 2  \Inc(e,S_t)\cdot \bm{1}_{\{e \in E_t\}}}\\
		&=
		\sum\limits_{e \in E} \omega_{t} (e) - \frac 1 2\sum\limits_{e \in E}   \omega_{t} (e)   \Inc(e,S_t)\cdot \bm{1}_{\{e \in E_t\}}\\
		&=
		\omega_t(E) \round{
			1
			-
			\frac 1 2 \sum\limits_{e \in E} \frac{\omega_{t} (e)}{\omega_{t}(E)}   \Inc(e,S_t)\cdot \bm{1}_{\{e \in E_t\}}
		}.
	\end{align*}
	Unfolding this recursion and using the fact that $ 1+a \leq \exp(a)$, we get
	\begin{align*}
		\omega_{t+1}(E)
		&<
		|E| \cdot
		\exp \round{
			-
			\frac 1 2 \sum\limits_{i=1}^t \sum\limits_{e \in E} \frac{\omega_{i} (e)}{\omega_{i}(E)}   \Inc(e,S_i)\cdot \bm{1}_{\{e \in E_i\}}
		}.
	\end{align*}
	Combining the obtained upper and the lower bounds on $\omega_{t+1}(E)$ and taking the logarithm of each side, we get

	\begin{equation*}
		\ln\round{ \frac 1 2} \cdot \min\limits_{e\in \tilde E_t} \sum\limits_{i=1}^{t} \Inc(e,S_i)\cdot \bm{1}_{\{e \in E_i\}}
		<
		-
		\frac 1 2 \sum\limits_{i=1}^t \sum\limits_{e \in E} \frac{\omega_{i} (e)}{\omega_{i}(E)}   \Inc(e,S_i)\cdot \bm{1}_{\{e \in E_i\}}
		+
		\ln |E|,
	\end{equation*}
	which is equivalent to
	\begin{equation}\label{eq:edge-weights-after-log}
		\sum\limits_{i=1}^t \sum\limits_{e \in E} \frac{\omega_{i} (e)}{\omega_{i}(E)}   \Inc(e,S_i)\cdot \bm{1}_{\{e \in E_i\}}
		<
		2 \ln(2) \cdot \min\limits_{e\in \tilde E_t} \sum\limits_{i=1}^{t} \Inc(e,S_i)\cdot \bm{1}_{\{e \in E_i\}}
		+
		2\ln |E|.
	\end{equation}
	If $\proba = 1$, then $\bm 1_{\{e\in E_i\}} = 1$ for all $i$ and $e \in E$, thus taking total expectation we conclude
	\begin{equation*}
	\sum\limits_{i=1}^t \EE\bracket{ \sum\limits_{e \in E} \frac{\omega_{i} (e)}{\omega_{i}(E)}   \Inc(e,S_i)}
		<
		2 \ln(2) \cdot
		\EE\bracket{
			\min\limits_{e\in \tilde E_t} \sum\limits_{i=1}^{t}  \Inc(e,S_i)}
		+
		2\ln |E|.
	\end{equation*}

	\noindent
	Assume that $\proba< 1$. Since $\min f(x) - \min g(x) \leq \max (f(x) - g(x))$, \Cref{eq:edge-weights-after-log} implies
	\begin{align*}
		\sum\limits_{i=1}^t \sum\limits_{e \in E} \frac{\omega_{i} (e)}{\omega_{i}(E)}   \Inc(e,S_i)\cdot \bm{1}_{\{e \in E_i\}}
		<~~
		& 2\ln(2)\cdot \max\limits_{e  \in \tilde E_t} \sum\limits_{i=1}^{t} \Inc(e,S_i)\cdot\round{\bm 1_{\{e \in E_i\}} -\epsproba }
		\\
		&+2\ln(2) \cdot\min\limits_{e\in \tilde E_t} \sum\limits_{i=1}^{t} \Inc(e,S_i)\cdot \epsproba
		+ 2\ln |E| \,.
	\end{align*}

	\noindent
	Taking total expectation of each side, and using that for each fixed $i$, the random variables $\{\omega_i, S_i\}$ and $E_i$ are independent, we get

	\begin{equation}\label{eq:edge-wt-after-exp-and-proba}
	\begin{aligned}		\proba\cdot \sum\limits_{i=1}^{t} \sum\limits_{e \in E} \EE\bracket{\frac{\omega_{i} (e)}{\omega_{i}(E)}   \Inc(e,S_i)}
		<~~
		&2\ln(2)\cdot \EE\bracket{\max\limits_{e  \in \tilde E_t} \sum\limits_{i=1}^{t} \Inc(e,S_i)\cdot\round{ \bm 1_{\{e \in E_i\}} - \epsproba }}\\
		&~+2\ln(2)\cdot    \EE\bracket{\min\limits_{e\in \tilde E_t} \sum\limits_{i=1}^{t} \Inc(e,S_i)\cdot\epsproba}
		+ 2\ln |E| \,.
	\end{aligned}
	\end{equation}
	We need the following claim whose proof uses \Cref{lemma:Azuma} and is similar to \Cref{claim:martingale-bound}.
	\begin{claim}
	\begin{align*}
		\PP\bracket{ \max\limits_{e \in \tilde E_t} \sum\limits_{i=1}^{t} \Inc(e,S_i)\cdot\round{ \bm 1_{\{e  \in E_i\}} - \epsproba }
		\geq \peps  \ln(|E|t)}  \leq \frac{1}{t}.
	\end{align*}
	\end{claim}
	This, together with the fact that $\sum\limits_{i=1}^{t} \Inc(e,S_i)\cdot\round{\bm 1_{\{e  \in E_i\}}- \epsproba } \leq t$ always holds imply
	\begin{align*}
		\EE\bracket{\max\limits_{e \in \tilde E_t} \sum\limits_{i=1}^{t} \Inc(e,S_i)\cdot\round{\bm 1_{\{e  \in E_i\}}} - \epsproba }
		\leq
		 \peps  \ln(|E|t)
		 +
		 t \cdot \frac 1 t
		 \leq
		 \peps  \ln(|E|t) + 1.
		 \end{align*}
	Hence \Cref{eq:edge-wt-after-exp-and-proba} yields
	\begin{align*}
	 &\sum\limits_{i=1}^{t} \proba \cdot\sum\limits_{e \in E} \EE\bracket{\frac{\omega_{i} (e)}{\omega_{i}(E)}   \Inc(e,S_i)}\\
	&~< \frac{6\ln 2}{16}  \cdot \EE\bracket{\min\limits_{e\in \tilde E_t} \sum\limits_{i=1}^{t} \Inc(e,S_i)\cdot \proba}
	+   2\ln(2) \cdot \round{\peps  \ln(|E|t) + 1} 	+	2\ln |E|.
	\end{align*}
	Dividing both sides by $\proba$, we get
\begin{align*}
	 \sum\limits_{i=1}^{t} \sum\limits_{e \in E} \EE\bracket{\frac{\omega_{i} (e)}{\omega_{i}(E)}   \Inc(e,S_i)}
	< \frac{3\ln 2}{8}  \cdot \EE\bracket{\min\limits_{e\in \tilde E_t} \sum\limits_{i=1}^{t} \Inc(e,S_i)}
	+   \frac{\round{\frac{32\ln 2}{13} + 2}  \ln(|E|t) }{\proba}.
	\end{align*}
\end{proof}

\noindent
We need one more lemma to tie the previous two together.

\begin{lemma}\label{lemma:double-counting}
	For any $i \in [1,t]$, we have
	\begin{equation*}
		\EE\bracket{  \sum\limits_{S \in \S} \frac{\pi_{i} (S)}{\pi_{i}(\S)}   \Inc(e_i,S)}
		=
		\EE\bracket{  \sum\limits_{e \in E} \frac{\omega_{i} (e)}{\omega_{i}(E)}   \Inc(e,S_i)}.
	\end{equation*}
\end{lemma}

\begin{proof}

 Let  $F_{i} = \sigma\round{e_1,\dots, e_{i}, S_1, \dots, S_i, E_1, \dots, E_i, \S_1, \dots \S_i}.$ We have
\begin{align*}
 	\EE \bracket{ \sum\limits_{S \in \S} \frac{\pi_{i} (S)}{\pi_{i}(\S)}   \Inc(e_i,S) }
 	&=
 	\EE\bracket{\EE \bracket{ \sum\limits_{S \in \S} \frac{\pi_{i} (S)}{\pi_{i}(\S)}   \Inc(e_i,S) ~\bigg|~ F_{i-1}}  } \quad \text{and}\\
 \EE \bracket{  \sum\limits_{e \in E} \frac{\omega_{i} (e)}{\omega_{i}(E)}   \Inc(e,S_i) }
 	&=
 	\EE\bracket{\EE \bracket{  \sum\limits_{e \in E} \frac{\omega_{i} (e)}{\omega_{i}(E)}   \Inc(e,S_i) ~\bigg|~ F_{i-1}}  }.
\end{align*}
  Observe that $\omega_i$ and $\pi_i$ are measurable with respect to $F_{i-1}$, thus
	\begin{align*}
		&\EE\bracket{ \sum\limits_{S \in \S} \frac{\pi_{i} (S)}{\pi_{i}(\S)}   \Inc(e_i,S) ~\bigg|~ F_{i-1}}
		=
		\sum\limits_{e \in E} \frac{\omega_{i} (e)}{\omega_{i}(E)} \cdot  \round{ \sum\limits_{S \in \S} \frac{\pi_{i} (S)}{\pi_{i}(\S)}   \Inc(e,S)}\\
		&=
		\sum\limits_{e \in E} \sum\limits_{S \in \S} \frac{\omega_{i} (e)}{\omega_{i}(E)} \cdot    \frac{\pi_{i} (S)}{\pi_{i}(\S)}   \Inc(e,S)\\
		&=
		\sum\limits_{S \in \S} \frac{\pi_{i} (S)}{\pi_{i}(\S)}  \cdot \round{  \sum\limits_{e \in E}    \frac{\omega_{i} (e)}{\omega_{i}(E)}   \Inc(e,S) }
		=
		\EE\bracket{ \sum\limits_{e \in E} \frac{\omega_{i} (e)}{\omega_{i}(E)}   \Inc(e,S_i) ~\bigg|~F_{i-1}}.
	\end{align*}
\end{proof}
\noindent
Finally, we combine  Lemmas \ref{lemma:hyperplane-weights}, \ref{lemma:edge-weights}, and \ref{lemma:double-counting} in the following way

\begin{align*}
		&\EE\bracket{
			\max\limits_{S \in \S} \sum\limits_{i=1}^{t} \Inc(e_i,S)}
		\leq
		\frac{4}{3\ln 2}\cdot \sum\limits_{i=1}^t \EE\bracket{ \sum\limits_{S \in \S} \frac{\pi_{i} (S)}{\pi_{i}(\S)}   \Inc(e_i,S)}
		+
		\frac{\round{16 + \frac{4}{\ln 2}}\ln(|\S|t)}{3\probaS }\\
		&~=
		\frac{4}{3\ln 2}\cdot \sum\limits_{i=1}^t \EE\bracket{  \sum\limits_{e \in E} \frac{\omega_{i} (e)}{\omega_{i}(E)}   \Inc(e,S_i)}
		+
		\frac{\round{16 + \frac{4}{\ln 2}}\ln(|\S|t)}{3\probaS }\\
		&~<
		\frac 1 2\cdot\EE\bracket{\min\limits_{e\in \tilde E} \sum\limits_{i=1}^{t} \Inc(e,S_i)}
		+
		\frac{\round{\frac{128}{13} + \frac{8}{\ln 2}}\ln(|E|t)}{3\proba}
		+
		\frac{\round{16 + \frac{4}{\ln 2}}\ln(|\S|t)}{3\probaS }
\end{align*}

\noindent
This conludes the proof of the \Cref{lemma:MW-bound} and thus completes the proof of \Cref{thm:main-matching-result}.
\qed

\subsection{Proof of \nameref{thm:main}}\label{sec:mainthm-proof}

We will prove a more general statement of \nameref{thm:main} using \nameref{assumption}:

 \begin{theorem}\label{thm:main-disc-result}
    Let $(X,\S)$ be a set system that satisfies \nameref{assumption}.
   	The algorithm \nameref{discalgo}$\round{(X,\S), a,b,\gamma}$ constructs a coloring $\chi$ of $X$ of with expected discrepancy at most
   \[
   3\sqrt{ \frac{a n^\gamma\ln m}{\gamma}  + \frac{b\ln m\log n}{2} + 12 \ln^2 m \log n},
   \]
   with an expected number of Oracle calls at most
		\[
		\min\curly{  \frac{24n^{3-\gamma} \ln n}{a}  + \frac{18mn^{1-\gamma}\ln mn}{a}\cdot\min\curly{\frac{2}{1-\gamma},~ \log n} ,~ \frac{n^3}{7} + \frac{mn}{2}}.
		\]
\end{theorem}

\noindent The algorithm \nameref{discalgo} is presented in \Cref{algo:main-discalgo-compressed}.
It is easy to check that \nameref{thm:main} follows immediately from
   \Cref{thm:main-disc-result} by substituting $a = \frac{(2c)^{1/d} }{2\ln 2(1-1/d)}$, $b = \frac{\ln m}{\ln 2}$, and $\gamma = 1-1/d$.

\begin{algorithm}[ht]
\algotitle{\textsc{LowDiscColor}}{discalgo}
\caption{  \textsc{LowDiscColor}$\big((X,\S),  a, b, \gamma\big)$}
	\label{algo:main-discalgo-compressed}

	$n \leftarrow |X|$

	$ \curly{e_1, \dots, e_{\ceil{n/2}}} \leftarrow \nameref{mainalgo}\big((X,\S),  a, b, \gamma\big)$ \tcp*{see \Cref{algo:main}}

	\For{ $i = 1, \dots, \ceil{n/2}$}{ 

	$\{x_i,y_i\} \leftarrow \mathrm{endpoints}\round{e_i}$

	$\chi(x_i) = \begin{cases} 1 &\text{ with probability } 1/2 \\ -1 &\text{ with probability } 1/2\end{cases}$

	$\chi(y_i) = - \chi (x_i)$ \tcp*{we skip this step if $y_i = x_i$}

	}
	
	\textbf{return} $\chi$
	
\end{algorithm}

\noindent 
We will prove \Cref{thm:main-disc-result} using \Cref{thm:main-matching-result} and  the following lemma.

\begin{restatable}{lemma}{matchtodiscLemma}\label{lemma:matching-to-disc}
	Let $(X,\S)$ be a set system, $n = |X|$, $m = |\S|\geq 34$, and let $M$ be a perfect matching of $X$ with crossing number $\kappa$ with respect to $\S$ and for each edge $\{x,y\} \in M$ define
	 \[
	 \chi_M(x) = \begin{cases} 1 &\text{ with probability } 1/2 \\ -1 &\text{ with probability } 1/2 \end{cases}
	 \] 
	 and $\chi_M(y) = -\chi_M(x)$.
	Then the expected discrepancy of $\chi_M$ is at most $  \sqrt{3\kappa \ln m} $.
\end{restatable}
\begin{remark*}
A `high probability version' of \Cref{lemma:matching-to-disc} is well-known \cite[Lemma 2.5]{MWW93} and implies the above bound through basic probabilistic calculations, see Appendix for the precise proof.
\end{remark*}

\subsection*{Proof of \Cref{thm:main-disc-result}}

Let $M$ be the matching returned by \nameref{mainalgo}$\round{(X,\S), a,b, \gamma}$ during the run of \nameref{discalgo}$\round{(X,\S), a,b, \gamma}$ and let $\kappa(M)$ denote its crossing number with respect to $\S$. By \Cref{thm:main-matching-result},
\[
	\EE\bracket{\kappa(M)} \leq \frac{3a}{\gamma}  n^\gamma + \frac{3b\log n}{2} + 18 \ln \round{nm}\log n
\]
Using \Cref{lemma:matching-to-disc}, taking total expectation over the matchings returned by \nameref{mainalgo}, and applying Jensen's inequality, we get
\[
	\EE\bracket{\disc_{ \S } \round{\chi_M}} \leq \EE\bracket{\sqrt{3\kappa (M) \ln m}}
	\leq \sqrt{3\EE\bracket{\kappa (M)} \ln m}\,.
\]
Therefore, the expected discrepancy of the coloring returned by \nameref{discalgo}$\round{(X,\S), a,b, \gamma}$ is at most
\[
	\sqrt{3\round{\frac{3a}{\gamma}  n^\gamma + \frac{3b\log n}{2} + 18 \ln \round{nm}\log n} \ln m}\,.
\]
Each call of the membership Oracle is performed during the call of \nameref{mainalgo}, thus the bound on the expected number of membership Oracle calls follows immediately from \Cref{thm:main-matching-result}. This concludes the proof of \Cref{thm:main-disc-result} and thus of \nameref{thm:main}. \qed

\subsection{Proof of Corollaries \ref{cor:epsapproximations} and \ref{cor:vcdim-apx-result-inverted-dualvc}}

\subsubsection*{Proof of \Cref{cor:epsapproximations}}
The problems of low-discrepancy colorings and $\eps$-approximations are naturally connected:  finding a set of $|X|/2$ elements with low approximation error is essentially equivalent to finding a low-discrepancy coloring of $(X,\S)$: 

\begin{restatable}[{\cite[Lemma 2.1]{MWW93}}]{lemma}{disctoapx}\label{lemma:disc-to-apx}
	Let $(X,\S)$ be a set system with $|X| = n$, $X \in \S$ and let $\chi$ be a coloring with discrepancy $\disc_\S(\chi) = \Delta$ and let $A \subset X$ be a set of $\ceil{n/2}$ elements from the larger color class of $\chi$. Then $A$ is a $\round{2\Delta/n}$-approximation of $(X,\S)$.
\end{restatable}

\noindent
One can obtain lower order approximations by iteratively halving the point-set along a low-discrepancy colorings. The final approximation error can be bound  using \Cref{lemma:disc-to-apx} and the following basic fact. 
\begin{fact}\label{fact:apx-of-apx}
If $A_1$ is an $\eps_1$-approximation of $(X,\S)$ and $A_2$ is an $\eps_2$-approximation of $(A_1, \S|_{A_1})$, then $A_2$ is an $(\eps_1 + \eps_2)$-approximation of $(X,\S)$.
\end{fact}
\noindent
These ideas yield the following corollary of \Cref{thm:main-disc-result}, which immediately implies \Cref{cor:epsapproximations} by substituting $a = \frac{(2c)^{1/d} }{2\ln 2(1-1/d)}$, $b = \frac{\ln m}{\ln 2}$, and $\gamma = 1-1/d$.

\begin{corollary}\label{thm::main-apx-result}
    Let $(X,\S)$ be a set system that satisfies \nameref{assumption} and let $\eps  \in (0,1) $.
    Then
     \nameref{apxalgo}$\big((X,\S),  a, b, \gamma, \eps \big)$ returns a set $A \subset X$ of size at most
     \[
     2 \max\curly{ \round{30\sqrt{\frac{a\ln m }{\gamma}}\cdot \frac{1}{\eps}}^\frac{2}{2-\gamma}, \frac{12 \sqrt{\round{\frac b 2 + 12\ln m}\ln m\log n }}{\eps}} +1,
     \]
     with expected approximation guarantee
    $
    	\EE[\eps(A,X,\S)] \leq \eps,
    $
    and with an expected
    \[
    	\min\curly{ \frac{8n^{3-\gamma}\ln n}{a} + \frac{18 mn^{1-\gamma} \ln(mn)}{a}  \min\curly{\frac{4}{(1-\gamma)^2}, \log^2 n},~ \frac{n^3}{49}  + \frac{mn}{2}  }
    \]
    calls to the membership Oracle of $(X,\S)$.
\end{corollary}

\noindent
The precise analysis of the `halving process' (used in the algorithm \nameref{apxalgo} and in the deduction of \Cref{thm::main-apx-result} from \Cref{thm:main-disc-result}) is well-known, therefore we only present it in the Appendix (\Cref{sec:main-apx-proof}).
\qed

\begin{algorithm}[ht]
\algotitle{\textsc{Approximate}}{apxalgo}
\caption{  \textsc{Approximate}$\big((X,\S),  a, b, \gamma, \eps\big)$}
	\label{algo:main-apx}

	$A_0 \leftarrow X$

	\BlankLine

	$j = \floor{\log|X| +\min\curly{ \frac{2}{2-\gamma} \log \frac{\eps\sqrt{\gamma}}{30\sqrt{a\ln(|\S|)}},   \log \frac{\eps}{12 \sqrt{\round{\frac b 2 + 12\ln(|\S|)}\ln(|\S|)\log|X| }}} }$

	\BlankLine

	\For{$i = 1, \dots, j$}{

	$\chi\leftarrow \nameref{discalgo}\big((A_{i-1},\S|_{A_{i-1}}),  a, b, \gamma\big)$

	$A_i \leftarrow \chi^{-1}(1)$

	}

	\BlankLine

	\textbf{return} $A_j$

\end{algorithm}

\subsubsection*{Proof of \Cref{cor:vcdim-apx-result-inverted-dualvc}}
For set systems with bounded VC-dimension, one can obtain small-sized $\eps$-approximations via uniform sampling:

\begin{theorem}[{\cite[Theorem 8.3.23]{vershynin2018high}}]\label{thm:uniform-apx-bound-in-expectation}
	There is a universal constant $\capx$ such that for any constant $\dvc$ and any set system $(X,\S)$ with VC-dimension at most $\dvc$, a uniform random sample $A$ of $X$ satisfies
	\[
    	\EE[\eps(A,X,\S)] \leq  \sqrt{\frac{\capx\cdot \dvc}{|A|}}.
    \]
\end{theorem}

\noindent
Let $A_0$ be a uniform random sample of 
$$
\frac{4\capx\dvc}{\eps^{2}}
$$
elements from $X$. 
By \Cref{thm::main-apx-result}, the algorithm \nameref{apxalgo}$\big((A_0, \S|_{A_0}), a,b,\gamma, \eps/2 \big)$ returns a set $A_1$ with $\EE[\eps(A_1,A_0,\S_{A_0})] \leq \eps/2$. 
By \Cref{thm:uniform-apx-bound-in-expectation}, $A_0$ satisfies $\EE[\eps(A_0,X,\S)] \leq \eps/2$, and thus, by \Cref{fact:apx-of-apx},
\[
	\EE[\eps(A_1,X,\S)] \leq \EE[\eps(A_1,A_0,\S|_{A_0})] + \EE[\eps(A_0,X,\S)] \leq \eps.\\
\]

\noindent
The resulting guarantees are summarized in the next corollary.

\begin{corollary}\label{cor:vcdim-apx-result-inverted}
    Let $(X,\S)$ be a set system that satisfies \nameref{assumption}, let $\dvc$ denote the VC-dimension of $(X,\S)$, and let $\eps  \in (0,1) $.
    Let $A_0$ be a uniform random sample of $\frac{\capx\dvc}{(\eps/2)^2}$ elements from $X$.
    Then
     \nameref{apxalgo}$\big((A_0,\S|_{A_0}),  a, b, \gamma, \eps/2 \big)$ returns a set $A \subset X$ of size at most
     \[
     2 \max\curly{ \round{30\sqrt{\frac{a\ln |\S|_{A_0}| }{\gamma}}\cdot \frac{2}{\eps}}^\frac{2}{2-\gamma}, \frac{24 \sqrt{\round{\frac b 2 + 12\ln  |\S|_{A_0}|}\ln  |\S|_{A_0}|\log  |A_0| }}{\eps}} +1,
     \]
     with expected approximation guarantee
    $
    	\EE[\eps(A,X,\S)] \leq \eps,
    $
    and with an expected
    \begin{align*}
    	\min\Bigg\{ &\frac{8|A_0|^{3-\gamma}\ln |A_0|}{a} + \frac{18 |\S|_{A_0}||A_0|^{1-\gamma} \ln\round{|\S|_{A_0}||A_0|}}{a}  \min\curly{\frac{4}{(1-\gamma)^2}, \log^2 |A_0|},\\ &\frac{|A_0|^3}{49}  + \frac{|\S|_{A_0}||A_0|}{2}  \Bigg\}
    \end{align*}
    calls to the membership Oracle of $(X,\S)$.
\end{corollary}

\noindent
We can deduce \Cref{cor:vcdim-apx-result-inverted-dualvc} from \Cref{cor:vcdim-apx-result-inverted} by using that $|\S_{A_0}| = O(|A_0|^\dvc)$ by the Sauer-Shelah lemma~\citep{Sa72,Sh72} and substituting $a = \frac{(2c)^{1/d} }{2\ln 2(1-1/d)}$, $b = \frac{\ln m}{\ln 2}$, and $\gamma = 1-1/d$.



\noindent

\subsection{Proof of \Cref{thm:presampled-disc-guarantee}}\label{sec:Presampling-proof}



\noindent
We will deduce \Cref{lemma:presampling-guarantee} implies \Cref{thm:presampled-disc-guarantee}.

\subsection*{Proof of \Cref{thm:presampled-disc-guarantee}}

The randomized algorithm that achieves the guarantees of \Cref{thm:presampled-disc-guarantee} is presented in \Cref{algo:pres-discalgo}. 

\begin{algorithm}[ht]
\algotitle{\textsc{LowDiscColorPresampled}}{sampleddiscalgo}
\caption{  \textsc{LowDiscColorPresampled}$\big((X,\S), d, \alpha\big)$}
	\label{algo:pres-discalgo}

	$n \leftarrow |X|$

	$ \curly{e_1, \dots, e_{\ceil{n/2}}} \leftarrow \nameref{sampledalgo}\big((X,\S), d, \alpha\big)$

	\For{ $i = 1, \dots, \ceil{n/2}$}{ 

	$\{x_i,y_i\} \leftarrow \mathrm{endpoints}\round{e_i}$

	$\chi(x_i) = \begin{cases} 1 &\text{ with probability } 1/2 \\ -1 &\text{ with probability } 1/2\end{cases}$

	$\chi(y_i) = - \chi (x_i)$ \tcp*{we skip this step if $y_i = x_i$}

	}
	
	\textbf{return} $\chi$
	
\end{algorithm}
\begin{algorithm}[ht]
\algotitle{\textsc{MatchingPresampled}}{sampledalgo}
\caption{  \textsc{MatchingPresampled}$\big((X,\S),  d, \alpha\big)$}
	\label{algo:pres-matchalgo}

	$M \leftarrow \emptyset$

	\While{$|X| > 16$}{

	$n \leftarrow |X|$

	$E \leftarrow$ sample of $O(n^{1+\alpha} \ln n)$ edges from $\binom{X}{2}$

	$ \curly{e_1, \dots, e_{\ceil{n/16}}} \leftarrow \nameref{subalgo}\big((X,\S), E,  (2c)^{1/d}, \ln m, 1-\alpha/d, \lceil n/16\rceil\big)$

	$M \leftarrow M \cup \curly{e_1, \dots, e_{\ceil{n/16}}}$

	$X \leftarrow X \setminus \mathrm{endpoints}(M)$

	}

	match the remaining elements of $X$ randomly and add the edges to $M$

	\textbf{return} $M$

\end{algorithm}
















\noindent
Recall the following lemma from \Cref{sec:mainthm-proof}:
\matchtodiscLemma*

\noindent
By \Cref{lemma:matching-to-disc}, it is sufficient to show that the algorithm \nameref{sampledalgo} constructs a matching with expected crossing number $O\round{ n^{1- \alpha/d} + \ln |\S|\log n}$.
To prove this, recall the following statement on \nameref{subalgo}.

\mainlemma*
\noindent
Substituting $(a,b,\gamma,t) = \big((2c)^{1/d}, \ln m, 1-\alpha/d, \lceil n/16\rceil\big)$ and the proper values for $\proba$ and $\probaS$, we get the following bound on the expected crossing number of $\curly{e_1, \dots, e_{\ceil{n/16}}}$:
\begin{equation}\label{eq:pres-MW-bound}
		\EE\bracket{
			\max\limits_{S \in \S} \sum\limits_{i=1}^{\ceil{n/16}} \Inc(e_i,S)}
		\leq
\frac 1 2\cdot \EE\bracket{\min\limits_{e\in \tilde E_{\lceil n/16\rceil}} \sum\limits_{i=1}^{\ceil{n/16}} \Inc(e,S_i)}
		+
		O\round{n^{1-\alpha/d}}.
	\end{equation}
It remains to bound the expectation on the right-hand side of \Cref{eq:pres-MW-bound}.
By \Cref{lemma:presampling-guarantee}, with probability at least $1-\frac 1 n$, the initial sample $E$ contains a matching $M_0$ of size $\lceil n/4\rceil$ with crossing number
    \[
       c_0\cdot \round{ n^{1- \alpha/d} + \ln |\S|}
    \]
for some fixed constant $c_0$.
Assume that it happens, then clearly $M_0 \cap \tilde E_{\ceil{n/16}}$ also has crossing number at most $c_0\cdot \round{ n^{1- \alpha/d} + \ln |\S|}$ with respect to $\S$. Moreover,  since we only zeroed the weights of edges adjacent to $2\cdot \lceil n/16\rceil$ distinct vertices of $X$, there are at least $\lceil n/8-2\rceil $ edges of $M_0$  with positive weight when \nameref{subalgo} terminates. That is, $\abs{M_0 \cap \tilde E_{\ceil{n/16}} }\geq \lceil n/8-2\rceil$ and $\lceil n/8-2\rceil >0$ since $n > 16$ at each call of \nameref{subalgo} . By the pigeonhole principle, there is an edge in $M_0 \cap \tilde E_{\ceil{n/16}}$ which is crossed by at most
\[
	\frac{c_0\cdot \round{ n^{1- \alpha/d} + \ln |\S|} \cdot \ceil{n/16} }{\lceil n/8-2\rceil} = O\round{ n^{1- \alpha/d} + \ln |\S|}
\]
sets from $S_1, \dots, S_{\ceil{n/16}}$.
Therefore, we have 
\begin{equation*}
		\EE\bracket{
			\max\limits_{S \in \S} \sum\limits_{i=1}^{\ceil{n/16}} \Inc(e_i,S)}
			\leq
\frac 1 2\cdot \EE\bracket{\min\limits_{e\in \tilde E_{\lceil n/16\rceil}} \sum\limits_{i=1}^{\ceil{n/16}} \Inc(e,S_i)}
		{+}
		O\round{n^{1-\alpha/d}}
		=
		O\round{ n^{1- \alpha/d} + \ln |\S|},
	\end{equation*}
	where the last bound holds with with probability at least $1-\frac 1 n$.
Since the crossing number of any matching of $X$ is $O(n)$, the expected crossing number of the matching returned by the subroutine $\nameref{subalgo}\big((X,\S), E,  (2c)^{1/d}, \ln m, 1-\alpha/d, \lceil n/16\rceil\big)$ is $O\round{ n^{1- \alpha/d} + \ln |\S|}$. The algorithm \nameref{sampledalgo} makes $\log n$ calls to \nameref{subalgo} with exponentially decreasing input sizes. It can easily be deduced (with calculations analogous to the ones in \Cref{subsec:matching-proof}) that the  expected crossing number of the matching returned by \nameref{sampledalgo} is $O\round{ n^{1- \alpha/d} + \ln |\S|\log n}$. Hence we have shown that  \Cref{thm:presampled-disc-guarantee} is a consequence of \Cref{lemma:presampling-guarantee}.
\qed

\subsection*{Proof of \Cref{lemma:presampling-guarantee}}
\noindent
Our starting point is the following algorithm which is a variant of the classical MWU method \citep{Wel88,CW89, Wel92}.


\begin{algorithm}[ht]
\algotitle{\textsc{RelaxedMWU}}{RCWRalgo}
\caption{  \textsc{RelaxedMWU}$\big((X,\S), \alpha, E \big)$}
	\label{algo:RelaxCW}
	$\omega_1(S) \leftarrow 1$ for all $S\in \S$

	$X_1 \leftarrow X$

	\For{$i = 1, \dots, n/2$}{

	$\mathcal{E}_i\leftarrow$ the $|X_i|^{2-\alpha}$ lightest edges in $\binom{X_i}{2}$ w.r.t. $\omega_i$

	\BlankLine

	\If{$E \cap \mathcal{E}_i = \emptyset$}{
		set $ T = i-1$ and \Return $\curly{e_1, \dots, e_{i-1}}$
		}
		\Else{Pick an edge $e_i$ from $E \cap \mathcal{E}_i$ uniformly at random

	Define $\omega_{i+1}$ from $\omega_i$ by doubling the weights of each set crossing $e_i$

	$X_{i+1} \leftarrow X_{i} \setminus \mathrm{endpoints}(e_i)$
	}
	}
	set $T=n/2$ and \Return $\curly{e_1, \dots, e_{n/2}}$
\end{algorithm}

\noindent
The first part of  \Cref{lemma:presampling-guarantee} is implied by the following two properties of \nameref{RCWRalgo}:
\begin{enumerate}
	\item for any halting time $T=t$, the edges returned by \nameref{RCWRalgo} have crossing number $O\round{t^{1- \alpha/d}}$;
	\item if $E \subseteq \binom{X}{2}$ is an i.i.d. sample where each edge is picked with probability
\[
	p =  \min \curly{\frac{2\ln n}{n^{1-\alpha}} + \frac{4\ln(2/\delta)}{n^{2-\alpha}}, ~1}~,
\]
then $T \geq n/4$ with probability at least $1-\delta$. In other words,  \nameref{RCWRalgo}$((X,\S), \alpha, E)$  returns at least $n/4$ edges with probability at least $1-\delta$.
\end{enumerate}
\paragraph{1. Bounding the crossing number of the output.} Assume that $\tau: \NN \times \RR \to \RR$ is a function such that at iteration $i$, \nameref{RCWRalgo} picks an edge which is crossed by ranges of total weight at most $\tau(|X_i|,\omega_i(\S))$ or it terminates.
Then at each iteration, the total weight of ranges in $\S$ changes as
    \begin{align*}
        \omega_{i+1}(\S)
        &\leq
        \omega_i(\S) + \tau(|X_i|, \omega_i(\S))
        =
        \omega_i(\S)\round{ 1 + \frac{\tau(|X_i|, \omega_i(\S))}{\omega_i(\S)}}\\
        &\leq
        \omega_1(\S) \prod_{j=1}^{i}\round{1+ \frac{\tau(|X_j|, \omega_j(\S))}{\omega_j(\S)}}
        =
        |\S|\cdot \prod_{j=1}^{i}\round{1+ \frac{\tau(|X_j|, \omega_j(\S))}{\omega_j(\S)}}
    \end{align*}
Let $ t \in [1, n/2]$ be a stopping time and let $\kappa_t$ denote the maximum number of edges in $\{e_1, \dots,e_{t} \}$ that are crossed by any set in $\S$, then by the update rule,
\[
    \omega_{t+1}(\S) \geq \max_{S \in \S} \omega_{t+1}(S) = 2^{\kappa_t}.
\]
We get that
\begin{align*}
    2^{\kappa_t}
    \leq
    \omega_{t+1}(\S)
    \leq
    |\S|\cdot \prod_{j=1}^{t}\round{1+ \frac{\tau(|X_j|, \omega_j(\S))}{\omega_j(\S)}}
    \leq
    |\S|\cdot \exp\round{\sum_{j=1}^{t} \frac{\tau(|X_j|, \omega_j(\S))}{\omega_j(\S)}}
\end{align*}
which implies
\begin{equation}\label{eq:CW-crnr-bound}
    \kappa_t
    \leq
    \frac{1}{\ln 2} \round{\ln |\S| + \sum_{j=1}^{t} \frac{\tau(|X_j|, \omega_j(\S))}{\omega_j(\S)}}.
\end{equation}

\noindent
We use the following lemma to bound the function $\tau(\cdot,\cdot)$ for set systems with polynomially bounded dual shatter function.

        \begin{lemma}\label{lemma:robust-short-edge-presample}
            Let $(X, \S)$ be a set system with dual shatter function $\pi^*_{\S}(k) \leq c_1 \cdot k^d$. Then for any $Y \subset X$,  $w: \S \to \NN$, and  parameter $ |Y| \leq \ell \leq \binom{|Y|}{2}$ there are at least $\ell$ distinct edges in $\binom{Y}{2}$ such that any of these edges are crossed by sets of total weight at most $\tau_\ell(|Y|, w(\S)) = (10c_1)^{1/d} \cdot \frac{w(\S)\cdot \ell^{1/d}}{|Y|^{2/d}}$.
        \end{lemma}
        \begin{proof}
    	Let $(\S_w, \R_Y)$ denote the set system where $\S_w$ contains $w(S)$ copies of each $S \in \S$, $\R_Y =  \{ R_y : y \in Y\}$,  and $R_y = \{ S\in \S_w : y \in S\}$.
		Note that $|\S_w| = w(\S)$ and the shatter function of $(\S_w, \R_X)$ is the dual shatter function of $(Y,\S)$.
    	Recall the following lemma of  \cite{Haussler92spherepacking}.
    \begin{restatable}[Packing Lemma]{thm}{packinglemma}
         Let  $(X,\S)$ be a set system with shatter function $\pi_{\S} (k) \leq c_1\cdot k^d$ and  $1 < \delta <|X|$ be a parameter. Furthermore, let $\P \subset \S$ be a $\delta$-separated set, that is, $|S_1 \Delta S_2| \geq \delta$  for any $S_1,S_2 \in \P$. Then
         \[
         	|\P| \leq 2c_1 \round{\frac{|X|}{\delta}}^d.
         \]
    \end{restatable}
    	\noindent
    	For the choice of 
    	\[
    		\delta_\ell = \round{ 10c_1 \cdot \frac{w(\S)^{d}\ell}{|Y|^2} }^{1/d},
    	\]
    	the Packing Lemma implies that any $\delta_{\ell}$-separated subset of ranges in $\R_Y$  has cardinality at most
    	\[
    		C_{\ell} = 2c_1 \round{\frac{w(\S)}{\delta_{\ell}}}^d = \frac{|Y|^2}{5\ell}
    	\]
    	Observe that for any pair $x,y \in Y$, the set $R_x \Delta R_y$ contains precisely the sets in $\S_w$ that cross the edge $xy$. Consider the graph $G_Y$ on $Y$ defined by the edges that are crossed by at least $\delta_{\ell}$ sets in $\S_w$. The Packing Lemma implies that $G_Y$ does not contain a clique on $C_{\ell} + 1$ vertices.
    	Thus by the classical theorem of extremal graph theory \cite{turan1941external}, the number of pairs that are \emph{not} edges in $G_Y$ is at least
    	\[
    		C_{\ell} \binom{ \left\lfloor |Y| / C_{\ell} \right\rfloor}{2}
    		\geq
    		C_{\ell} \cdot \frac{\round{|Y| / C_{\ell} -1} \round{|Y| / C_{\ell}-2}}{2}
    		\geq
    		\frac{|Y|^2}{2C_{\ell}} - \frac{3 |Y|}{2}
    		=
    		\frac{5\ell}{2} - \frac{3|Y|}{2}
    		=\ell,
    	\]
    	where we used that $|Y| \leq \ell$. That is, there are at least $\ell$ edges which cross ranges of total weight at most $\delta_\ell$.
    	This concludes the proof of \Cref{lemma:robust-short-edge-presample}.
        \end{proof}

\noindent
At iteration $i$, we have $|X_i|=n-2i+2$ and we pick one of the $|\mathcal{E}_i| =|X_i|^{2-\alpha}$ lightest edges of $\binom{X_i}{2}$. By 
\Cref{lemma:robust-short-edge-presample}, each edge in $\mathcal{E}_i$ crosses ranges of total weight at most
\[
		(10c_1)^{1/d} \cdot \frac{w_i(\S)\cdot |\mathcal{E}_i|^{1/d}}{|X_i|^{2/d}} 
		=
		\frac{(10c_1)^{1/d}w_i(\S)}{|X_i|^{\alpha/d}}
		=
		\frac{(10c_1)^{1/d}w_i(\S)}{(n-2i+2)^{\alpha/d}}~,
\]
which bounds $\tau(|X_i|, \omega_i(\S)) $. Thus 
 \Cref{eq:CW-crnr-bound} implies that for any stopping time $ t \in [1, n/2]$, the matching $\curly{e_1,\dots,e_t}$ returned by \nameref{RCWRalgo} has crossing number at most
    \begin{equation}\label{eq:pres-cr-bound}
        \frac{\ln|\S|}{\ln 2}  + \frac{(10c_1)^{1/d}}{\ln 2}  \sum_{j=1}^{t} \frac{1}{(n-2j+2)^{\alpha/d}}
        \leq
        \frac{\ln|\S|}{\ln 2} +  \frac{(10c_1)^{1/d}}{\ln 2} \cdot \frac{t^{1- \alpha/d}}{1-\alpha/d}~.
    \end{equation}



\paragraph{2. Halting time on a random input.} 
Now we show that if $E \subseteq \binom{X}{2}$ is a random edge-set, where each edge is picked i.i.d with probability
\[
	p = \min \curly{ \frac{2\ln n}{n^{1-\alpha}} + \frac{4\ln(2/\delta)}{n^{2-\alpha}},~1}
\]
then with probability at least $1-\delta$, the algorithm \nameref{RCWRalgo}$((X,\S), \alpha, E)$  
satisfies
\[
	\PP\bracket{T \leq n/4} \leq \delta.
\]
If $p = 1$, then the statement is trivially true, therefore  we assume that $p < 1$. We will bound the probabilities $\PP[T = i]$ for each $i = 1, \dots, n/4$.
Since $E$ is an i.i.d. uniform random sample of $\binom{ X}{2}$,
$$
    \PP [T= 1] = \PP[E \cap \mathcal{E}_1 = \emptyset] = (1-p)^{|\mathcal{E}_1|} = (1-p)^{n^{2-\alpha}}.
$$
Observe that in iteration $i\geq 2$ of the algorithm \nameref{RCWRalgo}, the edge-set $\mathcal{E}_i$ depends on the previously picked edges. To signify this, for any set of edges $e^1, \dots, e^{i-1}$, we denote the set of $(n-2i+2)^{2-\alpha}$ shortest edges of $X\setminus\mathrm{endpoints}( e^1, \dots, e^{i-1} )$ as $\mathcal{E}_i(e^1, \dots, e^{i-1})$.  We say that a vector of edges $(e^1, \ldots, e^{i})$ is \emph{feasible} if $e^1 \in \mathcal{E}_1, e^2 \in \mathcal{E}_2(e^1), \dots, e^i \in \mathcal{E}_{i}(e^1, \ldots, e^{i-1})$. For brevity, we write $\mathbf{e}^{i} = (e^1, \dots, e^i)$ for a vector of edges with the agreement $\mathbf e^0 = \emptyset$ and $\mathbf{e}_{i} = (e_1, \dots, e_i)$ for the vector of random variables from \nameref{RCWRalgo}. Observe that
\begin{align}\label{eq:One1}
\begin{split}
    &\PP[T = i + 1]
    =
    \PP\bracket{E \cap \mathcal{E}_1 \neq \emptyset, E \cap \mathcal{E}_2 \neq \emptyset, \ldots, E \cap \mathcal{E}_{i + 1} = \emptyset}\\
    &=
   \sum_{\mathbf e^{i} \text{ feasible}}\PP\bracket{ E \cap \mathcal{E}_{i + 1}(\mathbf e^{i}) = \emptyset,\, E \cap \mathcal{E}_j(\mathbf e^{j-1}) \neq \emptyset ~\forall j \in [1,i] ~\big|~ \mathbf e_{i} = \mathbf e^i}
    \cdot \PP\bracket{\mathbf e_{i} = \mathbf e^i}\\
    &\leq
    \sum_{\mathbf e^{i} \text{ feasible}}\PP\bracket{E \cap \mathcal{E}_{i + 1}(\mathbf e^{i}) = \emptyset ~\big|~ \mathbf e_{i} = \mathbf e^i}\PP\bracket{\mathbf e_{i} = \mathbf e^i}
    .
    \end{split}
\end{align}
Note that $\mathcal{E}_{i + 1}(e^1, \dots, e^i)$ is a fixed, non-random set. Using  Bayes' theorem 
we can express the conditional probabilities in the right hand side of \Cref{eq:One1} as
\begin{align*}
  &\PP\bracket{E \cap \mathcal{E}_{i + 1}(\mathbf e^i) = \emptyset ~\big|~ \mathbf e_{i} = \mathbf e^i}
  =
  \frac{ \PP\bracket{\mathbf e_{i} = \mathbf e^i ~\big|~ E \cap \mathcal{E}_{i + 1}(\mathbf e^i) = \emptyset } \cdot \PP\bracket{E \cap \mathcal{E}_{i + 1}(\mathbf e^i) = \emptyset} }{ \PP\bracket{\mathbf e_{i} = \mathbf e^i} }~.
\end{align*}
Substituting this back to \Cref{eq:One1}, we get
\begin{align}\label{eq:Two2}
\begin{split}
\PP\bracket{T = i+1}
&=
\sum_{\mathbf e^i \text{ feasible}} \PP\bracket{\mathbf e_{i} = \mathbf e^i \mid E \cap \mathcal{E}_{i + 1}(\mathbf e^i) = \emptyset }
\cdot \PP\bracket{E \cap \mathcal{E}_{i + 1}(\mathbf e^i) = \emptyset}\\
&=  \sum_{\mathbf e^i \text{ feasible}} \PP\bracket{\mathbf e_{i} = \mathbf e^i \mid E \cap \mathcal{E}_{i + 1}(\mathbf e^i) = \emptyset } \cdot (1-p)^{|\mathcal{E}_{i + 1}(\mathbf e^i)|}\\
&=  \sum_{\mathbf e^i \text{ feasible}} \PP\bracket{\mathbf e_{i} = \mathbf e^i \mid E \cap \mathcal{E}_{i + 1}(\mathbf e^i) = \emptyset } \cdot (1-p)^{(n-2i)^{2-\alpha}}~.
\end{split}
\end{align}
as we have $|\mathcal{E}_{i + 1}(\mathbf e^i)| = (n-2i)^{2-\alpha}$ for any $\mathbf e^i = (e^1, \dots, e^i)$. We proceed by bounding the probability $\PP\bracket{\mathbf e_{i} = \mathbf e^i \mid E \cap \mathcal{E}_{i + 1}(\mathbf e^i) = \emptyset }$.
 Observe that
\begin{align*}
    &\PP\bracket{\mathbf e_{i} = \mathbf e^i \mid E \cap \mathcal{E}_{i + 1}(\mathbf e^i) = \emptyset }\\
&= \PP\bracket{e_{i} = e^i \mid E \cap \mathcal{E}_{i + 1}(\mathbf e^i ) = \emptyset,~ \mathbf e_{i-1} = \mathbf e^{i-1} }\cdot \PP \bracket{\mathbf e_{i-1} = \mathbf e^{i-1} \mid E \cap  \mathcal{E}_{i + 1}(\mathbf e^i) = \emptyset}\\
&= \prod_{j=2}^i \PP\bracket{e_{j} = e^j \mid E \cap \mathcal{E}_{i + 1}(\mathbf e ) = \emptyset, \mathbf e_{j-1} = \mathbf e^{j-1} } \cdot \PP\bracket{e_{1} = e^1 \mid E \cap \mathcal{E}_{i + 1}(\mathbf e^i ) = \emptyset}
\end{align*}
Recall that $e_1$ was picked uniformly at random from $\mathcal{E}_1 \cap E$, where $\mathcal{E}_1$ is a fixed set such that  $e^1 \in \mathcal{E}_1$ for any feasible $e^1$, and that $E$ is a uniform random sample. This implies the following 
\begin{align*}
 &\PP\bracket{e_{1} = e^1 ~\big|~ E \cap \mathcal{E}_{i + 1}(\mathbf e^i ) = \emptyset}\\
   &=
    \sum_{ S' \subset \mathcal{E}_1} \PP \bracket{e_1 = e^1 ~\big|~ E \cap \mathcal{E}_{i + 1}(\mathbf e^i ) = \emptyset, \, E \cap \mathcal{E}_1 = S'}
    \cdot \PP \bracket{E \cap \mathcal{E}_1 = S' ~\big|~ E \cap \mathcal{E}_{i + 1}(\mathbf e^i ) = \emptyset}\\
    &=
    \sum_{e^1 \in  S' \subset \mathcal{E}_1 \setminus \mathcal{E}_{i+1}(\mathbf e^i)} \frac{1}{|S'|}\cdot
        p^{|S'|}\cdot (1-p)^{|\mathcal{E}_1 \setminus \mathcal{E}_{i+1}(\mathbf e^i)| - |S'|}\\
    &=
    \sum_{\ell=1}^{|\mathcal{E}_1 \setminus \mathcal{E}_{i+1}(\mathbf e^i)|} \binom{|\mathcal{E}_1 \setminus \mathcal{E}_{i+1}(\mathbf e^i)|-1}{  \ell-1}\cdot  \frac{1}{\ell}\cdot  p^\ell \cdot  (1-p)^{|\mathcal{E}_1 \setminus \mathcal{E}_{i+1}(\mathbf e^i)| - \ell} \\
    &=
    \sum_{\ell=1}^{|\mathcal{E}_1 \setminus \mathcal{E}_{i+1}(\mathbf e^i)|}
    \frac{1}{|\mathcal{E}_1 \setminus \mathcal{E}_{i+1}(\mathbf e^i)|} \binom{|\mathcal{E}_1 \setminus \mathcal{E}_{i+1}(\mathbf e^i)|}{ \ell}\cdot    p^\ell \cdot  (1-p)^{|\mathcal{E}_1 \setminus \mathcal{E}_{i+1}(\mathbf e^i)| - \ell}\\
    &=
    \frac{1}{|\mathcal{E}_1 \setminus \mathcal{E}_{i+1}(\mathbf e^i)|}  \round{ p + (1-p)}^{|\mathcal{E}_1 \setminus \mathcal{E}_{i+1}(\mathbf e^i)|} - (1-p)^{|\mathcal{E}_1 \setminus \mathcal{E}_{i+1}(\mathbf e^i)|}\\
    &= \frac{1}{|\mathcal{E}_1 \setminus \mathcal{E}_{i+1}(\mathbf e^i)|} \round{1 - (1-p)^{|\mathcal{E}_1 \setminus \mathcal{E}_{i+1}(\mathbf e^i)|} }.
\end{align*}

\noindent
The last step is to bound the probabilities
$
    \PP\bracket{e_{j} = e^j \mid E \cap \mathcal{E}_{i + 1}(\mathbf e^i ) = \emptyset,~ \mathbf e_{j-1} = \mathbf e^{j-1} }
$
for $j \geq 2$.
Note that, given the realization $\mathbf e_{j-1} = \mathbf e^{j-1}$, the set $\mathcal{E}_j(\mathbf e^{j-1})$ is not random, and thus we have a similar relation as before
\begin{align*}
    &\PP\bracket{e_{j} = e^j \mid E \cap \mathcal{E}_{i + 1}(\mathbf e^i ) = \emptyset, ~\mathbf e_{j-1} = \mathbf e^{j-1} }
    =
    \PP\bracket{e_{j} = e^j \mid E \cap \mathcal{E}_{i + 1}(\mathbf e^i ) = \emptyset}\\
    &=\sum_{e^j \in   S' \subset \mathcal{E}_j(\mathbf e^{j-1}) \setminus \mathcal{E}_{i+1}(\mathbf e^i)} \frac{1}{|S'|}\cdot
        p^{|S'|}\cdot (1-p)^{|\mathcal{E}_j(\mathbf e^{j-1}) \setminus \mathcal{E}_{i+1}(\mathbf e^i)| - |S'|}\\
    &=
    \sum_{\ell=1}^{|\mathcal{E}_j(\mathbf e^{j-1}) \setminus \mathcal{E}_{i+1}(\mathbf e^i)|} \binom{|\mathcal{E}_j(\mathbf e^{j-1}) \setminus \mathcal{E}_{i+1}(\mathbf e^i)|-1 }{ \ell-1}\cdot  \frac{1}{\ell}\cdot  p^\ell \cdot  (1-p)^{|\mathcal{E}_j(\mathbf e^{j-1}) \setminus \mathcal{E}_{i+1}(\mathbf e^i)| - \ell} \\
    &=
    \frac{1}{|\mathcal{E}_j(\mathbf e^{j-1}) \setminus \mathcal{E}_{i+1}(\mathbf e^i)|} \round{1 - (1-p)^{|\mathcal{E}_j(\mathbf e^{j-1}) \setminus \mathcal{E}_{i+1}(\mathbf e^i)|} }.
\end{align*}
Recall that $\abs{\mathcal E_i (\mathbf e^{i-1})} = (n-2i+2)^{2-\alpha}$ for all $i \in [1,T]$. Thus for each $1 \leq j \leq i$ we have $$ (n-2(j-1))^{2-\alpha} - (n-2i)^{2-\alpha} \leq |\mathcal{E}_j (\mathbf e^{j-1}) \setminus \mathcal{E}_{i+1}(\mathbf e^i)| \leq(n-2(j-1))^{2-\alpha}$$ and so the  probability  $\PP\bracket{e_{j} = e^j \mid E \cap \mathcal{E}_{i + 1}(\mathbf e^i ) = \emptyset, ~\mathbf e_{j-1} = \mathbf e^{j-1} }$ is maximized if $$|\mathcal{E}_j(\mathbf e^{j-1}) \setminus \mathcal{E}_{i+1}(\mathbf e^i)| = (n-2(j-1))^{2-\alpha} - (n-2i)^{2-\alpha}$$
Putting everything together and using the notation $k_i = \abs{\mathcal E_i(\mathbf e^{i-1})} = (n-2(i-1))^{2-\alpha}$, we get
\begin{align*}
\PP\bracket{T = i+1}
&\leq
(1-p)^{k_{i+1}} \cdot  \sum_{\mathbf e^i \text{ feasible }} \prod_{j=1}^{i} \frac{1}{k_j - k_{i+1}} \round{1 - (1-p)^{ k_j - k_{i+1} } }\\
&=(1-p)^{k_{i+1}} \cdot k_1 \cdot k_2 \cdots k_i \cdot\prod_{j=1}^{i} \frac{1}{k_j - k_{i+1}} \round{1 - (1-p)^{ k_j - k_{i+1} } }\\
&= (1-p)^{k_{i+1}} \cdot \prod_{j=1}^i \frac{k_j}{k_j - k_{i+1}} \round{1 - (1-p)^{ k_j - k_{i+1} } }
\end{align*}
For any $i \geq 1$, we conclude the following bound on the probability of $T \leq i+1$
\begin{align*}
    \PP \bracket{T \leq i+1}
    =
    \sum\limits_{j=1}^{i+1} \PP \bracket{T = j}
    \leq
    (1-p)^{k_1} +
    \sum_{\ell =1}^{i+1} (1-p)^{k_{\ell+1}} \cdot \prod_{j=1}^\ell \frac{k_j}{k_j - k_{\ell+1}} \round{1 - (1-p)^{ k_j - k_{\ell+1} } }
\end{align*}
Using the bounds $k_1\geq k_2\geq  \dots\geq k_{n/4-1} \geq k_{n/4}$ and $k_j - k_{i+1} \geq 1$ for all $1 \leq j \leq i$, we can bound the probability of $T \leq n/4$ as
\begin{align*}
    \PP\bracket{T\leq n/4}
    &\leq
    (1-p)^{k_{n/4}}+
    (1-p)^{k_{n/4}} \sum_{\ell=1}^{n/4} \prod_{j=1}^{\ell} k_j \cdot p
    \leq
    (1-p)^{k_{n/4}} \sum_{\ell=0}^{n/4} \round{k_1 \cdot p }^\ell\\
    &=
    (1-p)^{k_{n/4}}
    \frac{1-(pk_1)^{n/4+1}}{1-pk_1}
    \leq
    2(1-p)^{k_{n/4}} \cdot (pk_1)^{n/4}
    \leq
    2\exp (-pk_{n/4}) \cdot k_1^{n/4}
\end{align*}
Substituting $k_1 = n^{2-\alpha}$, $k_{n/4} \geq (n/2)^{2-\alpha}\geq n^{2-\alpha}/4$ and $p = \frac{2\ln n}{n^{1-\alpha}} + \frac{4\ln(2/\delta)}{n^{2-\alpha}}$, we conclude
\begin{align*}
	\PP\bracket{T\leq n/4}
	\leq
	2\exp\round{ -\frac{n\ln n}{2} - \ln \frac 2 \delta} \cdot \round{n^{2-\alpha}}^{n/4}
	=
	2\cdot \frac{n^{n/2-\alpha n/4}}{n^{n/2}} \cdot \frac{\delta}{2}
	\leq
	\delta
\end{align*}
Therefore, with probability at least $1 - \delta$, \nameref{RCWRalgo} returns a matching of size $n/4$. This, together with \Cref{eq:pres-cr-bound} implies the first part (upper bound) of \Cref{lemma:presampling-guarantee}.






\paragraph{Lower bound construction. }
	The example is a geometric set system induced by half-spaces on a subset of the integer grid, more precisely, let $ X$ be the set of $n = d \cdot \left\lceil n_0^{1/d} \right\rceil^d$ points defined as $\times_{i=1}^d \bracket{1, \left\lceil n_0^{1/d} \right\rceil } \subset \mathbb Z^d$ and let $\S$ consist of the $d \cdot \left\lfloor n_0^{1/d}\right\rfloor$ subsets of $X$ induced by half-spaces of the form
	\[
		\left\{ x_i \leq j + 1/2 ~\bigg|~ i = 1, \dots, d,~ j = 1, \dots \left\lfloor n_0^{1/d} \right\rfloor \right\}.
	\]


	Observe that for any edge $\{\mathbf x, \mathbf y \} \in \binom{X}{2}$, the number of ranges in $\S$ that crosses $\{\mathbf x, \mathbf y\}$ is precisely the $\ell_1$-distance of $\mathbf x$ and $\mathbf y$, which is defined as
	\[
		\ell_1(\mathbf x, \mathbf y) = \sum\limits_{i=1}^d \abs{ \mathbf x_i - \mathbf y_i}~.
	\]
	Using this observation, it is easy to see that for any fixed $k$, the number of edges crossed by at most $k$ sets from $\S$ is at most $nk^d$. We refer to these edges as \emph{$k$-good} and denote their set with $\mathcal G_k$.
	 Let $p(n) = o(n^{\alpha-1})$ be a function and define $k_p(n) = \round{\frac{1}{16p(n)}}^{1/d}$. The expected number of $k_p(n)$-good edges in $E$ is
	\[
		\EE\bracket{|E\cap \mathcal G_{k_p(n)}| } \leq n\round{k_p(n)}^d \cdot p(n) = \frac{n}{16}~.
	\]
	Thus, by Markov's inequality, we have $|E\cap \mathcal G_{k_p(n)}| \leq \frac n 8$ with probability at least $1/2$. Assume that $|E\cap \mathcal G_{k_p(n)}| \leq \frac n 8$ holds and let
	 $M \subset E$ be any subset of size $n/4$. Then $M$ contains at least $n/8$ edges which are not $k_p(n)$-good.
	Therefore, the number of crossings between the edges of $M$ and the sets of $\S$ is at least
	\[
		\frac{n}{8} \cdot \round{\frac{1}{16p(n)}}^{1/d}~.
	\]
	Recall that $|\S| = d \cdot \left\lfloor n_0^{1/d}\right\rfloor \leq d n^{1/d}$ and so by the pigeonhole principle, we get that there is a range in $\S$ that crosses at least
	\[
		\frac{\frac{n}{8} \cdot \round{\frac{1}{16p(n)}}^{1/d}}{|\S|}\geq
		\frac{\frac{n}{8} \cdot \round{\frac{1}{16p(n)}}^{1/d}}{dn^{1/d}} = \frac{n^{1-1/d}}{8d\cdot(16)^{1/d}} \cdot \underbrace{\round{\frac{1}{p(n)}}^{1/d}}_{\omega\round{n^{(1-\alpha)/d}}} = \omega\round{n^{1-\alpha/d}}
	\]
	edges of $M$, which concludes the proof of \Cref{lemma:presampling-guarantee} and thus completes the proof of \Cref{thm:presampled-disc-guarantee}.

\qed

\section{Geometric Set Systems} \label{sec:corollaries}

In this section, we apply our algorithms for set systems induced by geometric objects.  
We will show that \Cref{thm:main-disc-result} implies improved constructions of low-discrepancy colorings in several  geometric set systems, see \Cref{table:disc-for-dualshat}.

\begin{table}[ht!]

		\centering
		\resizebox{\textwidth}{!}{
			\begin{tabular}{ccccc}
								\toprule



				\multicolumn{1}{c}{}
				& \multicolumn{2}{|c|}{\textsc{\textbf{Our Method}}}
				& \multicolumn{2}{c}{\textsc{Previous-best}}

				\\ 

				\multicolumn{1}{c|}{\textsc{Set system}}
				& \multicolumn{1}{c}{\textsc{Discrepancy}}
				& \multicolumn{1}{c|}{\textsc{Time}}
				& \multicolumn{1}{c}{\textsc{Discrepancy}}
				& \multicolumn{1}{c}{\textsc{Time}}

				\\ \midrule




				\multicolumn{1}{c|}{	\begin{tabular}{@{}c@{}} geometric \\ induced by $\B_d$  \end{tabular}}
				&    $\round{12d+o(d) } \sqrt{n^{1{-}\sfrac{1}{d}} \ln m }$
				& \multicolumn{1}{c|}{
					\begin{tabular}{@{}c@{}}
						$\bm{\aO\round{dn^{2+\sfrac{1}{d}}}}$
						\\
						(\Cref{cor:our-bound-balls})
					\end{tabular}
				}
				&          $O\round{\sqrt{n^{1{-}\sfrac{1}{d}}\ln m}}$
				& \multicolumn{1}{c}{
					\begin{tabular}{@{}c@{}}
						$O\round{n^{4+ \sfrac{1}{d}}}$
						\\
						\cite{MWW93}
					\end{tabular}
				}

				\\  \midrule

				\multicolumn{1}{c|}{
					\begin{tabular}{@{}c@{}} geometric \\ induced by  $\Gamma_{d,\Delta,s}$  \end{tabular}  }
				&$\round{24\sqrt{\Delta s} + o(1)} \sqrt{n^{1{-}\sfrac{1}{d}}\ln m}$
				& \multicolumn{1}{c|}{ \begin{tabular}{@{}c@{}} $\bm{\aO\round{s\Delta^d\round{mn^{\sfrac{1}{d}}{+}n^{2{+}\sfrac{1}{d}}}}}$ \\ {(\Cref{cor:our-semialg-matching-bound}) } \end{tabular}  }

				& $O\round{\sqrt{10^ds\Delta n^{1{-}\sfrac{1}{d}}\ln m}}$
				& \multicolumn{1}{c}{\begin{tabular}{@{}c@{}} $O\round{n^{O(d^3)}}$ \\  \cite{AMS13}  \end{tabular}}

				\\

				\multicolumn{1}{c|}{}
				&
				&\multicolumn{1}{c|}{}
				& $O\round{\sqrt{\Delta s n^{1{-}\sfrac{1}{d}}\ln m}}$
				& \multicolumn{1}{c}{\begin{tabular}{@{}c@{}} $\aO(s\Delta^d mn^3)$ \\ { \cite{MWW93}  }  \end{tabular}}

				\\  \midrule

				\multicolumn{1}{c|}{
					\begin{tabular}{@{}c@{}} geometric \\ induced by $\cH_d$  \end{tabular}
				}
				& $\round{8d +o(d)}\sqrt{n^{1{-}\sfrac{1}{d}} \ln m }$
				& \multicolumn{1}{c|}{\begin{tabular}{@{}c@{}} $\aO\round{ dn^{2{+}\sfrac{1}{d}}}$ \\ {(\Cref{cor:halfspace})  }  \end{tabular} }
				&          $O\round{ \sqrt{n^{1{-}\sfrac{1}{d}}\ln m}}$
				& \multicolumn{1}{c}{\begin{tabular}{@{}c@{}} $\aO(n)$ \\ { \cite{Chan12} }  \end{tabular}}

				 \\
				 \bottomrule

			\end{tabular}%
		}
		\caption{Summary of guarantees for geometric set systems. We use the notation $\cH_d$ for half-spaces in $\RR^d$, $\B_d$ for balls in $\RR^d$, and $\Gamma_{d,\Delta,s}$ for semialgebraic sets in $\RR^d$ described by at most $s$ equations of degree at most $\Delta$.  }
		\label{table:disc-for-dualshat}
	\end{table}

\noindent
Formally, given a set $X$ of $n$ points and a collection $\C$ of geometric objects in $\RR^d$,  we say that a set $Y \subset X$ is \emph{induced} by $\C$ if $Y = X \cap C$ 
for some $ C \in \C$. 
We say that a set system $(X, \S)$ is induced by $\C$ if each range in $\S$ is induced by $\C$.

\subsection{Semialgebraic set systems.}
Let $\Gamma_{d,\Delta,s}$ denote the collection of semialgebraic sets in $\RR^d$ that can be defined  as the solution set of a Boolean combination of at most $s$ polynomial inequalities of degree at most $\Delta$.
First, we give a bound on the VC-dimension and dual shatter function of set systems induced by $\Gamma_{d,\Delta,s}$.

\begin{lemma}\label{lemma:dual-shat-fn-bound}
    Let $X$ be a set of points in $\RR^d$ and $(X,\S)$ be a set system induced by $\Gamma_{d,\Delta,s}$. Then $\vcdim(X,\S) \leq 2s \log(es) \binom{\Delta + d }{ d} $ and the dual shatter function of $(X,\S)$ can be upper-bounded as $\pi^*_{\S}(k) \leq (4e\Delta s)^d\cdot k ^d$.
\end{lemma}
\begin{proof}
The bound on the VC-dimension can be deduced from Propositions 10.3.2 and Proposition 10.3.3 in \cite{matousek2013lectures}.
To bound $\pi^*_\S(k)$, let  $\R \subseteq \Gamma_{d,\Delta,s}$ be a set of $k$ ranges with defining  polynomials $\P = \{p_{ij} ~:~ 1 \leq i \leq k, 1 \leq j \leq s\}$, where each $p_{ij}$ is  a $d$-variate polynomial  of degree at most $\Delta$. Observe that  if  $\mathrm{sign}\bracket{p(x)} = \mathrm{sign} \bracket{p(y)}$ for all $p \in \P$, then $x,y$ are equivalent with respect to $\R$. Therefore, $\pi^*_{\Gamma_{d,\Delta,s}}(k)$ can be upper-bounded by the number of different sign patterns in $\{-1,1\}^{ks}$ induced by $ks$ $d$-variate polynomials of degree at most $\Delta$.
This quantity is bounded by $(4e\Delta s)^d\cdot k ^d $, see \cite[Theorem 3]{warren1968lower}.
\end{proof}

\noindent
By \Cref{lemma:dual-shat-fn-bound}, we get that set systems induced by $\Gamma_{d,\Delta,s}$ satisfy \nameref{assumption} with parameters $a=\frac{4e\Delta s }{\ln 2(1-1/d)}$, $b = \frac{\ln |\S|}{\ln 2}$, and $\gamma = 1-1/d$. Furthermore,  any set system induced by $\Gamma_{d,\Delta,s}$ has a membership Oracle of time complexity $O\round{s\binom{\Delta+d }{ d}}$. Thus, we can apply Theorems \ref{thm:main-disc-result}, \ref{thm:main-matching-result}, and \Cref{thm::main-apx-result} and obtain the following.

\begin{corollary}\label{cor:our-semialg-matching-bound}
	Let $X$ be a set of $n$ points in $\RR^d$ and $(X,\S)$ be a set system with $m$ ranges induced by $\Gamma_{d,\Delta,s}$.    Then
	\begin{enumerate}[i)]
	\item \nameref{discalgo}$\big((X,\S),~\frac{4e\Delta s }{\ln 2(1-1/d)} ,~ \frac{\ln m}{\ln 2} ,~1-\frac{1}{d}\big)$ constructs a coloring $\chi$ of $X$ of with expected discrepancy at most
	   \[
	   3\sqrt{ \frac{ 4e\Delta s  \ln m}{\ln 2(1-1/d)^2} \cdot n^{1-1/d}  + 19 \ln^2 m \ln n}
	   \]
	   in expected time $O \round{s\binom{\Delta + d }{ d} \round{mn^{1/d}\ln(mn) \ln n + n^{2+1/d}\ln n}}$.
	\item \nameref{mainalgo}$\big((X,\S),~\frac{4e\Delta s }{\ln 2(1-1/d)} ,~ \frac{\ln m}{\ln 2} ,~1-\frac{1}{d}\big)$ returns a perfect matching of $X$ with expected crossing number at most
	\[
		\frac{12e s \Delta}{(1-1/d)^2 \ln 2}  \cdot n^{1-1/d} + O\big(\ln m\ln n \big)
	\]
	in expected time $O \round{s\binom{\Delta + d }{ d} \round{mn^{1/d}\ln(mn) \ln n + n^{2+1/d}\ln n}}$.
   \item if $\eps \in (0,1)$, $\dvc := \vcdim(X,\S)$, and $A_0$ is a uniform random sample of $X$ of size $\frac{4\capx \dvc}{\eps^2} $, then \nameref{apxalgo}$\big((A_0,\S|_{A_0}),~\frac{4e\Delta s }{\ln 2(1-1/d)} ,~ \frac{\ln |\S|_{A_0}|}{\ln 2} ,~1-\frac{1}{d},~ \eps/2\big)$ returns a set $A \subset X$ of size 
     \begin{align*}
     O\round{ 
     \max\curly{
     	\round{\Delta s \cdot \frac{\dvc}{\eps^2} \ln\frac{1}{\eps}}^\frac{d}{d+1}, \frac{\dvc  }{\eps}  \ln^{3/2} \round{\frac{\dvc}{\eps}}
     }}
     \end{align*}
     with expected approximation guarantee satisfying
    	$\EE[\eps(A,X,\S)] \leq \eps$,
    and in expected  time
    $
    	O\round{ n + s\binom{\Delta + d}{d}\bracket{\round{\frac \dvc {\eps^2}}^{2 +1/d}\ln \frac \dvc {\eps^2} +  \round{\frac \dvc {\eps^2}}^{\dvc + 1/d} \ln\round{\frac \dvc {\eps^2}}^{\dvc + 1}  \ln^2 \frac \dvc {\eps^2}} }.
    $
\end{enumerate}
\end{corollary}
\begin{remark*}
The previous best algorithm for constructing matchings with low crossing numbers with respect to $\Gamma_{d,\Delta,s}$ relies on the polynomial partitioning technique \citep{AMS13}. It computes a perfect matching of $n$ points in general position with crossing number $O(10^d s \Delta n^{1-1/d})$ with respect to \emph{any} set in $\Gamma_{d,\Delta,s}$ in time $O(n^{O(d^3)})$, notably the running time is independent of $m$.
Our algorithm provides improved running time bounds for specific instances with $m=n^{o(d^3)}$.
\end{remark*}

\subsection{Balls and half-spaces.}
 Let $\cH_d$ and $\B_d$ denote the set of all half-spaces and balls in $\RR^d$ respectively.
Half-spaces and balls are semialgebraic sets, in particular, $\cH_d = \Gamma_{d,1,1}$ and $\B_d \subset \Gamma_{d,2,1}$.
What distinguishes their case from the general one is the existence of \emph{test-sets}. Test-sets are small-sized subfamilies (of half-spaces and balls) such that if a matching has low crossing number with respect to this subfamily, then it is guaranteed to have low crossing number with respect to any member of the family (of half-spaces and balls respectively): 


\begin{lemma}[Test-set lemma~\citep{Mat92}]\label{lemma:test-set}
	Let $X$ be a set of $n$ points in $\RR^d$ and  $t$ be a parameter. There exists a set $\T(t)$ of at most $(d+1)  t^d$  hyperplanes such that  if a perfect matching of $X$ has crossing number $\kappa$ with respect to $\T(t)$, then its crossing number with respect to any half-space in $\RR^d$ is at most $(d+1)\kappa + \frac{6d^2n}{t}$.
\end{lemma}
We will use \Cref{lemma:test-set} as black-box to obtain a test-set lemma for balls. It is well known that there are  mappings $\alpha : \RR^d \to \RR^{d+1}$ and $\beta: \B_d \to \cH_{d+1}$ such that for any $p \in \RR^d$ and $B \in \B_d$, we have $p \in B$ if and only if $\alpha(p) \in \beta (B)$, see e.g. \cite[Chap. 10]{matousek2013lectures}. This mapping  and \Cref{lemma:test-set} applied in $\RR^{d+1}$ with $t = n^{1/d}$ give the following lemma.

\begin{lemma}\label{lemma:test-set-balls}
	Let $X$ be a set of $n$ points in $\RR^d$. There exists a set $\Q$ of at most $(d+2)n^{1+1/d}$ balls such that  if a perfect matching of $X$ has crossing number $\kappa$ with respect to $\Q$, then its crossing number with respect to any ball in $\RR^d$ is at most $(d+2)\kappa + 6(d+1)^2n^{1-1/d}$.
\end{lemma}

In contrast to previous setups (where we required to have a finite set of $m$ ranges as an input), 
test-sets allow us to efficiently construct matchings with low crossing number with respect to \emph{any} half-space or ball in $\RR^d$. 
For half-spaces, we have a membership Oracle of time complexity $O(d)$, thus \Cref{cor:our-semialg-matching-bound} and \Cref{lemma:test-set} imply the following.

\begin{corollary}\label{cor:halfspace-matching}
	Let $X$ be a set of $n$ points in $\RR^d$ and $\T =  \T(n^{1/d})$ be the set of half-spaces provided by \Cref{lemma:test-set}. Then
	\nameref{mainalgo}$\big( (X,\T),~ \frac{4e}{(1-1/d)\ln 2},~ \frac{\ln \round{(d+1)n}}{\ln 2}, 1-\frac{1}{d}\big)$ returns a perfect matching of $X$ with expected crossing number at most
	\[
		\round{6d^2 + \frac{12e (d+1)}{ (1-1/d)^2\ln 2} }n^{1-1/d}   + O \big(\ln(dn)\ln n\big)
	\]
	 with respect to half-spaces in $\RR^d$,
	 in expected  time  $O \round{d  n^{2+1/d}\ln n }$.
\end{corollary}

Similarly, we can apply \Cref{cor:our-semialg-matching-bound} and \Cref{lemma:test-set-balls} to set systems induced by balls. Note that in case of balls, the Oracle complexity can be improved to $O(d)$ from the $O(d^2)$ bound used in  \Cref{cor:our-semialg-matching-bound} for $\Gamma_{d,2,1}$.

\begin{corollary}\label{cor:ball-matching}
	Let $X$ be a set of $n$ points in $\RR^d$ and let $\Q$ be the set of balls provided by \Cref{lemma:test-set-balls}. Then
	\nameref{mainalgo}$\big( (X,\Q),~ \frac{8e}{(1-1/d)\ln 2},~ \frac{\ln \round{(d+2)n^{1+1/d}}}{\ln 2}, 1-\frac{1}{d}\big)$ returns a matching $\curly{e_1, \dots, e_{n/2}}$ with expected crossing number at most
	\[
		\round{6(d+1)^2 + \frac{24e (d+2)}{(1-1/d)^2 \ln 2}  } n^{1-1/d} + O \big(\ln(dn)\ln n\big)
	\]
	 with respect to balls in $\RR^d$,
	 in expected  time  $O \round{d  n^{2+1/d}\ln n }$.
\end{corollary}
\begin{remark*}
The previous-best algorithm to construct matchings with crossing number $O(n^{1-1/d})$ with respect to $\B_d$ had time complexity $\aO(mn^3)$ \cite{MWW93}, which combined with \Cref{lemma:test-set-balls} yields an $\aO\round{n^{4 + \sfrac{1}{d}}}$ time algorithm. Alternatively, one can obtain a matching with sub-optimal crossing number $O\round{n^{1- 1/(d+1)}}$ by lifting $X$ into $\RR^{d+1}$, where the image of each range in $\B_d$ can be represented by a range in $\cH_{d+1}$ and applying the algorithm of \cite{Chan12} with time complexity $\aO(n)$.
\end{remark*}

We can show that test-sets can also be used as an input the algorithms \nameref{discalgo} and \nameref{apxalgo} using \Cref{lemma:matching-to-disc}:

\matchtodiscLemma*

Notice that the algorithm \nameref{discalgo} creates a coloring from the output of \nameref{mainalgo} precisely as it is defined in \Cref{lemma:matching-to-disc}. Therefore, a matching $M$ returned by \nameref{mainalgo} on a test-set (with low expected crossing number with respect to $\cH_d$) can be used to construct a coloring $\chi_M$ with low expected discrepancy with respect to $\cH_d$. Similarly, $\chi_M$ can be used to construct a small-sized $\eps$-approximation of $(X, \cH_d)$. 
These observations lead to the last two corollaries of this section.

\begin{corollary}\label{cor:halfspace}
	Let $X$ be a set of $n$ points in $\RR^d$ and $\T =  \T(n^{1/d})$ be the set of $(d+1)n$ half-spaces provided by \Cref{lemma:test-set}. Then
	\begin{enumerate}[i)]

	\item \nameref{discalgo}$\big((X,\T),~\frac{4e }{\ln 2(1-1/d)} ,~ \frac{\ln \round{(d+1)n}}{\ln 2} ,~1-\frac{1}{d}\big)$ constructs a coloring $\chi$ of $X$ of with expected discrepancy at most
	   \[
	   3\sqrt{ \round{6d^2 + \frac{12e (d+1)}{ (1-1/d)^2\ln 2} }n^{1-1/d} \ln m  + O \big( \ln(dn)\ln n\ln m \big) }
	   \]
	  with respect to half-spaces in $\RR^d$,
	  in expected  time  $O \round{d n^{2+1/d}\ln n}$.

   \item if $\eps \in (0,1)$ and  $A_0$ is a uniform random sample of $X$ of size $\frac{4\capx (d+1)}{\eps^2} $, then
   \nameref{apxalgo}$\big((A_0,\T|_{A_0}),~\frac{4e }{\ln 2(1-1/d)} ,~ \frac{\ln |\T|_{A_0}|}{\ln 2} ,~1-\frac{1}{d},~ \eps\big)$ returns a set $A \subset X$ of size 
     \begin{align*}
     O\round{ \max\curly{
     	\round{\frac{d  }{\eps^2}\ln\frac 1 \eps }^{\frac{d}{d+1}},
     	\frac{d}{\eps}\ln^{3/2} \round{\frac d \eps}
     	}}
     \end{align*}
     with expected approximation guarantee satisfying
    	$\EE[\eps(A,X,\cH_d)] \leq \eps$,
    and in expected  time
    $
    	O \round{ n + d\round{\frac d {\eps^2}}^{2+1/d}\ln \frac d \eps }.
    $
\end{enumerate}
\end{corollary}



\begin{corollary}\label{cor:our-bound-balls}
	Let $X$ be a set of $n$ points in $\RR^d$ and let $\Q$ be the set of balls provided by \Cref{lemma:test-set-balls}. Then
	\begin{enumerate}[i)]

	\item \nameref{discalgo}$\big((X,\Q),~\frac{8e }{\ln 2(1-1/d)} ,~ \frac{\ln \round{(d+2)n^{1+1/d}}}{\ln 2} ,~1-\frac{1}{d}\big)$ constructs a coloring $\chi$ of $X$ of with expected discrepancy at most
	   \[
	   3\sqrt{ \round{6(d+1)^2 + \frac{24e (d+2)}{(1-1/d)^2 \ln 2}  } n^{1-1/d} \ln m  + O \big( \ln(dn)\ln n\ln m \big)}
	   \]
	  with respect to balls in $\RR^d$, in
	   expected  time  $O \round{d  n^{2+1/d}\ln n }$.

   \item if $\eps \in (0,1)$,  $A_0$ is a uniform random sample of $X$ of size $\frac{4\capx (d+2)}{\eps^2} $, then the algorithm
   \nameref{apxalgo}$\big((A_0,\Q|_{A_0}),~\frac{8e }{\ln 2(1-1/d)} ,~ \frac{\ln |\Q|_{A_0}|}{\ln 2} ,~1-\frac{1}{d},~ \eps \big)$ returns a set $A \subset X$ of size at most
     \begin{align*}
     O\round{ \max\curly{
     	\round{\frac{d }{\eps^2}\ln \frac 1 \eps }^{\frac{d}{d+1}},
     	\frac{  d}{\eps}\ln^{3/2}\round{ \frac d \eps}
     	}}
     \end{align*}
     with expected approximation guarantee satisfying
    	$\EE[\eps(A,X,\B_d)] \leq \eps$,
    and in expected  time
    $
    	O \round{ n + d\round{\frac d {\eps^2}}^{2+1/d}\ln\frac d \eps  }.
    $
\end{enumerate}
\end{corollary}

\bibliographystyle{plainnat}
\bibliography{spanningtreealgo.bib}

\section{Appendix}

\subsection{Proof of \Cref{lemma:matching-to-disc}}\label{sec:Appendix-match-to-disc}
	Let $ S \in \S$ be a fixed range. We express the sum $\chi_M(S)$ of colors over elements of $S$ as 
	\[
		\chi_M(S) 
		= \sum\limits_{\{x,y\} \in M; x,y \in S} \round{\chi_M(x) + \chi_M(y)} + \sum\limits_{x \in \mathrm{cr}(S,M)} \chi_M(x) = \sum\limits_{x \in \mathrm{cr}(S,M)} \chi_M(x)~,
	\]
	where $\mathrm{cr}(S,M) = \{x\in S  : \{x,y\} \in M, y \notin S\}$.
	Since $\mathrm{cr}(S,M)\leq \kappa$ for any $S \in \S$, $\disc(S,\chi_M)$ is a sum of at most $\kappa$ \emph{independent} random variables. 
	We use the following concentration bound from \cite{alonspencerprobabilistic}
	\begin{claim}[Theorem A.1.1 from \cite{alonspencerprobabilistic}]\label{claim:Alon-concentration}
		Let $X_1, \dots, X_k$ be independent $\{-1,1\}$-valued random variables with $\PP[X_i = -1] = \PP[X_i = 1] = 1/2$. Then for any $\alpha \geq 0$
			\[
				\PP\bracket{\abs{\sum_{i=1}^k X_i} > \alpha} \leq 2e^{-{\alpha^2}/{2k}}.
			\]
	\end{claim}
	\noindent
	Applying \Cref{claim:Alon-concentration}, we get that for any fixed $S \in \S$ and $\alpha > 0$,
		\[
			\PP\bracket{\abs{\chi_M(S)}  > \alpha}  \leq 2e^{-{\alpha^2}/{2\kappa}}.
		\]
	By the union bound, we get
		\[
			\PP \bracket{ \disc_\S(\chi_M) > \alpha } = \PP \bracket{\max\limits_{S \in \S} \abs{\chi_M(S)}> \alpha } \leq m\cdot 2e^{-{\alpha^2}/{2\kappa}}.
		\]
	Finally, we bound the expected discrepancy by applying Fubini's theorem
		\begin{align*}
			\EE\bracket{ \disc_\S(\chi_M) }
			&\overset{\text{def}}{=}
			\int\limits_{0}^\infty \PP \bracket{ \disc_\S(\chi_M) > \alpha } d\alpha
			\leq
			\int\limits_{0}^\infty \min\left\{2m\cdot e^{-{\alpha^2}/{2\kappa}}, ~ 1 \right\} d\alpha\\
			&=
			\int\limits_{0}^{\sqrt{2\kappa\ln(2m)}} 1 d\alpha
			+
			\int\limits_{\sqrt{2\kappa\ln(2m)}}^\infty 2m\cdot e^{-{\alpha^2}/{2\kappa}} d\alpha\\
			&=
			\sqrt{2\kappa\ln(2m)}
			+
			2m\sqrt{2\kappa}\int\limits_{\sqrt{\ln(2m)}}^\infty  e^{-t^2}  ~dt\\
			&=
			\sqrt{2\kappa\ln(2m)}
			+
			2m\sqrt{2\kappa}\int\limits_{\sqrt{\ln(2m)}}^\infty  \frac{t}{t}\cdot e^{-t^2}  ~dt\\
			&\leq
			\sqrt{2\kappa\ln(2m)}
			+
			2m\sqrt{\frac{2\kappa}{\ln(2m)}}\int\limits_{\sqrt{\ln(2m)}}^\infty  te^{-t^2}  ~dt\\
			&=
			\sqrt{2\kappa\ln(2m)}
			+
			2m\sqrt{\frac{2\kappa}{\ln(2m)}}\bracket{  -\frac{e^{-t^2}}{2}}_{\sqrt{\ln(2m)}}^\infty \\
			&=
			\sqrt{2\kappa\ln(2m)}
			+
			\sqrt{\frac{\kappa}{2\ln(2m)}}
			\leq
			\sqrt{3\kappa\ln m},
		\end{align*}
		if $m \geq 34$. This concludes the proof of \Cref{lemma:matching-to-disc}.
\qed

\subsection{Proof of \Cref{thm::main-apx-result}}\label{sec:main-apx-proof}

In this section, we show how \Cref{thm::main-apx-result} is implied by \Cref{thm:main-disc-result}.
Recall that the algorithm \nameref{apxalgo} constructs a sequence of sets $A_0, A_1, A_2, \dots, A_j \subseteq X$ iteratively. In particular, it sets $A_0 = X$ and for $i = 1, \dots, j$, $A_i \subseteq A_{i-1}$ is defined as $\chi_i^{-1}(1)$, where $\chi_i: A_i \to \{-1,1\}$ is the coloring provided by \nameref{discalgo}$\round{(A_i,\S|_{A_i}), a,b,\gamma}$. Note that $ |A_{i+1}| = \ceil{|A_i|/2} = \ceil{n/2^{i+1}}$.
We bound the approximation guarantee using the following lemma.

\disctoapx*

\noindent
By \Cref{thm:main-disc-result},
	\[
	  \EE\bracket{\disc_{\S|_{A_i}} \round{\chi_i}} \leq  3\sqrt{ \frac{a}{\gamma}  |A_i|^\gamma\ln |\S|_{A_i}| + \round{\frac{b}{2} + 12\ln |\S|_{A_i}|}\log |A_i| \cdot \ln |\S|_{A_i}| } .
	\]
Thus for $ i = 0, \dots, j-1$, by \Cref{lemma:disc-to-apx},
\begin{equation}\label{eq:apx-bound-once}
	\EE\bracket{\eps(A_{i+1}, A_i, \S|_{A_i})} \leq \frac{6}{\ceil{n/2^i}} \cdot \sqrt{\frac{a}{\gamma}  \cdot \ceil{n/2^i}^{\gamma}\ln m + \round{\frac{b}{2} + 12\ln m}\log \ceil{n/2^i} \ln m}.
\end{equation}
Recall that if $A_1$ is an $\eps_1$-approximation of $(X,\S)$ and $A_2$ is an $\eps_2$-approximation of $(A_1, \S|_{A_1})$, then $A_2$ is an $(\eps_1 + \eps_2)$-approximation of $(X,\S)$.
Therefore,
\[
	\eps(A_{j},X,\S) \leq \eps(A_{j}, A_{j-1}, \S|_{A_{j-1}}) + \eps(A_{j-1}, A_{j-2}, \S|_{A_{j-2}}) + \dots + \eps(A_2,A_1, \S|_{A_1}) + \eps(A_1,X, \S),
\]
which by linearity of expectation and \Cref{eq:apx-bound-once} yield
\begin{align*}
	\EE\bracket{\eps(A_{j},X,\S) }
	&\leq
	\sum\limits_{i=0}^{j-1} \frac{6}{\ceil{n/2^i}} \cdot \sqrt{\frac{ a}{\gamma}  \cdot \ceil{n/2^i}^{\gamma}\ln m + \round{\frac{b}{2} + 12\ln m}\log \ceil{n/2^i} \ln m}\\
	&\leq
	\frac{6}{n^{1-\gamma/2}} \sqrt{ \frac{a\ln m}{\gamma}} \cdot \sum_{i=0}^{j-1} \round{2^{1-\gamma/2}}^i  + \frac{6}{n} \sqrt{ \frac{b\log n \ln m}{2} + 12\log n \ln^2 m} \cdot \sum_{i=0}^{j-1} 2^i \\
	&\leq
	\frac{15}{n^{1-\gamma/2}} \sqrt{ \frac{a\ln m}{\gamma}} \cdot \round{2^{1-\gamma/2}}^j
	+ \frac{6\cdot (2^j-1)}{n}\sqrt{ \frac{b\log n \ln m}{2} + 12\log n \ln^2 m} \\
	&\leq
	15 \sqrt{ \frac{a\ln m}{\gamma}}  \cdot \round{\frac{2^j}{n}}^{1-\gamma/2}
	+ \frac{6\cdot 2^j}{n}\sqrt{ \frac{b\log n \ln m}{2} + 12\log n \ln^2 m}.
\end{align*}
Substituting
\[
	j = \floor{\log n +\min\curly{ \frac{2}{2-\gamma} \log \frac{\eps\sqrt{\gamma}}{30\sqrt{a\ln m}},   \log \frac{\eps}{12 \sqrt{\round{\frac b 2 + 12\ln m }\ln(m)\log n }}} },
\]
we get that $\EE\bracket{\eps(A_{j},X,\S) } \leq \eps$ and
\[	
	|A_j| = \frac{n}{2^j} \leq 2 \max \curly{ \round{\frac{30\sqrt{a \ln m }}{\eps \sqrt{\gamma}} }^{\frac{2}{2-\gamma}}, \frac{12 \sqrt{\round{\frac b 2 + 12\ln m }\ln(m)\log n }}{\eps} } .
\]
By \Cref{thm:main-disc-result}, constructing a the coloring $\chi_i$ requires at most
\[
	\min\curly{ \frac{ 24 {|A_i|}^{3-\gamma} \ln|A_i|}{ a}  + \frac{18 m{|A_i|}^{1-\gamma} \ln\round{m|A_i|}}{a} \min\curly{\frac{2}{1-\gamma}, \log |A_i|},~ \frac 1 {7} {|A_i|}^3 + \frac{m|A_i|}{2} }
\]
calls to the membership Oracle, in expectation. Since $|A_i| = \ceil{n/2^i}$, the expected number of membership Oracle  calls that \nameref{apxalgo}$\big((X,\S),  a, b, \gamma, j \big)$ performs is at most
\begin{align*}
&
\sum_{i = 0}^{j} \min\curly{ \frac{ 24 \round{\frac{n}{2^i}}^{3-\gamma} \ln\frac{n}{2^i}}{ a}  + \frac{18 m\round{\frac{n}{2^i}}^{1-\gamma} \ln\frac{mn}{2^i}}{a} \min\curly{\frac{2}{1-\gamma}, \log \frac{n}{2^i}},~ \frac 1 {7 }\round{\frac{n}{2^i}}^3 + {\frac{mn}{2^{i+1}}} }\\
&\leq
 \min\curly{ \sum_{i = 0}^j \round{ \frac{ 24 \round{\frac{n}{2^i}}^{3-\gamma} \ln\frac{n}{2^i}}{ a}  + \frac{18 m\round{\frac{n}{2^i}}^{1-\gamma} \ln\frac{mn}{2^i}}{a} \min\curly{\frac{2}{1-\gamma}, \log \frac{n}{2^i}}},~ \sum_{i = 0}^j \round{ \frac{n^3}{7 \cdot 2^{3i}} + {\frac{mn}{2^{i+1}}}}}\\
 &\leq
 \min\curly{ \frac{32n^{3-\gamma}\ln n}{a} + \frac{18 mn^{1-\gamma} \ln(mn)}{a} \round{ \min\curly{\frac{2}{1-\gamma}, \log n}}^2,~ \frac{8n^3}{49}  + mn  }.
\end{align*}
This concludes the proof of \Cref{thm::main-apx-result}.\qed

	\end{document}